\documentclass[lettersize,journal]{IEEEtran}
\usepackage{amsthm}
\usepackage{amsmath,amssymb,amsfonts}
\usepackage{textcomp}
\usepackage[x11names,dvipsnames]{xcolor}

\usepackage{booktabs}
\usepackage{subcaption} 
\usepackage{diagbox}
\usepackage{float}
\usepackage{makecell} 
\usepackage{array}
\usepackage{svg}
\usepackage{epstopdf}
\usepackage{authblk}
\usepackage{hyperref}
\hypersetup{
    colorlinks=true,
    linkcolor=blue,
    citecolor=blue,
    urlcolor=blue
}
\usepackage[capitalize,noabbrev]{cleveref}
\usepackage{algorithm, algpseudocode}
\usepackage{array}
\usepackage{url}
\newcommand{\orcidlink}[1]{%
  \textsuperscript{\href{https://orcid.org/#1}{\includegraphics[width=1.6ex]{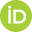}}}%
}
\theoremstyle{plain}
\newtheorem{theorem}{Theorem}[section]      
\newtheorem{lemma}[theorem]{Lemma}
\newtheorem{corollary}[theorem]{Corollary}
\newtheorem{proposition}[theorem]{Proposition}
\theoremstyle{definition}
\newtheorem{definition}[theorem]{Definition}

\theoremstyle{remark}
\newtheorem*{remark}{Remark}

\newcommand{\parhead}[1]{\vspace{3pt plus 1pt minus 1pt}\par\noindent\textbf{#1}\hspace{.4em plus .2em minus .2em}}
\renewcommand{\paragraph}[1]{\parhead{#1}}

\begin{document}

\markboth{\emph{Conformal-DP}: A Density-Aware Mechanism for Differential Privacy over Riemannian Manifolds via Conformal Transformation}{He. {et al.}}

\title{\emph{Conformal-DP}: A Density-Aware Mechanism for Differential Privacy over Riemannian Manifolds via Conformal Transformation}

\author{Peilin~He\,\orcidlink{0000-0003-3553-9949},~\IEEEmembership{Member,~IEEE,}
        Liou~Tang\,\orcidlink{0009-0005-5220-5176},~\IEEEmembership{Member,~IEEE,}
        M.~Amin~Rahimian\,\orcidlink{0000-0001-9384-1041},
        and~James~Joshi\,\orcidlink{0000-0003-4519-9802},~\IEEEmembership{Fellow,~IEEE}%
\IEEEcompsocitemizethanks{%
\IEEEcompsocthanksitem P. He, L. Tang, and J. Joshi are with Department of Informatics and Networked Systems, University of Pittsburgh, Pittsburgh, PA 15213 USA. E-mail: \{peilin.he, liou.tang, jjoshi\}@pitt.edu. Corresponding author: Peilin He
\IEEEcompsocthanksitem M. A. Rahimian is with Department of Industrial Engineering, University of Pittsburgh, Pittsburgh, PA 15260 USA. E-mail: rahimian@pitt.edu.
}%
}


\maketitle
\begin{abstract}
Differential Privacy (DP) is being increasingly adopted for non-Euclidean data that lie on complex, high-dimensional manifolds. Existing DP mechanisms for manifold data consider geometric properties when calibrating privacy perturbations, but they largely fail to capture variations in data density within datasets, leading to biased perturbations and suboptimal privacy-utility trade-offs due to heterogeneous data distributions. In this paper, we propose a novel density-aware differential privacy mechanism on Riemannian manifolds, referred to as \textit{Conformal-DP}, that leverages conformal transformations to calibrate perturbations based on local densities and to induce a density-balanced geometry. We prove that our mechanism satisfies $\varepsilon$-differential privacy on any complete Riemannian manifold under mild regularity assumptions. In addition, we derive a closed-form expected geodesic error bound that depends only on the underlying data density ratio and is independent of global curvature. Our empirical results on synthetic and real-world datasets demonstrate that the proposed \textit{Conformal-DP} mechanism substantially improves the privacy–utility trade-off in heterogeneous data distribution settings, with worst-case performance comparable to state-of-the-art manifold DP mechanisms that assume uniformly distributed data.
\end{abstract}


\section{Introduction}\label{sec:intro}
Many real-world datasets are not well modeled as uniformly sampled points from the ambient input space. Instead, they often concentrate around structured, lower-dimensional regions induced by latent generative factors, a premise commonly formalized by the manifold hypothesis and widely used in manifold learning and representation learning  \cite{narayanan2010sample, fefferman2016testing, bengio2013representation}. This assumption is especially relevant for high-dimensional domains such as natural language, where syntactic, semantic, and pragmatic constraints restrict observed samples to a highly structured subset of the combinatorial sequence space.  Recent studies of contextual embeddings and language model representations further highlight pronounced geometric structure, anisotropy, and low-intrinsic-dimensional behavior  \cite{ethayarajh2019contextual, coenen2019visualizing, lee2025shared, modell2025origins}. Thus, when studying data used for LLM or agentic-AI training, it is more appropriate to model the distribution as structured and non-uniform rather than as uniform over the underlying space. These datasets usually exist in a non-Euclidean space. These include, but are not limited to, medical imaging data (e.g., MRI or CT scans) \cite{pennec2019riemannian, dryden2005statistical, dryden2009non}, terrain elevation models, and climate patterns on a spherical Earth \cite{belkin2006manifold, niyogi2013manifold}. In these cases, traditional linear algebraic methods are insufficient for fully capturing the properties of the data. Additionally, much of real-world data usually contains privacy-sensitive information, which presents challenges when data needs to be published or used for model training and inference in various application domains, such as the computer vision, and AI-enabled healthcare and climate modeling applications \cite{cheng2013novel, turaga2008statistical, turaga2016riemannian}. These datasets demand more advanced and effective modeling techniques, such as geometric or manifold-based methods, to facilitate their sharing and use in a privacy-preserving manner.

Differential Privacy (DP), first introduced by Dwork et al. in \cite{dwork2006differential}, is widely regarded as the gold standard for private data release because it provides rigorous mathematical guarantees that limit the risk of re-identifying individuals in a dataset. Several variations and refinements of DP have been proposed, including standard DPs \cite{dwork2006differential,dwork2010boosting,dwork2008differential}, Rényi DP \cite{mironov2017renyi}, concentrated DP \cite{dwork2016concentrated}, and random DP \cite{hall2011random}, to provide rigorous privacy guarantees while enabling privacy-preserving data release and use in analytics and AI model development. Most of the existing DP mechanisms focus on linear or Euclidean data and are not suitable for non-linear data \cite{tenenbaum2000global, roweis2000nonlinear, xie2016deepshape} in real-world settings. Reimherr et al. \cite{reimherr2021differential} extend Laplace-based DP to manifold-valued data. Their work shows that DP mechanisms can be constructed using only intrinsic geometric quantities, such as geodesic distances and manifold volumes, instead of extrinsic Euclidean approximations. Utpala et al. \cite{utpala2022differentially}, and Jiang et al. \cite{jiang2023gaussian} further extend Gaussian DP to manifolds and demonstrate that Gaussian perturbations can achieve better utility than Laplace-based perturbations at the expense of some accuracy, given the same setting. More recently, Jiang et al. \cite{jiang2026exponentialwrapped} proposed Exponential-Wrapped mechanisms for Hadamard manifolds, providing a practical non-MCMC approach to generating manifold-valued differentially private noise.

In existing works, manifold-valued data are often implicitly assumed to follow uniform distributions (e.g., Gaussian Distribution), which limits their applicability to real-world settings where a dataset has heterogeneous distributions. In practice, manifold-valued data often exhibit significant variation in local densities, with both dense and sparse regions. Such heterogeneity implies that different regions may require different levels of perturbation to appropriately balance their influence on downstream analysis tasks and privacy protection, as highlighted in prior studies \cite{zhao2017dependent, zhao2024scenario, dunning2012privacy, liang2024smooth}. The assumption of uniform data distribution limits the accuracy of existing approaches, since most of them rely on global sensitivity to calibrate noise, leading to uniform perturbation regardless of local data densities within subsets of data samples. In dense regions, the marginal contribution of each data point is naturally reduced due to redundancy, while in sparse regions the influence is stronger. Consequently, the worst-case noise calibration can over-perturb dense regions, resulting in avoidable, but increased, utility loss.

In this work, we address the following key research question: \textit{given the heterogeneously distributed manifold-valued data, how can a DP mechanism account for local density variations so that privacy guarantees are preserved without incurring unnecessary utility loss due to variations in data densities}? To answer this question, we propose a novel density-aware DP mechanism, namely ``Conformal-DP," on Riemannian manifolds. The proposed Conformal-DP mechanism
uses the conformal transformation on the manifold-valued data to change the shape and metric of the original manifold and contract the conformal manifold according to the data density. In this setting, the sensitivity in the conformal space is always less than or equal to the worst-case sensitivity on the original manifold because the upper bound of the added noise is fixed by the prior worst-case sensitivity and cannot exceed it. This is achieved in two stages: we first fit the conformal manifold to the data density distribution in a privacy preserving manner and then use the fitted manifold to compute the current conformal sensitivity and add global conformal Laplace noise in stage two. Finally, we map the conformal space back to the original space and adjust the metric accordingly, ensuring the noise is redistributed according to the data density. As a result, dense regions receive less noise due to reduced contraction, and sparse regions receive more noise, with the maximum bounded by the global worst-case noise, leading to better utility. Figure~\ref{fig:overview} illustrates the key idea behind the proposed Conformal-DP. We highlight our main contributions as follows:
\begin{itemize}
    \item We propose a density-aware Conformal-DP mechanism that leverages conformal transformation \cite{obata1970conformal, schoen1984conformal, aviles1988conformal, francesco2012conformal} as an intrinsic geometric tool to encode the data density distribution into the manifold metric and adaptively reshape geometric properties across the manifold, to guide privacy perturbation.
    \item We establish theoretical guarantees for the proposed mechanism by deriving bounds on privacy loss through bi-Lipschitz continuity between the original and conformally transformed geodesic distances and volumes.
    \item We conduct extensive experiments to demonstrate that our approach achieves superior privacy--utility trade-offs on manifold-valued datasets and provides a unifying framework that connects conventional DP mechanisms with those designed for manifolds.
\end{itemize}

The remainder of this paper is organized as follows. In Section \ref{sec:background}, we overview the foundational concepts related to DP, Riemannian geometry, and conformal transformations, and present the key notation used throughout the paper. In Section \ref{section3}, we present the construction of the differentially private conformal factor on $\mathcal{M}$ and introduce the associated conformal metrics that capture local variations in density. Building on this formulation, we present our Conformal-DP mechanism, followed by a theoretical analysis of its privacy–utility trade-offs and optimality under bounded curvature assumptions in Section \ref{section4}. Next, we present the algorithmic design and report experimental results on both synthetic and real-world datasets, compare our approach with existing methods, and demonstrate its practical effectiveness in Section \ref{section5}, Section \ref{sec: section6} and Section \ref{sec:experiment}. Finally, we conclude with a discussion of the applications, limitations, and directions for future work.

\begin{figure}[t]
  \centering
  \includegraphics[width=1\linewidth]{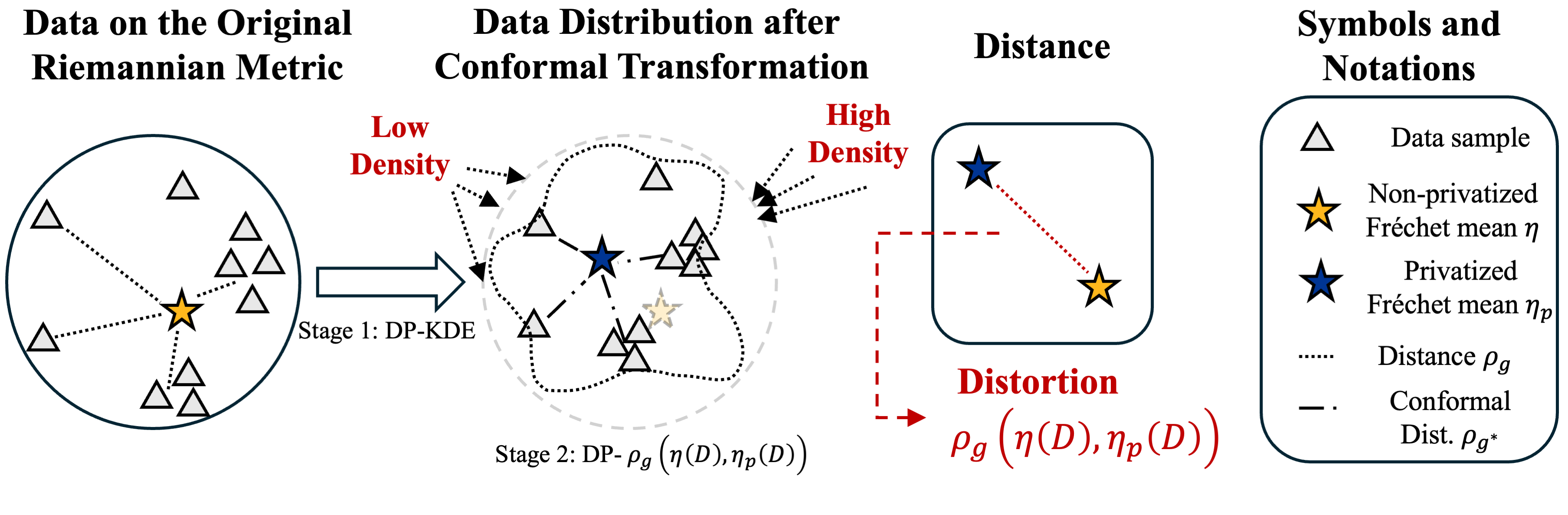}
  \caption{An overview of the Conformal-DP mechanism: we achieve desired $\epsilon$ level differential privacy in two stages: first, by perturbing the counts of data points in a randmly generated Voronoi tessellation of the manifold we can obtain a differentially private kernel density fit to the data, which we then use for the DP-Fr\'echet Mean release in stage 2. The conformal transformation in stage 2 changes the distributions of the data, and the Fr\'echet Mean will be distorted; it contracts the manifold as dense regions, resulting in less noise added, while sparse regions have more noise added.}
  \label{fig:overview}
  \vspace{-15pt}
\end{figure}

\section{Related Work and Background}
\label{sec:background}
In this section, we present related work, basic notations in \cref{grouped_symbols}, and an overview of key concepts related to DP, Riemannian manifolds, and conformal transformation. For more details, we refer the readers to Dwork et al. \cite{dwork2006differential, dwork2008differential, dwork2014algorithmic} for DP, and to Reimherr et al. and Jiang et al. \cite{reimherr2021differential, jiang2023gaussian} for extensions of DP over Riemannian manifolds. For more theoretical background on the Riemannian manifolds and the conformal space transformation, we refer the readers to \cite{francesco2012conformal, obata1970conformal, ginsparg1990applied, katsman2024riemannian}.

\paragraph{Differential Privacy over Manifolds.}
\label{para:manifold_dp}
DP provides a rigorous framework for protecting individual records while enabling statistical analysis of sensitive datasets~\cite{dwork2006differential}. In recent years, there has been growing interest in extending DP to settings where data naturally reside on non-Euclidean domains, particularly the Riemannian manifolds. Foundational work in statistics on manifolds by Fletcher et al.~\cite{fletcher2004principal}, Pennec et al.~\cite{pennec2006intrinsic}, and Dryden et al.~\cite{dryden2009non} establish the importance of respecting intrinsic geometry when performing inference on curved spaces, while earlier studies such as by Fisher et al.~\cite{fisher1995statistical} demonstrate the advantages of manifold representations for structured data. Building on these developments, several works have extended classical DP mechanisms to manifold-valued data~\cite{reimherr2021differential, utpala2022differentially, soto2022shape, jiang2023gaussian}. Reimherr et al. in ~\cite{reimherr2021differential} introduce Laplace and $k$-norm mechanisms using intrinsic geometric distances, while later work explores Gaussian and geometry-aware perturbation methods that respect manifold curvature and improve utility~\cite{utpala2022differentially, soto2022shape, jiang2023gaussian, dong2022gaussian}. These efforts highlight the promise of incorporating geometric structure into privacy mechanisms, enabling noise injection that preserves important manifold properties while maintaining formal privacy guarantees.

\paragraph{Riemannian Geometry.}
\label{para:rie_geo}
We briefly review the geometric concepts that are used throughout this work; standard references for these include \cite{greene1989review, lee2018introduction, lang2006introduction, srivastava2016functional, dryden2005statistical}. We consider a $d$-dimensional complete Riemannian manifold $\mathcal{M}$. For $m\in\mathcal{M}$, the tangent space is denoted by $T_m\mathcal{M}$, and the Riemannian metric is $g=\{\langle\cdot,\cdot\rangle_m: m\in\mathcal{M}\},$ which equips each $T_m\mathcal{M}$ with an inner product and induces the intrinsic notions of length, distance, and volume on $\mathcal{M}$.

\paragraph{Geodesic Distance.}
\label{para:rie_geodesic_distance}
For a smooth curve $\gamma:[0,1]\to\mathcal{M}$, its length is defined by$$L(\gamma)=\int_0^1 \|\dot{\gamma}(t)\|_{\gamma(t)}\,dt =\int_0^1\Big(\langle \dot{\gamma}(t),\dot{\gamma}(t)\rangle_{\gamma(t)}\Big)^{\frac{1}{2}}\,dt.$$The Riemannian distance between $m_1,m_2\in\mathcal{M}$ is$$\rho(m_1,m_2) :=\inf_{\substack{\gamma:\gamma(0)=m_1\\ \gamma(1)=m_2}} L(\gamma),$$and a curve attaining the infimum is a (minimizing) geodesic.

\paragraph{Exponential and Logarithm Maps.} 
\label{para:exponential}
The exponential map at $m$ is $\exp_m:T_m\mathcal{M}\to\mathcal{M}$. If a unique geodesic $\gamma$ from $m_1$ to $m_2$ exists, then $\exp_{m_1}(\dot{\gamma}(0))=m_2$ under the corresponding parameterization. The map $\exp_m$ is locally a diffeomorphism near $0\in T_m\mathcal{M}$; we denote its local inverse by $\log_m$ (the logarithm map). The largest radius on which $\log_m$ is well defined is the injectivity radius at $m$, and the global injectivity radius is $\mathrm{inj}(\mathcal{M})=\inf_{m\in\mathcal{M}}\mathrm{inj}(m)$.

\paragraph{Volume Measure.} 
\label{para:volume_measure}
The metric $g$ induces a canonical volume form and volume measure $\mu$ on $\mathcal{M}$. For a coordinate chart $\varphi:U\subset\mathcal{M}\to\mathbb{R}^d$ with metric components $g_{ij}=\langle\partial_i,\partial_j\rangle_m$, the local volume form is
\[
d \mu_g=\sqrt{\operatorname{det}\left(g_{i j}\right)} d x^1 \cdots d x^d
\]
which defines a global measure (up to orientation) via a partition of unity\footnote{A \emph{partition of unity} on $\mathcal{M}$ is a collection of smooth non-negative functions $\{\varphi_i\}$ summing to $1$ on $\mathcal{M}$, with each $\varphi_i$ supported in a coordinate chart; see \cite[Chapter~2]{lee2018introduction}.}. In our proposed framework, $\rho$, $\exp_m$, and $\mu$ provide the intrinsic tools for defining density-aware perturbations on $\mathcal{M}$.

\paragraph{Conformal Transformation over Riemannian Manifolds}
\label{section2.3}
A conformal transformation (or conformal diffeomorphism) is a smooth map $f:(\mathcal{M},g)\to(\mathcal{M^*},g^*)$ that preserves angles. Equivalently, it rescales the Riemannian metric by a strictly positive function $g^* = e^{2\sigma} g, \sigma:\mathcal{M}\to\mathbb{R},$ where $\sigma$ is \emph{conformal scaling function}~\cite{francesco2012conformal,obata1970conformal,ginsparg1990applied}. Here, $\mathcal{M}^*=\mathcal{M}$, the map $f$ is a conformal symmetry of $(\mathcal{M},g)$~\cite{del2020differentiable}. Metrics related by $g^*=e^{2\sigma}g$ are said to belong to the same conformal class~\cite{obata1970conformal}. The core intuition is that the inner product on each tangent space is scaled uniformly by $e^{2\sigma(m)}$. Because of this uniform scaling, angles between tangent vectors (and hence between smooth curves) remain completely unchanged~\cite{ottazzi2010liouville}. However, lengths and distances are modified by a position-dependent factor: infinitesimal lengths scale by $e^{\sigma(m)}$. If $e^{2\sigma}$ is constant, the transformation is homothetic and simply scales all lengths uniformly~\cite{weinberg2013gravitation}. In short, conformal transformations preserve local shapes while allowing distances and volumes to vary flexibly across the manifold. We adopt the conformal rescaling $g^*=e^{2\sigma}g$ for two practical reasons. First, in Section~\ref{section3}, we construct $\sigma$ via an elliptic PDE, which guarantees the existence, uniqueness, and smoothness of the conformal factor. Second, this multiplicative formulation provides explicit bi-Lipschitz bounds on geodesic distances. These bounds are crucial for managing volume sensitivities, ultimately driving the theoretical guarantees in our privacy--utility analysis for conformal DP in Section~\ref{section4}.

\paragraph{Topological Analysis.}
\label{para:topo_analy}
Let $\mathcal{M}$ be a compact and smooth Riemannian manifold of dimension $d$ and be equipped with: 1) Riemannian metric $g$; 2) Density-Aware conformal metric $g^* = \phi \cdot g$, where the conformal factor $\phi:\mathcal{M} \rightarrow R > {0}$. The Riemannian manifolds $(\mathcal{M},g)$ and $(\mathcal{M},g^*)$ share the same underlying smooth topological structure. Specifically, the identity map $\mathbf{id_{\mathcal{M}}} : (\mathcal{M},g) \rightarrow (\mathcal{M},g^*)$ is a diffeomorphism \cite{aviles1988conformal, obata1970conformal, schoen1984conformal}, ensuring that any point $z$ in the topological manifold $\mathcal{M}$ is identified with itself, independently of $g$ or $g^*$ for all $z \in \mathcal{M}$, so that the topology and smooth structure of $\mathcal{M}$ remain unchanged. The metrics $g$ and $g^*$ differ only in their geometric measurements (e.g., geodesic distances and volumes) on $\mathcal{M}$. For example, consider a stochastic output $\vec{m} \in \mathcal{M}$ calculated under the metric $g^*$. Because $\vec{m}$ is fundamentally a point in the topological manifold $\mathcal{M}$, no additional ``Inverse operation'' needs to be calculated to interpret $\vec{m}$ in $(\mathcal{M},g)$. The difference arises only in the way the geometric properties of $\vec{m}$ are quantified under $g$ versus $g^*$.

\begin{table}[t!]
\centering
\small
\setlength{\abovecaptionskip}{2pt}
\setlength{\belowcaptionskip}{2pt}
\setlength{\tabcolsep}{3pt}
\renewcommand{\arraystretch}{0.92}
\caption{Key notations and parameter-setting guidelines.}
\label{grouped_symbols}

\resizebox{\linewidth}{!}{
\begin{tabular}{ll}
\toprule
\multicolumn{2}{l}{\textbf{Original Riemannian Manifold Notations}} \\ \midrule
$\mathcal{M}$ & Compact, complete Riemannian manifold. \\
$d$ & Intrinsic dimension of $\mathcal{M}$. \\
$D=\{x_1,\ldots,x_N\}$ & Dataset of $N$ manifold-valued samples. \\
$D\sim D'$ & Adjacent datasets under substitution adjacency. \\
$g$ & Original Riemannian metric. \\
$\langle\cdot,\cdot\rangle_m$ & Riemannian inner product at $m$. \\
$\rho_g(x,y)$ & Geodesic distance under $g$. \\
$T_m\mathcal{M}$ & Tangent space at $m$. \\
$\exp_m,\log_m$ & Exponential and logarithm maps at $m$. \\
$\mu_g,d\mu_g$ & Riemannian volume measure and volume element. \\
$V=\mu_g(\mathcal{M})$ & Total Riemannian volume. \\
$\mathrm{inj}(\mathcal{M})$ & Injectivity radius. \\
$\mathrm{diam}(\mathcal{M})$ & Geodesic diameter. \\
$B_r(m_0)$ & Geodesic ball of radius $r$ centered at $m_0$. \\
$\Delta_g$ & Laplace--Beltrami operator. \\
$\kappa$ & Sectional-curvature upper bound. \\
$S(x)$ & Scalar curvature at $x$. \\
$h(r,\kappa)$ & Curvature term in the Fr\'echet-mean sensitivity bound. \\
$F(x)$ & Original Fr\'echet energy. \\
$\eta_{\mathrm{np}}(D)$ & Non-private Fr\'echet mean under $g$. \\

\midrule
\multicolumn{2}{l}{\textbf{Anchor-Based Density Estimation Notations}} \\ \midrule
$\mathcal{P}=\{p_1,\ldots,p_J\}$ & Public anchors constructed by an $\epsilon$-net or CVT. \\
$J$ & Number of anchors, balancing resolution and reconstruction cost. \\
$\mathcal{V}=\{V_1,\ldots,V_J\}$ & Voronoi partition induced by $\mathcal{P}$. \\
$c(D)=(c_1,\ldots,c_J)$ & Anchor histogram count vector. \\
$\Delta_1$ & Histogram sensitivity, $\Delta_1=2$. \\
$Y_j$ & Laplace noise added to $c_j$. \\
$\widetilde c_j,\widetilde c_j^+$ & Noisy and truncated anchor counts. \\
$w_j$ & Normalized noisy anchor weight. \\
$K(\cdot),K_h(\cdot)$ & Base and bandwidth-scaled kernels. \\
$\overline K_{h,j}(x)$ & Normalized kernel centered at $p_j$. \\
$h$ & KDE bandwidth selected using Eq.~(3). \\
$R(K),\mu_2(K)$ & Kernel roughness and second moment. \\
$\lambda_0$ & Small positive regularizer in Eq.~(3). \\
$f_{\mathrm{data}}(x)$ & Non-private intrinsic KDE. \\
$\widetilde f_{\mathrm{data}}(x)$ & Sanitized anchor-based density estimate. \\
$f_{\min},f_{\max}$ & Extrema of $f_{\mathrm{data}}$. \\
$\widetilde f_{\min},\widetilde f_{\max}$ & Extrema of $\widetilde f_{\mathrm{data}}$. \\
$C_K$ & Maximum kernel-overlap constant. \\

\midrule
\multicolumn{2}{l}{\textbf{Conformal Metric Notations}} \\ \midrule
$\upsilon$ & PDE regularization parameter selected using Eqs.~(10)--(11). \\
$-\Delta_g+\upsilon$ & Helmholtz--Poisson operator. \\
$\sigma(x)$ & Sanitized conformal scaling function. \\
$\sigma_{\mathrm p},\sigma_{\mathrm{np}}$ & Private and non-private scaling functions. \\
$\phi(x)=e^{2\sigma(x)}$ & Sanitized conformal factor. \\
$\phi_{\mathrm p},\phi_{\mathrm{np}}$ & Private and non-private conformal factors. \\
$g^*=\phi g$ & Conformal metric. \\
$\rho_{g^*}(x,y)$ & Geodesic distance under $g^*$. \\
$\mu_{g^*},d\mu_{g^*}$ & Conformal volume measure and volume element. \\
$\phi_{\min},\phi_{\max}$ & Lower and upper bounds of $\phi$. \\
$r^*$ & Dataset radius under $g^*$. \\
$\mathcal{C},\mathcal{W}$ & Dense-core and sparse regions. \\
$\tau$ & Upper bound of $\phi$ on $\mathcal{W}$. \\
$\theta$ & Maximum path proportion in $\mathcal{C}$. \\
$C_{\tau,\theta}$ & Contraction factor, $\theta+\sqrt{\tau}(1-\theta)$. \\
$\kappa^*$ & Sectional-curvature upper bound under $g^*$. \\
$C_1,C_2$ & Elliptic regularity constants. \\
$F^*(x)$ & Conformal Fr\'echet energy. \\
$\eta_{g^*}(D)$ & Conformal Fr\'echet mean. \\

\midrule
\multicolumn{2}{l}{\textbf{Differential Privacy Mechanism Notations}} \\ \midrule
$\Delta,\Delta^*$ & Original and conformal Fr\'echet-mean sensitivities. \\
$S^*$ & Analytical upper bound on $\Delta^*$. \\
$\varepsilon_\phi,\varepsilon_{\mathrm{conf}}$ & Stage-1 and Stage-2 privacy budgets. \\
$\varepsilon_{\mathrm{total}}$ & Total budget, $\varepsilon_\phi+\varepsilon_{\mathrm{conf}}$. \\
$\lambda^*$ & Rate parameter, $\varepsilon_{\mathrm{conf}}/(2S^*)$. \\
$C(\eta,\lambda^*)$ & Normalization constant. \\
$p^*(z\mid\eta)$ & Conformal Laplace density. \\
$P^*(\cdot\mid\eta)$ & Probability measure induced by $p^*$. \\
$\mathcal{A}_{\mathrm{anchor}}$ & Stage-1 anchor mechanism. \\
$\mathcal{A}_{\mathrm{conf}}$ & Stage-2 conformal mechanism. \\
$\mathcal{A}_{\mathrm{total}}$ & Sequentially composed mechanism. \\
$L_{D,D'}(z)$ & Privacy loss at output $z$. \\
$Z,z$ & Random private output and its realization. \\
$C_{\mathrm{stab}}(r,\kappa)$ & Fr\'echet-mean stability constant. \\
$E_\phi,E_{\mathrm{conf}}$ & Geometry-distortion and sampling errors. \\
$C_\phi,C_{\mathrm{conf}}$ & Stage-1 and Stage-2 utility coefficients. \\
$\varepsilon_\phi^{\mathrm{opt}},
 \varepsilon_{\mathrm{conf}}^{\mathrm{opt}}$
& Optimal privacy-budget allocation in Eqs.~(39)--(40). \\

\midrule
\multicolumn{2}{l}{\textbf{Computational Notations}} \\ \midrule
$M$ & Number of manifold-discretization vertices. \\
$T$ & Number of MCMC iterations. \\

\bottomrule
\end{tabular}
}
\end{table}

\section{Differentially Private Conformal Factor Release}
\label{section3}
Our conformal-DP mechanism contains two stages of private release. In stage-1, we use an anchor-based Laplace mechanism to release the differentially private kernel density estimation (KDE), thus to protect the conformal factor $\phi$ to avoid information leakage during conformal transformation. In stage-2, we release the final differentially private Fr\'echet mean by adding the Laplace noise that protects the datasets on the Riemannian manifold.

\subsection{Computing Sanitized Conformal Factor.}
\label{section3.1}
In the first stage of designing our mechanism, we need to derive the conformal scaling function, denoted by $\sigma(x)$  and the conformal factor, denoted by $\phi(x)$, which are related by $\phi(x) = e^{\sigma(x)}, x \in R$. The conformal factor plays a crucial role that connects the original manifold with the Riemannian metric $g$ and the conformal metric $g^*$. Here, note that the conformal transformation is not a projection or mapping of the original manifold onto a new manifold; the underlying differentiable manifold is unchanged; only the Riemannian metric structure is altered, as explained in Paragraph \ref{para:topo_analy} (topological analysis).

Given a $d$-dimensional ($d \geq 2$) compact, boundaryless, complete Riemannian manifold  $(\mathcal{M},g)$ equipped with the Borel $\sigma$-algebra \cite{wasserman2010statistical,pmlr-v97-awan19a}, Riemannian metric $g=\{\langle\cdot,\cdot\rangle_m: m\in \mathcal{M}\}$, and volume measure $\mu_g$, to calculate $\sigma(x)$ on $\mathcal{M}$, we adopt a generalized approach with the Laplace–Beltrami (LB) operator $\Delta_g = \text{div}_g\nabla$ \cite{urakawa1993geometry}. \footnote{\label{ft:need-to-discrete}There is rarely an explicit form for $\Delta_g$; we discretize $\Delta_g$ and $H$ over $D\subset\mathcal{M}^N$ to solve for $\sigma$.}

\paragraph{Intrinsic Data Density $f_{\mathrm{data}}(x)$.}
We initially compute $f_{\mathrm{data}}$ via kernel density estimation (KDE) \cite{dc7c4be6-dcd9-3f65-915e-b5469ed5cadb}, which encodes the underlying data density information. We pre-select bandwidth $h\in\bigl(0,\tfrac12\mathrm{inj}(\mathcal M)\bigr)$ (see Paragraph \ref{para:h}) and the nonnegative, compactly supported public dimensionless base radial profile $K\in C^\infty([0,\infty))$ (as captured in Theorem \ref{theorem:h}). {Detailed parameter selection is presented in \cref{section5}.} To construct the bandwidth-scaled spatial kernel $K_h$\,\footnote{\label{ft:kernel_profile} Here, we choose a normalized $C^\infty$ bump base profile, which is tightly supported and infinitely differentiable. To strictly decouple its structural shape from the physical bandwidth $h$, it is defined purely over a dimensionless generic variable $u \geq 0$: $K(u)= \begin{cases} C_K \exp \left(-\frac{1}{1-u}\right), & 0 \leq u<1 \\ 0, & u \geq 1\end{cases}$, where $C_K$ is the normalization constant ensuring $\int_{\mathbb{R}^d} K(\mathrm{u}) d\mathrm{u} = 1$, $u$ represents the square of the vector length $\| \mathrm{v}\|^2$.}, we derive the kernel radial function as follows: denote the geodesic distance $r=\rho_g\left(x, X_i\right) \geq 0$,
\begin{equation}
\label{eq:kernel}
    K_h(r)\ =\ h^{-d}\,K\!\Big(\frac{r^2}{h^2}\Big),
\end{equation}
and the intrinsic KDE is explicitly evaluated as:
\begin{equation}
\label{eq:kde}
    f_{\mathrm{data}}(x)\ =\ \frac{1}{N}\sum_{i=1}^N K_h\!\big(\rho_g(x,X_i)\big), \forall x\in\mathcal M.
\end{equation}

\paragraph{Choose KDE bandwidth $h$.}
\label{para:h}
Kernel bandwidth $h$ is used to acquire the Kernel Density Estimation (KDE) function \cite{dc7c4be6-dcd9-3f65-915e-b5469ed5cadb}. It is a prior public parameter and only related to the manifold properties and public sample size $N$. To find the optimal $h$ on different manifolds, we propose Theorem \ref{theorem:h}.
\begin{theorem}
\label{theorem:h}
    Consider $(\mathcal{M},g)$ with volume $V$, injectivity radius $\mathrm{inj}(\mathcal{M})$, scalar curvature $S(x)$, sample size $N$ and a uniform reference prior $f_{ref}(x) = 1/V$. Let $K$ be a spherically symmetric radial kernel on the tangent space such that $\int_{\mathbb{R}^d} K(\|u\|^2) du = 1$, with roughness $R(K) = \int_{\mathbb{R}^d} K(\|u\|^2)^2 du$ and second moment $\mu_2(K) = \int_{\mathbb{R}^d} u_1^2 K(\|u\|^2) du$.  For a given strictly positive structural regularizer $\lambda_0 > 0$ (very small floating-point numbers, avoid $h \to \infty$), the bandwidth $h$ that minimizes the regularized Asymptotic Mean Integrated Squared Error ($\mathrm{AMISE}$) on $\mathcal{M}$ is as follows:
    {\small
    \begin{equation}
        h=\min \left\{\left(\frac{9 d \cdot V^2 \cdot R(K)}{N \cdot \mu_2(K)^2 \cdot\left(\|S\|_{L^2(\mathcal{M})}^2+\lambda_0\right)}\right)^{\frac{1}{d+4}}, \frac{1}{2} \operatorname{inj}(\mathcal{M})\right\},
    \end{equation}}
    where $\|S\|_{L^2(\mathcal{M})}^2 = \int_{\mathcal{M}} S(x)^2 d\mu_g(x)$ is the squared $L^2$-norm of the scalar curvature over $\mathcal{M}$.
\end{theorem}

\begin{proof}
\label{proof: h}
We provide detailed proof in Appendix \ref{proof_theorem_h}.
\end{proof}

\paragraph{Sanitizing KDE $f_{\mathrm{data}}(x)$ (Stage 1).}
\label{paragraph:adj}
Directly bounding functional sensitivity in the infinite-dimensional space $L^1(\mathcal{M})$ is difficult when the output is defined over a continuous Riemannian manifold. Similar issues have been observed in prior work on DP mechanisms for functional outputs \cite{awan2019benefits}. To obtain a tractable privacy analysis, we replace the continuous formulation with a discrete anchor-based representation.

Let $\mathcal{P} = \{p_1, \dots, p_J\} \subset \mathcal{M}$ be a pre-defined, data-independent set of public anchor points on $\mathcal{M}$, where $J$ denotes the total number of anchors and $j \in \{1,\dots,J\}$ indexes each anchor $p_j$. The anchor set provides a geometric discretization of the manifold prior to privatization. Since the resulting quantization error depends on the geometry of $\mathcal{M}$, anchor construction should be adapted to the target manifold while remaining fully data-independent. In practice, $\mathcal{P}$ can be constructed using deterministic covering schemes, such as epsilon-nets or Centroidal Voronoi Tessellations (CVTs), with respect to the volume measure $\mu_g$ \cite{kanai1986rough,du1999cvt,du2003cvt}. Once $\mathcal{P}$ is fixed, it induces a discrete partition of $\mathcal{M}$, for example, through a geodesic Voronoi tessellation $\mathcal{V} = \{V_1, \dots, V_J\}$ under the intrinsic geodesic metric $\rho_g$, where each anchor $p_j$ is associated with a Voronoi cell $V_j$. Therefore, assigning a data point to its nearest anchor is equivalent to assigning it to the corresponding cell. The manifold-specific construction procedures are described in Section \ref{sec:experiment}.

For a raw manifold dataset $D = \{x_1, \dots, x_N\}$, each point $x_i$ is mapped to its nearest anchor, yielding a discrete histogram count vector $c(D) = (c_1, \dots, c_J) \in \mathbb{N}^J$, where
$c_j = \sum_{i=1}^N \mathbb{I}(x_i \in V_j),$ and $\sum_{j=1}^J c_j = N$. This discretization reduces the privacy analysis to the sensitivity of a finite-dimensional count query.

We define adjacency using neighboring datasets $D \sim D'$ that differ in exactly one data point. Under this substitution model, replacing one point moves at most one unit of mass from one Voronoi cell to another. Therefore, the count vector changes by $-1$ in one coordinate and $+1$ in another, or remains unchanged if both points fall in the same cell. The resulting $L_1$ global sensitivity is the standard histogram sensitivity under substitution adjacency \cite{dwork2014algorithmic}:
\begin{equation}
\label{eq:sensitivity-constant}
    \Delta_1 = \sup_{D\sim D'} \| c(D) - c(D') \|_1 \equiv 2.
\end{equation}
This gives a constant, dimension-independent sensitivity bound and avoids the need to directly control perturbations in continuous $L^1(\mathcal{M})$.

\paragraph{Anchor-Based Laplace Mechanism and Smooth Reconstruction.}
\label{paragraph:privacy_mechanism}
Given a privacy budget $\varepsilon_\phi > 0$ and a discrete count vector $c \in \mathbb{N}^J$, the mechanism $\mathcal{A}_{\mathrm{anchor}}(D)$ proceeds in three steps:

\begin{enumerate}
    \item \textbf{Noise injection and Formal DP Guarantee.}
    We apply the Laplace mechanism independently to each anchor count using the global sensitivity $\Delta_1 = 2$. The randomized output $\tilde{c} \in \mathbb{R}^J$ is generated via:
    \begin{equation}
        \tilde{c}_j = c_j + \mathrm{Lap}\left(\frac{\Delta_1}{\varepsilon_{\phi}}\right), \quad \forall j \in \{1, \dots, J\}.
    \end{equation}
    The output density of the mechanism at a noisy count vector $\tilde{c}$ is proportional to
    $\exp\left(-\frac{\varepsilon_\phi}{\Delta_1}\|\tilde{c}-c(D)\|_1\right).$
    For any two adjacent datasets $D\sim D'$, the likelihood ratio at $\tilde{c}$ is bounded as
    {\small
    \begin{equation}
    \label{eq:eps_phi}
        \frac{\mathbb{P}(\mathcal{A}_{\mathrm{anchor}}(D)=\tilde{c})}
             {\mathbb{P}(\mathcal{A}_{\mathrm{anchor}}(D')=\tilde{c})}
        \le
        \exp\left(
        \frac{\varepsilon_\phi}{\Delta_1}\|c(D)-c(D')\|_1
        \right)
        \le
        e^{\varepsilon_\phi}.
    \end{equation}}
    Therefore, the intermediate discrete mechanism $\mathcal{A}_{\mathrm{anchor}}$ satisfies pure $\varepsilon_\phi$-differential privacy.

    \item \textbf{Projection to a valid discrete distribution (post-processing).}
    To obtain nonnegative weights, we truncate the noisy counts as
    $\tilde{c}_j^+ = \max(0, \tilde{c}_j),$ and normalize them to define $w \in \Delta^{J-1}$:
    $
        w_j = \frac{\tilde{c}_j^+}{\sum_{r=1}^J \tilde{c}_r^+}.
    $
    In the event that $\sum_{r=1}^J \tilde{c}_r^+ = 0$, we set $w_j = 1/J$ for all $j$. This step uses only the noisy counts and is protected by post-processing immunity \cite{dwork2014algorithmic}.

    \item \textbf{Continuous reconstruction via manifold-adapted kernels (post-processing).}
    To reconstruct a continuous density on $\mathcal{M}$, we directly utilize the scaled spatial kernel $K_h$ defined in Eq.~\ref{eq:kernel}. While the original intrinsic KDE (Eq.~\ref{eq:kde}) relies on asymptotic Euclidean scaling $h^{-d}$, securing an exact probability measure in $\mathcal{F}_\infty$ for downstream PDE stability requires correcting for localized geometric volume distortions \cite{pelletier2005kernel}. Therefore, we define a pointwise normalized kernel for each anchor $p_j$:
    \begin{equation}
    \label{eq:manifold_adapted_kernel}
        \bar{K}_{h, j}(x) = \frac{K_h\big(\rho_g(x, p_j)\big)}{\int_{\mathcal{M}} K_h\big(\rho_g(y, p_j)\big) d\mu_g(y)},
    \end{equation}
    where $d\mu_g$ is the Riemannian volume measure. The sanitized intrinsic KDE is then deterministically synthesized as a smooth convex combination of these localized kernels:
    \begin{equation}
    \label{eq:reconstructed_kde}
        \widetilde{f}_{\mathrm{data}}(x) = \sum_{j=1}^J w_j \bar{K}_{h, j}(x).
    \end{equation}
\end{enumerate}

\paragraph{Computational Complexity.}
Given $N$ data points, $J$ anchors, and $M$ downstream spatial evaluation points, constructing the anchor histogram by direct sample-to-anchor accumulation costs $\mathcal{O}(NJ)$, while noise injection and normalization over the $J$ anchor weights cost $\mathcal{O}(J)$. Reconstructing $\widetilde{f}_{\mathrm{data}}$ at $M$ downstream spatial locations from the $J$ anchors costs $\mathcal{O}(MJ)$, matching the standard complexity of direct kernel summation, which scales linearly in the number of centers and evaluation points \cite{fastkde, qin2019akde}. Because the normalization terms in Eq.~\ref{eq:manifold_adapted_kernel} are data-independent, they can be pre-computed offline. Thus, the total online complexity is $\mathcal{O}(NJ + MJ)$. This separates the cost of privatization from mesh refinement: increasing $M$ affects downstream evaluation only, not the privacy mechanism itself.

\subsection{From KDE to Conformal Factor and PDE Regularity}
\label{subsec:eps_phi}

Given the sanitized density $\tilde{f}_{\mathrm{data}} \in \mathcal{F}_{\infty}$, we define the conformal factor $\phi(x)=e^{2\sigma(x)}$ through the scaling function $\sigma(x):\mathcal M\to\mathbb R$, which is the unique solution to the Helmholtz--Poisson equation \cite{hilbert1985methods,beck2016elliptic,fernandez2023regularity}:
\begin{equation}
\label{eq:poisson}
    (-\Delta_g+\upsilon)\,\sigma(x)\;= \tilde{f}_{\mathrm{data}}(x) - \tilde{f}_{\mathrm{max}}, \quad \upsilon>0,\ x\in\mathcal M.
\end{equation}
Here,
$
\tilde{f}_{\mathrm{max}} \triangleq \operatorname{\,sup}_{x \in \mathcal{M}} \tilde{f}_{\mathrm{data}}(x)
$
denotes the supremum of the sanitized density. Since $\tilde{f}_{\mathrm{max}}$ is computed deterministically from the already privatized output $\tilde{f}_{\mathrm{data}}$, it is obtained by post-processing and incurs no additional privacy cost. The parameter $\upsilon$ is public and ensures coercivity of the operator $(-\Delta_g+\upsilon)$; see Paragraph \ref{para:upsilon}.

Because the anchor-based mechanism reconstructs $\tilde{f}_{\mathrm{data}}$ using smooth manifold-adapted kernels, the right-hand side of \cref{eq:poisson} is smooth. This allows us to apply standard elliptic regularity theory directly to the manifold in the following Lemma (proof is in in Appendix \ref{proof_lemma1}.)

\begin{lemma}[{\cite{gilbarg1977elliptic,jost2008riemannian}}]
\label{lemma1}
Let $(\mathcal{M},g)$ be a closed Riemannian manifold. For any $\upsilon>0$ and any smooth right-hand side function $H\in C^\infty(\mathcal M)$, the equation
$
(-\Delta_g+\upsilon)\,u \;=\; H
$
admits a unique classical solution $u\in C^\infty(\mathcal M)$. In particular, with
$
H = \tilde{f}_{\mathrm{data}} - \tilde{f}_{\mathrm{max}},
$
\cref{eq:poisson} achieves a unique smooth solution $\sigma\in C^\infty(\mathcal M)$.
\end{lemma}

This regularity result is important because it avoids the weaker solution that would arise if only $L^\infty(\mathcal{M})$ control were available for the privatized density. In our setting, the reconstructed density satisfies $\tilde{f}_{\mathrm{data}} \in C^\infty(\mathcal{M})$, so the source term in \cref{eq:poisson} meets the assumptions of Lemma \ref{lemma1} directly, without requiring any additional smoothing step. As a result, the scaling function $\sigma$ is smooth, and the conformal factor $\phi(x)=e^{2\sigma(x)}$ is correspondingly smooth. Combined with the boundedness results established later, this ensures that the transformed metric remains well-defined and non-degenerate.

\paragraph{Regularization Parameter, $\upsilon$.}
\label{para:upsilon}
The regularization parameter $\upsilon$ controls the trade-off between the numerical stability of the operator $\mathcal{L}_\upsilon = -\Delta_g + \upsilon$ and the preservation of intrinsic geometric structures. Under the Laplacian eigenbasis expansion, a large $\upsilon$ strongly suppresses the dominant Fiedler mode $\lambda_1$ \cite{fiedler1973algebraic}, which can drive the conformal metric toward degeneration. From the spectral viewpoint, \(\upsilon\) creates a trade-off: a large value overly damps the main low-frequency geometric structure and can make the conformal transformation too weak, whereas a very small value leads to numerical instability. To balance numerical stability with geometric fidelity, we select $\upsilon$ by minimizing the competing effects between geometric signal distortion $(\mathcal{E}{\mathrm{geo}} \approx \upsilon/\lambda_1)$ and numerical instability $(\mathcal{E}{\mathrm{num}} \approx \lambda_1/\upsilon)$. Minimizing the joint relative penalty yields the strict theoretical optimum:
\begin{equation}
    \upsilon_{\mathrm{opt}} = \arg\min_{\upsilon > 0} \left( \frac{\upsilon}{\lambda_1} + \frac{\lambda_1}{\upsilon} \right) = \lambda_1
\end{equation}
In practice, extracting the exact first nonzero eigenvalue $\lambda_1$ via sparse eigendecomposition can be computationally expensive on dense meshes, since it requires solving a large-scale sparse eigenproblem \cite{nasikun2018fast}. Instead, we safely anchor $\upsilon$ using dataset-independent universal geometric priors. To choose \(v\), we rely on two classical lower bounds for \(\lambda_1\). The first is Cheeger’s inequality \cite{cheeger2015lower}, \(\lambda_1 \ge h_C^2/4\), where \(h_C\) reflects the geometric bottleneck of the manifold. The second is the Zhong–Yang bound \cite{chavel1984eigenvalues}, \(\lambda_1 \ge \pi^2/\operatorname{diam}(\mathcal{M})^2\), which applies when the manifold has non-negative Ricci curvature. Using the sharper of these public bounds lets us set \(v\) in a way that keeps the PDE operator stable, without touching the private dataset and therefore without affecting the pure \(\epsilon\)-DP guarantee.

To make the operator stable while avoiding too much attenuation of the main geometric component, we choose $v$ independently of the dataset using the best available lower bound on $\lambda_1$:
\begin{equation}
    \upsilon = \max\left\{ \frac{h_C^2}{4}, \, \frac{\pi^2}{\operatorname{diam}(\mathcal{M})^2} \right\}
\end{equation}
This choice yields a well-conditioned PDE operator while preserving the pure $\varepsilon_{\phi}$-DP guarantee, since the parameter is determined entirely from public geometric bounds.

\paragraph{Properties of the Conformal Factor}
\label{subsec:properties_conformal_factor}
This conformal transformation yields an important structural guarantee: the induced conformal factor is always bounded and can only contract, never expand, the original geometry, as we show in the next proposition. Such boundedness is important for ensuring that the privatized geometric representation remains stable and geometrically meaningful.

\begin{proposition}[Bound of the Conformal Factor.]
\label{prop:conformal_factor_properties}
Let $\sigma(x)$ be the unique solution to \cref{eq:poisson}, and define the conformal factor
$\phi(x) = e^{2\sigma(x)}.$
Then, for all $x \in \mathcal{M}$,
\begin{equation}
\label{eq:sigma_bounds_combined}
\frac{\tilde{f}_{\mathrm{min}} - \tilde{f}_{\mathrm{max}}}{\upsilon}
\le
\sigma(x)
\le
0,
\end{equation}
and therefore
\begin{equation}
\label{eq:phi_bounds_combined}
\exp\left(
2\frac{\tilde{f}_{\mathrm{min}} - \tilde{f}_{\mathrm{max}}}{\upsilon}
\right)
\le
\phi(x)
\le
1.
\end{equation}
Therefore, the conformal mechanism guarantees a uniformly bounded distortion of the geometry while preventing any local metric expansion.
\end{proposition}

\begin{proof}
    We give detailed proof in Appendix \ref{proof_prop:conformal_factor_properties}.
\end{proof}

\subsection{Utility and Accuracy Guarantees}
\label{sec:utility_accuracy}

To establish the statistical validity of the DP framework and quantify its downstream geometric fidelity, we rigorously bound the pointwise error between the private conformal factor $\phi_{p}(x) = e^{2\sigma_p(x)}$ and non-private counterpart $\phi_{np}(x) = e^{2\sigma_{np}(x)}$ (non-private conformal factor means no stage-1 DP incurred); see paragraph \ref{paragraph:adj}.

\paragraph{Error Representation and PDE Stability.}
\label{para:error_rep_pde}
Define the non-private scaling function $\sigma_{np}$ via $(-\Delta_g+\upsilon)\sigma_{np} = f_{\mathrm{data}} - f_{\mathrm{max}}$, where $f_{\mathrm{max}} = \max_{x\in\mathcal{M}} f_{\mathrm{data}}(x)$. Correspondingly, the private scaling function $\sigma_p$ is governed by $(-\Delta_g+\upsilon)\sigma_p = \tilde{f}_{\mathrm{data}} - \tilde{f}_{\mathrm{max}}$. Subtracting these respective governing equations isolates the error field $e_\sigma = \sigma_p - \sigma_{np}$:
\begin{equation}
    (-\Delta_g+\upsilon)e_\sigma = (\tilde{f}_{\mathrm{data}} - \tilde{f}_{\mathrm{max}}) - (f_{\mathrm{data}} - f_{\mathrm{max}}).
\end{equation}
The global $L^\infty$ stability of the elliptic PDE ensures that the continuous error field is uniformly bounded by the supremum of the source perturbation:
\begin{equation}
\label{eq:err-linfty}
    \|e_\sigma\|_{L^\infty(\mathcal{M})} \le \frac{1}{\upsilon} \left\| (\tilde{f}_{\mathrm{data}} - \tilde{f}_{\mathrm{max}}) - (f_{\mathrm{data}} - f_{\mathrm{max}}) \right\|_{L^\infty(\mathcal{M})}
\end{equation}

\paragraph{Impact on the Conformal Factor.}
\label{para:impact_conf_factor}
The PDE stability result transfers directly to the conformal factor. As shown in Paragraph \ref{para:error_rep_pde}, replacing the right-hand side with the corresponding essential supremum implies that $\sigma_p(x) \le 0$ and $\sigma_{np}(x) \le 0$ on $\mathcal{M}$.

This sign constraint is useful because it restricts both scaling functions to the non-expansive region of the exponential map. In particular, the function $y \mapsto e^{2y}$ is Lipschitz continuous on $(-\infty,0]$, with $L = \sup_{y \le 0} 2e^{2y} = 2.$
Applying the mean value theorem, the pointwise error of the conformal factor can therefore be bounded directly in terms of the PDE error:
\begin{equation}
\label{eq:phi-diff}
\begin{aligned}
        &|\phi_p(x) - \phi_{np}(x)|\\
        &\le 2 \|e_\sigma\|_{L^\infty(\mathcal{M})} \\
        &\le \frac{2}{\upsilon} \left\| (\tilde{f}_{\mathrm{data}} - \tilde{f}_{\mathrm{max}}) - (f_{\mathrm{data}} - f_{\mathrm{max}}) \right\|_{L^\infty(\mathcal{M})}
\end{aligned}
\end{equation}
Thus, the error in the conformal factor is controlled by the error in the sanitized density through the PDE stability bound, without additional exponential amplification. This shows that the induced geometric deformation remains bounded and predictable under privatization.


Next, we present Theorem \ref{theorem1} to prove that the conformally transformed metric $g^*$ satsifies the same properties of smoothness and completeness as $g$ (the proof is in Appendix D). 
\begin{theorem}
\label{theorem1}
Let $(\mathcal{M},g)$ be a smooth Riemannian manifold and
$\Omega_{g}$ a second-order, strongly elliptic operator with sufficiently smooth coefficients.
For any $H \in C^{\infty}(M)$ and $v>0$, the equation $\left(-\Delta_g+\right. v) \sigma=H$ produces a unique $\sigma \in C^{\infty}(M)$.
Consequently, the conformal factor $\phi(x)=e^{2\sigma(x)}>0$ is smooth, and the conformal metric
$g^{*}(x)=e^{2\sigma(x)}\,g(x)=\phi(x)\,g(x)$ is also smooth and positive-definite on $\mathcal{M}$. .
\end{theorem}

\subsection{Conformal Geodesic Distance under the Sanitized Metric}
\label{section3.2}

Using the sanitized conformal factor, we define the conformally transformed metric $g^*(x)=\phi(x)g(x).$
Since $(\mathcal{M},g)$ is a closed Riemannian manifold, it is complete, and the Hopf--Rinow theorem guarantees the existence of length-minimizing geodesics between any two points on $\mathcal{M}$ \cite[Theorem~6.19]{lee2018introduction}; see also \cite{reimherr2021differential}.

Let $\{x^1,\ldots,x^N\}$ be local coordinates, and let $\gamma:[a,b]\to\mathcal{M}$ be a differentiable curve with local representation $\gamma(t)=\bigl(\gamma^1(t),\ldots,\gamma^N(t)\bigr).$
Its length under the conformal metric $g^*$ is
\[
L_{g^*}(\gamma)=\int_a^b \sqrt{g^*_{\gamma(t)}(\dot{\gamma}(t),\dot{\gamma}(t))}\,dt.
\]
Using $g^*(x)=\phi(x)g(x)$, we obtain
\begin{equation}
\label{eq1}
    L_{g^*}(\gamma)=\int_a^b \sqrt{\phi(\gamma(t))}\sqrt{g_{\gamma(t)}(\dot{\gamma}(t),\dot{\gamma}(t))}\,dt.
\end{equation}
Accordingly, the geodesic distance induced by $g^*$ is
\begin{equation}
\label{eq2}
    \rho_{g^*}(x,y)=\inf_{\substack{\gamma:[a,b]\to\mathcal{M}\\ \gamma(a)=x,\gamma(b)=y}} L_{g^*}(\gamma),
\end{equation}
and substituting \cref{eq1} into \cref{eq2} gives
\begin{equation}
\label{eq3}
    \rho_{g^*}(x,y)=\inf_{\substack{\gamma:[a,b]\to\mathcal{M}\\ \gamma(a)=x,\gamma(b)=y}}
    \int_a^b \sqrt{\phi(\gamma(t))}\sqrt{g_{\gamma(t)}(\dot{\gamma}(t),\dot{\gamma}(t))}\,dt.
\end{equation}

In general, $\rho_{g^*}(x,y)$ does not admit a closed-form expression in terms of $\rho_g(x,y)$ unless $\phi$ is constant, since the geodesics of $g^*$ and $g$ need not coincide \cite{jost2008riemannian}. For our analysis, however, an explicit formula is not required. It is sufficient to use the uniform bounds on $\phi$ established in Proposition \ref{prop:conformal_factor_properties}.

\begin{lemma}[Bounds for the Conformal Geodesic Distance]
\label{lemma_Bounds for the Conformal Geodesic Distance}
There exist constants $\phi_{\min}$ and $\phi_{\max} = 1$ such that, for all $x\in\mathcal{M}$, $0<\phi_{\min}\le \phi(x)\le 1,$ with Eq. \ref{eq:phi_bounds_combined}.
\end{lemma}

These bounds immediately yield a comparison between curve lengths under $g$ and $g^*$. From \cref{eq1}, for any differentiable curve $\gamma:[a,b]\to\mathcal{M}$,
$
\sqrt{\phi_{\min}}\,L_g(\gamma)\le L_{g^*}(\gamma)\le L_g(\gamma),
$
Taking the infimum over all curves connecting $x$ and $y$ gives the corresponding distance comparison.

\begin{corollary}
\label{corollary1}
For every pair of points $x,y\in\mathcal{M}$, the geodesic distances under $g$ and $g^*$ satisfy
\begin{equation}
\label{eq4}
    \sqrt{\phi_{\min}}\,\rho_g(x,y)\le \rho_{g^*}(x,y)\le \rho_g(x,y).
\end{equation}
\end{corollary}

The proofs of Lemma \ref{lemma_Bounds for the Conformal Geodesic Distance} and Corollary \ref{corollary1} are provided in Appendix E. The key point is that these bounds allow us to control conformal geodesic distances without explicitly solving the geodesic equations under $g^*$. This is sufficient for the downstream sensitivity and distance analysis. 

\subsection{Conformal Volume Transformation under the Sanitized Metric}
\label{section3.3}

The conformal transformation also changes the Riemannian volume measure. Standard results in Riemannian geometry characterize how volume elements transform under conformal changes of the metric \cite{chavel1995riemannian,schoen1984conformal,topping2006lectures}. In local coordinates $\{x^1,\dots,x^N\}$, the volume element associated with $g$ is
\[
d\mu_g(x)=\sqrt{\det(g_{ij}(x))}\,dx^1\wedge\cdots\wedge dx^N,
\]
where $[g_{ij}(x)]_{i,j=1}^N$ is the metric tensor matrix.

Under the conformal metric
$
g^*(x)=\phi(x)g(x),
$
the local metric matrix becomes
$
[g^*_{ij}(x)]_{i,j=1}^N=\phi(x)[g_{ij}(x)]_{i,j=1}^N.
$
Since this is an $N\times N$ matrix, its determinant scales as
$
\det(g^*_{ij}(x))=\phi(x)^N\det(g_{ij}(x)).
$
Therefore, the corresponding volume element transforms as
\begin{equation}
\label{eq5}
\begin{aligned}
d\mu_{g^*}(x)
&=\sqrt{\det(g^*_{ij}(x))}\,dx^1\wedge\cdots\wedge dx^N\\
&=\sqrt{\phi(x)^N\det(g_{ij}(x))}\,dx^1\wedge\cdots\wedge dx^N\\
&=\phi(x)^{\frac{N}{2}}\sqrt{\det(g_{ij}(x))}\,dx^1\wedge\cdots\wedge dx^N\\
&=\phi(x)^{\frac{N}{2}}\,d\mu_g(x).
\end{aligned}
\end{equation}

Thus, the sanitized conformal factor rescales the original volume measure pointwise by $\phi(x)^{N/2}$. This identity will be used later when relating privatized density, distance, and geometric measure under the transformed metric.

\section{Differentially Private Fr\'echet Mean Release}
\label{section4}
Existing works have demonstrated that achieving $\varepsilon$-DP is feasible in measure spaces equipped with the Borel $\sigma$-algebra \cite{wasserman2010statistical, pmlr-v97-awan19a}. In the second stage, we propose the final DP mechanism, denoted as $\mathcal{A}_{\text{conf}}$, defined as the process of sampling from the Laplace kernel-based probability distribution (Eq. \ref{eq:conformal_kernel}), which is constructed using the conformal factor $\sigma$ and the conformal metric $g^*$ derived from the sanitized density $\tilde{f}_{\text{data}}$.

\subsection{Conformal Probability Measure}
\label{sec:conformal_measure}
Consider a dataset $D=\{x_1,\dots,x_n\}$ on a compact Riemannian manifold $(\mathcal{M},g)$. Let $\eta_{np}(D)$ denote a non-private statistical summary of the dataset (in this paper, we mainly focus on Fr\'echet Mean as the statistical summary that we want to release), and let $\rho_g$ denote the geodesic distance induced by the original metric $g$. Our goal is to construct a privatized counterpart $\eta_p(D)$ while preserving the geometric structure of the underlying manifold.

As a starting point, we recall the global sensitivity bound for the Fr\'echet mean under the original metric, following Reimherr et al.~\cite{reimherr2021differential}.

\begin{definition}
\label{def2}
(Reimherr et al.~\cite{reimherr2021differential})
Assume the data $D \subseteq B_r(m_0)$ for some $m_0$, where
\[
r<r^* =\frac{1}{2}\min\left\{\operatorname{inj}\mathcal{M}, \frac{\pi}{2}\kappa^{-1/2}\right\},
\]
and $\kappa>0$ is an upper bound on the sectional curvatures of $\mathcal{M}$. Consider two adjacent datasets
$
D=\{x_1,\ldots,x_{N-1},x_N\}
\,\text{and}\,
D'=\{x_1,\ldots,x_{N-1},x_N'\}.
$
If $\bar{x}$ and $\bar{x}'$ are the corresponding Fr\'echet means of $D$ and $D'$, then
{\footnotesize
\begin{equation*}
\rho_g(\bar{x},\bar{x}')
\le
\frac{2r(2-h(r,\kappa))}{N\,h(r,\kappa)},
\qquad
h(r,\kappa)=
\begin{cases}
2r\sqrt{\kappa}\cot(\sqrt{\kappa}\,2r), & \kappa>0,\\
1, & \kappa\le 0.
\end{cases}
\end{equation*}
}
\end{definition}

We denote this global sensitivity bound by
\begin{equation}
\label{eq:sensitivity}
    \Delta = \frac{2r(2-h(r,\kappa))}{N\,h(r,\kappa)}.
\end{equation}

This quantity controls the largest possible displacement of the Fr\'echet mean under a single-point substitution in $(\mathcal{M},g)$ \cite{reimherr2021differential}. Standard manifold DP mechanisms calibrate noise directly from this bound under the original metric. In contrast, our goal is to construct a \emph{density-aware} perturbation by calibrating the distribution in the conformally transformed metrics $(\mathcal{M},g^*)$. We first compute a privatized Fr\'echet mean $\eta_{p,\mathrm{conf}}(D)$ under the conformal metric and then transform it back to the original manifold to obtain the final private estimate $\eta_p(D)$. Importantly, the proposed method does not implement dynamic noise calibration at the mechanism level. Rather, it reshapes the geometry of the manifold by modifying its metric structure through the conformal factor. This geometric transformation changes how perturbations are distributed in the transformed space, and after pulling the result back to the original manifold, yields an effect equivalent to density-adaptive noise injection.

Let $\eta_{np}(D) \in \mathcal{M}$ denote the exact Fr\'echet mean of a dataset $D$. On the original Riemannian manifold $(\mathcal{M}, g)$, we assume that the global sensitivity of the Fr\'echet mean is bounded by a known scalar $\Delta$. Following Reimherr et al.~\cite{reimherr2021differential}, for any adjacent datasets $D \sim D'$ differing in one record, $\rho_g\bigl(\eta_{np}(D), \eta_{np}(D')\bigr) \le \Delta.$

To incorporate density information into the perturbation mechanism, we move from the original geometry $(\mathcal{M}, g)$ to the conformally transformed geometry $(\mathcal{M}, g^*)$ constructed in Stage 1. We then define a Laplace-type kernel centered at the Fr\'echet mean $\eta \in \mathcal{M}$ using the conformal geodesic distance $\rho_{g^*}$:
\begin{equation}
\label{eq:conformal_kernel}
    \widetilde{K}^*(z \mid \eta) = \exp\left\{-\lambda^* \rho_{g^*}(\eta, z)\right\},  z \in \mathcal{M},
\end{equation}
where $\lambda^* > 0$ is a global rate parameter under the conformal metric. Since $\mathcal{M}$ is compact, $\rho_{g^*}(\eta, \cdot)$ is continuous and bounded, so the kernel is integrable over $\mathcal{M}$. This yields the normalization constant
\begin{equation}
\label{eq:norm_constant}
    C(\eta, \lambda^*) = \int_{\mathcal{M}} \exp\left\{-\lambda^* \rho_{g^*}(\eta, u)\right\} \, d\mu_{g^*}(u).
\end{equation}

Normalizing \cref{eq:conformal_kernel} gives the sampling density of the mechanism:
\begin{equation}
\label{eq:conformal_pdf}
    \mathbb{P}^*(z \mid \eta) = \frac{\exp\left\{-\lambda^* \rho_{g^*}(\eta, z)\right\}}{C(\eta, \lambda^*)}, \qquad z \in \mathcal{M}.
\end{equation}

\begin{proof}
    To derive Eq. \ref{eq:conformal_pdf}, we provide detailed proofs in the Appendix F.
\end{proof}

\paragraph{Density-Awareness via Intrinsic Geometry.}
Crucially, this distribution is density-aware without explicitly varying the noise scale across space. The adaptation is induced by the conformal geometry itself: because $g^*=\phi(x)g$ with $0<\phi(x)\le 1$, the conformal distance $\rho_{g^*}$ reflects the sanitized density profile. As a result, the same global rate parameter $\lambda^*$ yields sharper concentration in denser regions and broader dispersion in sparser regions when viewed on the original manifold.

\subsection{Conformal Sensitivity}
\label{subsec:conformal_sensitivity}
The sensitivity of the Fr\'echet mean depends on the dataset radius $r$ and an upper bound $\kappa$ on the sectional curvature, see Eq. \ref{eq:sensitivity}. Under the conformal transformation $g^*=e^{2\sigma}g$, both the distance structure and the curvature are modified. By leveraging the Helmholtz--Poisson equation defining $\sigma$, we can derive an analytical sensitivity bound valid for general $d$-dimensional Riemannian manifolds. Although the curvature analysis differs between the cases $d=2$ and $d\ge 3$, the same leading-order sensitivity reduction is preserved in both settings.

\begin{theorem}[Density-Aware Conformal Sensitivity]
\label{thm:conformal_sensitivity}
Let $(\mathcal{M}, g)$ be a $d$-dimensional Riemannian manifold with $d\ge 2$. Suppose the dataset $D$ is contained in a geodesic ball $B_r(m_0)$, the conformal factor $\phi(x)=e^{2\sigma(x)}$ is determined by
Eq. \ref{eq:poisson}. Let the sparsity contraction factor be as follows:
\[
C_{\tau,\theta}=\theta+\sqrt{\tau}(1-\theta)<1,
\]
where $\tau$ bounds the sparse region $\mathcal{W}$ through $\phi\le \tau$, and $\theta\in[0,1)$ denotes the maximal path proportion lying in the dense core $\mathcal{C}$. Let $\kappa^*$ denote an upper bound on the sectional curvature of the conformal metric $g^*$. In particular,
\[
\kappa^*=\frac{1}{\phi_{\min }} \begin{cases}\kappa+\tilde{f}_{\max }-\tilde{f}_{\min }, & d=2 \\ \kappa+2 C_2+C_1, & d \geq 3\end{cases}
\]
where $C_1$ and $C_2$ are deterministic elliptic regularity constants controlling the gradient and Hessian of $\sigma$ \cite{gilbarg1998elliptic}. Then, provided the conformal radius remains in the strongly convex regime $C_{\tau,\theta}r < \frac{\pi}{2}(\kappa^*)^{-1/2}$, the conformal sensitivity satisfies the upper bound as follows:
\begin{equation}
\label{eq:merged_conformal_sensitivity}
\Delta^*
\le
C_{\tau,\theta}\left(\frac{2r}{N}\right)
+
C_{\tau,\theta}^3
\left(\frac{16r^3}{3n}\right)\kappa^*
+
C_{\tau,\theta}^5
\left(\frac{4\mathcal{O}r^5}{3n}\right)(\kappa^*)^2,
\end{equation}
where $\mathcal{O} = \sup_{z \in (0, z_{\max})} \frac{d^4}{dz^4}\left(2\frac{\tan z}{z}-1\right)$ with $z_{\max} = 2C_{\tau,\theta}r\sqrt{\kappa^*}$ bounds the explicit truncation error.
\end{theorem}

\begin{proof}
    We provide detailed proof in Appendix G.
\end{proof}

\subsection{Privacy Calibration}
\label{subsec:privacy_calibration}

\paragraph{Privatized Fr\'echet Mean under the Conformal Metric.}
\label{section5.1}
We use the Fr\'echet (or Karcher) mean as the target statistic for the \emph{Conformal}-DP mechanism, as it is a standard summary in Riemannian statistics \cite{karcher1977riemannian,kendall1990probability}. For a dataset $D=\{x_1,\dots,x_N\}\in\mathcal{M}^N$, the Fr\'echet energy on the original manifold $(\mathcal{M},g)$ is
\begin{equation}
    F(x)=\frac{1}{2N}\sum_{i=1}^N \rho_g^2(x,x_i).
\end{equation}
In the Euclidean case $(\mathcal{M},g)=\bigl(\mathbb{R}^d,\langle\cdot,\cdot\rangle\bigr)$, this reduces to the classical least-squares objective $F_{\mathbb{R}^d}(x)=\frac{1}{2N}\sum_{i=1}^N \|x-x_i\|_2^2.$
The Fr\'echet mean $\eta(D)$ is defined as the unique minimizer of this energy: $\eta(D)=\arg\min_{x\in\mathcal{M}} F(x).$
Under the standard small-ball and curvature conditions of Karcher's theorem \cite{karcher1977riemannian}, this minimizer exists and is unique. At the optimum, the Riemannian gradient vanishes:
\begin{equation}
    \nabla F(\eta(D))
    =
    -\frac{1}{N}\sum_{i=1}^N \log_{\eta(D)}(x_i)
    =
    0,
\end{equation}
where $\log_x(\cdot)$ denotes the Riemannian logarithmic map at $x$.

To adapt this construction to the conformal mechanism, we transform the original metric $g$ with the conformal metric $g^*= \phi(x)g=e^{2\sigma(x)}g.$
For a dataset $D=\{x_1,\dots,x_N\}$, the corresponding conformal Fr\'echet energy is
\begin{equation}
\label{eq:conformal_energy}
    F^*(x)=\frac{1}{2N}\sum_{i=1}^N \rho_{g^*}^2(x,x_i).
\end{equation}

\paragraph{Target Statistic: Conformal Fr\'echet Mean.}
Rather than privatizing the Fr\'echet mean under the original metric $g$, we define the target statistic directly with respect to the conformal metric $g^*$. For a dataset $D\in\mathcal{M}^N$, the \emph{conformal Fr\'echet mean} is
\begin{equation}
\label{eq:conformal_frechet_mean}
    \eta_{g^*}(D):=\arg\min_{x\in\mathcal{M}} F^*(x).
\end{equation}
The corresponding first-order optimality condition is
$
\sum_{i=1}^N \log^{(g^*)}_{\eta_{g^*}(D)}(x_i)=0,
$
where $\log^{(g^*)}$ is the logarithmic map under the conformal metric. Because the metric has been transformed to $g^*$, both the geodesic distance $\rho_{g^*}$ and the associated logarithmic maps differ from those under the original geometry. It is important to distinguish the computation of the conformal Fr\'echet mean from the randomization used for its private release. The conformal Fr\'echet mean $\eta_{g^*}(D)$ is a deterministic statistic obtained by minimizing the conformal Fr\'echet energy
$F^*(x)$ in Eq.~(29). MCMC is not used to compute or sample this
mean. After $\eta_{g^*}(D)$ has been obtained, MCMC is used only as
a numerical procedure to draw the final private output
$Z\in\mathcal{M}$ from the conformal-Laplace distribution defined
in Eq.~(25). Therefore, the optimization of the target statistic
and the sampling of the private output are two separate steps.

To calibrate the noise scale $\lambda^*$ for differential privacy, we bound the sensitivity of this target statistic under $g^*$. We define the \emph{conformal sensitivity} as
\begin{equation}
\label{eq:conformal_sensitivity}
    \Delta^* =\sup_{D\sim D'} \rho_{g^*}\bigl(\eta_{g^*}(D),\eta_{g^*}(D')\bigr),
\end{equation}
where $D\sim D'$ denotes adjacent datasets. A direct metric comparison gives the coarse bound $\Delta^*\le \sqrt{\phi_{\max}}\,\Delta$, but this does not exploit the density-aware geometry induced by the conformal factor. Instead, using the sensitivity analysis from the previous subsection, we adopt the sharper bound
\begin{equation}
\label{eq:tight_sensitivity_bound}
    \Delta^* \le \mathcal{S}^*
    :=
    C_{\tau,\theta}\left(\frac{2r}{N}\right)+\mathcal{K}_{\mathrm{curv}},
\end{equation}
where $C_{\tau,\theta}$ is the sparsity contraction factor defined in Theorem \ref{thm:conformal_sensitivity}, $r$ is the support radius, and $\mathcal{K}_{\mathrm{curv}}$ denotes the corresponding curvature correction term. We use $\mathcal{S}^*$ as the analytical upper bound for privacy calibration.

We now give the privacy guarantee of the conformal mechanism. The conformal factor $\phi$, the metric $g^*$, and the volume measure $\mu_{g^*}$ are pre-calculated in Stage 1 and treated as fixed in Stage 2. Accordingly, $\mu_{g^*}$ serves as the reference measure for the output density. To ensure that the normalization constant $C(\cdot,\lambda^*)$ is finite, we assume that $\mathcal{M}$ is compact or, more generally, that it has bounded volume growth under $g^*$.

\begin{theorem}[Conformal-DP Guarantee]
\label{thm:main_conformal_dp}
Let $\mathcal{A}_{\text{conf}}$ be the randomized mechanism that releases an output $z \in \mathcal{M}$ sampled from the conditional probability measure $\mathbb{P}^*(\cdot \mid \eta_{g^*}(D))$ governed by the density $p^*(z \mid \eta_{g^*}(D)) = \frac{1}{C(\eta_{g^*}(D), \lambda^*)} \exp\bigl\{-\lambda^* \rho_{g^*}(\eta_{g^*}(D), z)\bigr\}$ with respect to the reference measure $\mu_{g^*}$. To ensure the privacy loss does not exceed the budget given the analytical sensitivity bound $\Delta^* \le \mathcal{S}^*$, the rate parameter is calibrated as:
\begin{equation}
\label{eq:lambda_calibration}
    \lambda^* = \frac{\varepsilon_{\text{conf}}}{2\mathcal{S}^*} \le \frac{\varepsilon_{\text{conf}}}{2\Delta^*}.
\end{equation}
With this calibration, $\mathcal{A}_{\text{conf}}$ satisfies pure $\varepsilon_{\text{conf}}$-differential privacy under the conformal metric. That is, for all adjacent datasets $D \sim D'$ and any measurable set $S \subseteq \mathcal{M}$, the output distribution satisfies the measure-theoretic privacy bound:
\begin{equation}
\label{eq:formal_dp_bound}
    \frac{\int_S p^*\bigl(z \mid \eta_{g^*}(D)\bigr) \, d\mu_{g^*}(z)}{\int_S p^*\bigl(z \mid \eta_{g^*}(D')\bigr) \, d\mu_{g^*}(z)}  \le e^{\varepsilon_{\text{conf}}}. 
\end{equation}
\end{theorem}

\begin{proof}
    We give detailed proof in Appendix H. 
\end{proof}

\begin{remark}
\label{rem:topology_pullback}
Although the mechanism is defined under the conformal metric $g^*$, the output does not need to be mapped back to the original manifold. The conformal transformation $g \mapsto g^*=\phi g$ changes only the metric tensor on $\mathcal{M}$ and leaves the underlying point set and smooth structure unchanged. Thus, it alters geometry but not topology. Consequently, a privatized output $z$ sampled under the conformal density is already a valid point on $\mathcal{M}$. No additional pull-back operation is required. The original metric $g$ is used only for downstream evaluation, such as computing the utility error $\rho_g(z,\eta(D))$.
\end{remark}

\section{Theoretical Guarantees of Privacy and Utility}
\label{section5}

\subsection{Composition Analysis}
Since we have two stages of sanitization, we employ the composition theorem \cite{kairouz2015composition} to analyze the privacy budgets.

\begin{theorem}[Sequential Composition for Conformal-DP]
\label{thm:composition}
Let $\mathcal{A}_{\mathrm{anchor}}$ be an $\varepsilon_\phi$-differentially private mechanism that outputs a noisy count vector $\tilde{c}\in\mathbb{R}^J$. For each realized $\tilde{c}$, let $\mathcal{A}_{\mathrm{conf}}(\cdot \mid \tilde{c})$ denote the second-stage conformal mechanism defined under the conformal metric $g^*$ and reference measure $\mu_{g^*}$ induced by $\tilde{c}$. Specifically, $\mathcal{A}_{\mathrm{conf}}$ outputs $z\in\mathcal{M}$ according to the conditional density
\[
p^*\bigl(z \mid \eta_{g^*}(D), \tilde{c}\bigr)
\propto
\exp\bigl\{-\lambda^* \rho_{g^*}(\eta_{g^*}(D), z)\bigr\} \, w.r.t. \mu_{g^*}.
\]
Assume the rate parameter is calibrated as
\[
\lambda^*=\frac{\varepsilon_{\mathrm{conf}}}{2\mathcal{S}^*}
\le
\frac{\varepsilon_{\mathrm{conf}}}{2\Delta^*}.
\]
Let the two-stage mechanism be as follows:
\[
\mathcal{A}_{\mathrm{total}}(D):=(\tilde{c},z),
\qquad
\tilde{c}\sim \mathcal{A}_{\mathrm{anchor}}(D),
\quad
z\sim \mathcal{A}_{\mathrm{conf}}(\cdot\mid \tilde{c}).
\]
Then $\mathcal{A}_{\mathrm{total}}$ satisfies $(\varepsilon_\phi+\varepsilon_{\mathrm{conf}})$-differential privacy.
\end{theorem}

\begin{proof}
    We provide the detailed proof in Appendix I.
\end{proof}

\paragraph{Scope of the Current Privacy Guarantee.}
Thm. \ref{thm:composition} establishes a pure $(\varepsilon_{\phi}+\varepsilon_{\mathrm{conf}},0)$-DP
guarantee for the complete two-stage mechanism. This privacy
model is chosen deliberately because the present work considers
a single release of a manifold-valued statistical summary after
a private geometry-construction stage. Pure DP provides a
pointwise likelihood-ratio guarantee without an additional
failure probability. It also implies valid R\'enyi-DP bounds,
although these inherited bounds may not be the tightest ones for
applications involving many repeated releases. Developing a
mechanism specifically optimized for Gaussian DP or R\'enyi DP
would require a different Stage-2 distribution and a separate
privacy analysis, as discussed in Section VIII.

\subsection{Utility Analysis}
\label{subsec:utility}

\paragraph{Utility of the Sanitized Conformal Factor.}
\label{para:stage1_utility}
The first stage releases a sanitized conformal factor that deforms the underlying Riemannian geometry. We first establish the expected squared error bound for the Stage-1 mechanism in Theorem. \ref{thm:phi_error_bound}.

\begin{theorem}[Expected squared error bound for the conformal factor]
\label{thm:phi_error_bound}
Under the sanitization in Paragraph \ref{paragraph:adj}, given a dataset of size $N$, a Stage-1 privacy budget $\varepsilon_\phi$, and a regularization parameter $\upsilon$, the expected squared error between the privatized conformal factor $\phi_p$ and its non-private counterpart $\phi_{np}$ satisfies:
\begin{equation}
\label{eq:phi_error_bound}
\mathbb{E}\!\left[\|\phi_p-\phi_{np}\|^2_{L^\infty(\mathcal{M})}\right]
=
\mathcal{O}\!\left( \frac{C_K^2}{\upsilon^2 N^2\varepsilon_\phi^2}(\ln J+1)^2 \right),
\end{equation}
where $C_K := \sup_{x\in\mathcal{M}} \sum_{j=1}^J \bar{K}_{h,j}(x)$ controls the maximal overlap of the manifold kernels.
\end{theorem}

\begin{proof}
    We give a detailed proof in Appendix J.
\end{proof}

\paragraph{Utility of the Conformal Fr\'echet Mean.}
\label{para:conformal_frechrt_utility}
In the second stage, the mechanism releases a privatized Fr\'echet mean sampled under the conformal metric $g^*$. Utility is evaluated under the original metric $g$ through the expected squared geodesic error. Using the triangle inequality theorem, the total squared error decomposes into two terms: a deterministic geometric bias induced by replacing $g$ with $g^*$, and a stochastic sampling variance induced by the conformal mechanism. Before stating the utility guarantee, we fix the notation. Let $r^*$ be the radius of the geodesic ball under $g^*$. Denote by $\eta_{np}(D)$ the non-private Fr\'echet mean under $g$, and by $\eta_{g^*}(D)$ the Fr\'echet mean under $g^*$.
The Stage-2 mechanism samples $z\sim \mathbb{P}^*(\cdot\mid \eta_{g^*}(D))$ with privacy budget $\varepsilon_{\mathrm{conf}}$. We write $C_{\mathrm{stab}}(r,\kappa)=\frac{1}{h(r,\kappa)}$ for the Fr\'echet mean stability constant from Definition~\ref{def2}. Thus, we give the following theorem and its proof:

\begin{theorem}[Expected squared utility bound for the privatized mean]
\label{thm:stage2_utility}
Suppose $D$ is contained in a $g$-geodesic ball of radius $r$, and suppose the
Stage-2 conformal sensitivity satisfies $\Delta^*=\mathcal{O}\!\left(\frac{r^*}{N}\right)$.
Then the Stage-2 output $z$ satisfies
\[
    \mathbb{E}\!\left[
        \rho_g\bigl(z,\eta_{np}(D)\bigr)^2
    \right]
    \leq
    \mathcal{O}\!\left(
        \frac{C_{\mathrm{stab}}(r,\kappa)^2}{N^2\varepsilon_\phi^2}
        +
        \frac{d^2 r^2}{\phi_{\min}N^2\varepsilon_{\mathrm{conf}}^2}
    \right).
\]
\end{theorem}
\begin{proof}
A detailed proof is deferred to Appendix K.
\end{proof}

\begin{remark}
If the privacy budget is split proportionally, i.e.,
$
\varepsilon_\phi, \varepsilon_{\mathrm{conf}} = \Theta(\varepsilon),
$
then \cref{thm:stage2_utility} implies the asymptotic bound
\begin{equation*}
\begin{aligned}
\mathbb{E}\!\left[\rho_g\bigl(z,\eta_{np}(D)\bigr)^2\right]
=
\mathcal{O}\!\left( \frac{d^2 \mathbb{E}[(r^*)^2]}{\phi_{\min} N^2 \varepsilon^2} \right)
\le
\mathcal{O}\!\left( \frac{d^2 r^2}{\phi_{\min} N^2 \varepsilon^2} \right).
\end{aligned}
\end{equation*}
When $\phi_{\min}$ is treated as a geometric constant, the worst-case relaxation matching the $r^2$ bound recovers the optimal rate achieved by standard metric mechanisms on bounded spaces. More importantly, the tighter bound explicitly captures the advantage of the conformal construction. Standard mechanisms must calibrate noise using the worst-case radius $r$ under the original geometry. In contrast, under the sanitized conformal metric, high-density regions are contracted, so the relevant local radius becomes $r^* < r$. Consequently, the effective sensitivity scales as $\Delta^*=\mathcal{O}(r^*/N)$ rather than $\mathcal{O}(r/N)$. This leads to a smaller sampling variance under the same global privacy budget and formally captures the density-aware benefit of the conformal mechanism.
\end{remark}

\subsection{Privacy Budget Allocation and Optimal Trade-offs}
\label{subsec:budget_allocation}

As per Theorem \ref{thm:composition}, the total privacy budget satisfies $\varepsilon_{\mathrm{total}}=\varepsilon_\phi+\varepsilon_{\mathrm{conf}}.$
A natural design question is how to allocate this budget between the anchor mechanism $\mathcal{A}_{\mathrm{phi}}$ and the conformal mechanism $\mathcal{A}_{\mathrm{conf}}$ so as to minimize the expected error of the released point $z\in\mathcal{M}$.

We quantify utility loss by the expected geodesic error between the private output $z$ and the non-private Fr\'echet mean $\eta_{np}(D)$. Using the decomposition established in Theorem \ref{thm:stage2_utility}, this error separates into two components under the original metric $g$:
{\small
\begin{equation}
\label{eq:error_decomposition}
    \mathbb{E}\bigl[\rho_g(z,\eta_{np}(D))\bigr]
    \le
    \underbrace{\mathbb{E}\bigl[\rho_g(\eta_{g^*}(D),\eta_{np}(D))\bigr]}_{\mathbb{E}_{\mathrm{phi}}}
    +
    \underbrace{\mathbb{E}\bigl[\rho_g(z,\eta_{g^*}(D))\bigr]}_{\mathbb{E}_{\mathrm{conf}}}
\end{equation}}
The first term is the geometric bias induced by the sanitized conformal factor, while the second term is the sampling error of the conformal mechanism.

\paragraph{Geometry distortion error \texorpdfstring{$\mathbb{E}_{\mathrm{phi}}$}{E_anchor}.}
The first term reflects the displacement of the Fr\'echet mean caused by the perturbation of the conformal metric, which originates from the Laplace noise added to the anchor counts. From Theorem \ref{thm:stage2_utility}, we obtain
\begin{equation}
\label{eq:geom_error}
    \mathbb{E}_{\mathrm{phi}}(\varepsilon_\phi)
    \le
    \frac{8 C_{\mathrm{stab}}(r,\kappa) C_K (\ln J+1)}{\upsilon N\varepsilon_\phi}
    =
    \frac{C_\phi}{\varepsilon_\phi},
\end{equation}
where
$
C_K = \sup_{x\in\mathcal{M}} \sum_{j=1}^J \bar{K}_{h,j}(x)
$
is the maximal kernel overlap. Accordingly, 
$
C_\phi
=
\frac{8 C_{\mathrm{stab}}(r,\kappa) C_K (\ln J+1)}{\upsilon N}.
$

\paragraph{Conformal sampling error \texorpdfstring{$\mathbb{E}_{\mathrm{conf}}$}{E_conf}.}
The second term reflects the expected displacement introduced by sampling $z \sim \mathbb{P}^*(\cdot\mid \eta_{g^*}(D))$ under the conformal metric and then evaluating the result under the original metric. Using the Stage 2 utility bound, we write
\begin{equation}
\label{eq:samp_error}
    \mathrm{error}(\varepsilon_{\mathrm{conf}})
    \le
    \frac{2 d \Delta^*}{\varepsilon_{\mathrm{conf}}\sqrt{\phi_{\min}}}
    =
    \frac{C_{\mathrm{conf}}}{\varepsilon_{\mathrm{conf}}},
\end{equation}
where $C_{\mathrm{conf}} = \frac{2 d \Delta^*}{\sqrt{\phi_{\min}}}.$
Combining \cref{eq:geom_error,eq:samp_error}, the upper-bound the total expected error by
\[
\mathrm{error}_{\mathrm{total}}(\varepsilon_\phi,\varepsilon_{\mathrm{conf}})
=
\frac{C_\phi}{\varepsilon_\phi}
+
\frac{C_{\mathrm{conf}}}{\varepsilon_{\mathrm{conf}}},
\]

\begin{proposition}[Optimal budget allocation]
\label{prop:optimal_allocation}
Given a fixed total privacy budget $\varepsilon_{\mathrm{total}}>0$, let $(\varepsilon_\phi^{opti},\varepsilon_{\mathrm{conf}}^{opti})$ let the allocation that minimizes the upper bound be:
\[
\mathcal{E}_{\mathrm{total}}(\varepsilon_\phi,\varepsilon_{\mathrm{conf}})
=
\frac{C_\phi}{\varepsilon_\phi}
+
\frac{C_{\mathrm{conf}}}{\varepsilon_{\mathrm{conf}}}
\]
Then
{\footnotesize
\begin{align}
\label{eq:optimal_budgets}
\varepsilon_\phi^{opti}
=
\varepsilon_{\mathrm{total}}
\left(
\frac{\sqrt{C_\phi}}{\sqrt{C_\phi}+\sqrt{C_{\mathrm{conf}}}}
\right),
\varepsilon_{\mathrm{conf}}^{opti}
=
\varepsilon_{\mathrm{total}}
\left(
\frac{\sqrt{C_{\mathrm{conf}}}}{\sqrt{C_\phi}+\sqrt{C_{\mathrm{conf}}}}
\right).
\end{align}}
Moreover, the optimal allocation ratio is independent of the total budget $\varepsilon_{\mathrm{total}}$ and satisfies
\begin{equation}
\label{eq:optimal_ratio}
\frac{\varepsilon_\phi^{opti}}{\varepsilon_{\mathrm{conf}}^{opti}}
=
\sqrt{\frac{C_\phi}{C_{\mathrm{conf}}}}
=
\sqrt{
\frac{
4 C_{\mathrm{stab}}(r,\kappa)\, C_K\, (\ln J+1)\, \sqrt{\phi_{\min}}
}{
\upsilon\, N\, d\, \Delta^*
}
}.
\end{equation}
\end{proposition}

\begin{proof}
    We give detailed proof in Appendix L.
\end{proof}

\paragraph{Practical implications.}
First, the budget split is asymptotically stable with respect to the dataset size $N$. As noted in Theorem \ref{thm:stage2_utility}, the conformal sensitivity $\Delta^*$ scales on the order of $1/N$. Since the Stage-1 error term also contains a factor of $1/N$, these dependencies partially cancel out in \cref{eq:optimal_ratio}. As a result, the optimal allocation does not degenerate as the sample size increases.

Second, the number of anchors affects the Stage-1 bound only through the logarithmic factor $(\ln J+1)$. Hence, increasing the number of anchors leads to only a mild increase in the Stage-1 privacy overhead. This is useful in practice, because a finer anchor discretization can capture local geometric variation more accurately without substantially degrading the overall privacy-utility tradeoff.

Third, the conformal geometry enters the allocation through $\phi_{\min}$, $d$, and $\Delta^*$. Smaller $\phi_{\min}$, larger intrinsic dimension $d$, or larger conformal sensitivity $\Delta^*$ all increase the cost of Stage~2 and therefore shift more budget toward the conformal sampling stage. Conversely, stronger regularization through $\upsilon$ or greater stability of the Fr\'echet mean under metric perturbations reduce the Stage-1 cost and allow a larger fraction of the budget to be allocated to Stage 2.
In summary, \cref{eq:optimal_ratio} provides a principled initialization rule for privacy budget allocation across different manifold geometries and sampling regimes.

\paragraph{Parameter robustness and practical selection.}
The anchor discretization and the PDE regularization parameter
affect the mechanism in different ways. The number of anchors
$J$ controls the spatial resolution of the reconstructed density.
Increasing $J$ can reduce geometric quantization error and capture
finer density variations, but it also increases the reconstruction
cost and changes the kernel-overlap constant $C_K$. Accordingly,
the Stage-1 error bound depends on the anchor construction through
$C_K^2(\log J+1)^2$. In practice, the anchors are constructed using
a geometry-based covering method, such as an $\epsilon$-net or a
centroidal Voronoi tessellation, and $J$ is selected according to
the desired geometric resolution and computational budget.

The parameter $\nu$ controls the strength and stability of the
conformal deformation. A smaller $\nu$ produces a stronger
density-dependent deformation, but also increases the amplification
of density-estimation error and may reduce numerical stability.
A larger $\nu$ improves the conditioning of the Helmholtz--Poisson
operator, but produces a weaker conformal transformation. We select
$\nu$ using the spectral and geometric rule in Eqs.~(10)--(11).
These relations provide practical guidance for choosing $J$ and
$\nu$ without assuming that the output is completely insensitive
to their values.

\section{Algorithm Design}
\label{sec: section6}
In this section, we give a detailed algorithm design for our Conformal-DP mechanism in \cref{tab:walkthrough}. We also analyze the theoretical computational complexity for performing the Kernel Density Estimation process, solving PDEs on the manifold, and the MCMC sampling process.
\begin{table*}[t!]
    \centering
    \small
    \renewcommand{\arraystretch}{1.1}
    \caption{Algorithm design of the Conformal-DP mechanism (based on the descriptions in Section \ref{section3} and Section \ref{section4}).}
    \label{tab:walkthrough}
    \begin{tabular}{m{115pt}m{175pt}m{180pt}}
        \toprule
        \textbf{Theoretical Basis} & \textbf{Algorithm Design} & \textbf{Notes} \\
        \midrule

        \multicolumn{3}{l}{\textit{Stage 1: Differentially Private Conformal Factor Release (budget $\varepsilon_\phi$)}}\\
        \midrule
    
        \parbox[t]{115pt}{\vspace{2pt}Section \ref{section3.1}, \eqref{eq:sensitivity-constant}}
        & \parbox[t]{175pt}{\vspace{2pt}Map data to Voronoi cells $V_j$ via public anchors $\mathcal{P}=\{p_j\}$ to get counts:$c_j = \sum_{i=1}^N \mathbb{I}(x_i \in V_j).$}
        & \parbox[t]{180pt}{\vspace{2pt}Reduces the continuous query to the discrete count vector $c=(c_1,\dots,c_J)$, where adjacency implies $\|c(D)-c(D')\|_1\le \Delta_1=2$.} \\
    
        \parbox[t]{115pt}{\vspace{2pt}Section \ref{section3.1}, equations for $\tilde{c}_j$, \eqref{eq:manifold_adapted_kernel}, \eqref{eq:reconstructed_kde}}
        & \parbox[t]{175pt}{\vspace{2pt}Inject Laplace noise:
        $
        \tilde{c}_j = c_j + \mathrm{Lap}\left(\frac{\Delta_1}{\varepsilon_\phi}\right).
        $
        Truncate the noisy counts to $\tilde{c}_j^+ = \max(0, \tilde{c}_j)$, normalize the truncated counts into probability weights $w_j$, and reconstruct the sanitized density using the pointwise normalized manifold-adapted kernel:$\tilde{f}_{\mathrm{data}}(x) = \sum_{j=1}^J w_j \bar{K}_{h,j}(x).$}
        & \parbox[t]{180pt}{\vspace{2pt}Laplace perturbation on $\tilde{c}_j=c_j+\mathrm{Lap}(\Delta_1/\varepsilon_\phi)$ provides $\varepsilon_\phi$-DP; $\tilde{c}_j^+$, $w_j$, and $\tilde{f}_{\mathrm{data}}$ are protected by post-processing immunity.} \\
    
        \parbox[t]{115pt}{\vspace{2pt}Section \ref{subsec:eps_phi}, \eqref{eq:poisson}}
        & \parbox[t]{175pt}{\vspace{2pt}Solve the Helmholtz--Poisson equation using $\tilde{f}_{\mathrm{data}} - \tilde{f}_{\mathrm{max}}$:
        $
        (-\Delta_g+\upsilon)\,\sigma(x) = \tilde{f}_{\mathrm{data}}(x) - \tilde{f}_{\mathrm{max}}.
        $
        State $\phi(x) = e^{2\sigma(x)}$.}
        & \parbox[t]{180pt}{\vspace{2pt}The public parameter $\upsilon$ makes $(-\Delta_g+\upsilon)$ strictly positive; solving for $\sigma$ gives $\phi=e^{2\sigma}$ with $0<\phi_{\min}\le \phi(x)\le 1$.} \\
    
        \midrule
        \multicolumn{3}{l}{\textit{Stage 2: Conformal-DP Fr\'echet Mean Release on $\mathcal{M}, \, g^*=\phi\,g$ (budget $\varepsilon_{\text{conf}}$)}}\\
        \midrule
    
        \parbox[t]{115pt}{\vspace{2pt}Section \ref{subsec:privacy_calibration}, \eqref{eq:conformal_energy}, \eqref{eq:conformal_frechet_mean}}
        & \parbox[t]{175pt}{\vspace{2pt}Define the Conformal Fr\'echet mean by minimizing the conformal energy $F^*(x)$ under $g^*$:
        $
        F^*(x) = \frac{1}{2N}\sum_{i=1}^N \rho_{g^*}^2(x,x_i),
        $
        $
        \eta_{g^*}(D) = \arg\min_{x\in\mathcal{M}} F^*(x).
        $}
        & \parbox[t]{180pt}{\vspace{2pt}The target statistic is $\eta_{g^*}(D)=\arg\min_x F^*(x)$ under $g^*=\phi g$, so density-awareness is induced through $\rho_{g^*}$ rather than explicit spatial re-scaling of noise.} \\
    
        \parbox[t]{115pt}{\vspace{2pt}Sections \ref{subsec:conformal_sensitivity}, \ref{subsec:privacy_calibration}, \eqref{eq:conformal_sensitivity}, \eqref{eq:tight_sensitivity_bound}, \eqref{eq:lambda_calibration}}
        & \parbox[t]{175pt}{\vspace{2pt}State the bound $\Delta^* \le \mathcal{S}^*$:
        $
        \mathcal{S}^* = C_{\tau,\theta}\left(\frac{2r}{N}\right) + \mathcal{K}_{\mathrm{curv}},
        $
        and calibrate $\lambda^* = \frac{\varepsilon_{\text{conf}}}{2\mathcal{S}^*}$.}
        & \parbox[t]{180pt}{\vspace{2pt}The sensitivity bound $\Delta^*\le \mathcal{S}^*$ separates the leading radius term $C_{\tau,\theta}(2r/N)$ from the curvature correction $\mathcal{K}_{\mathrm{curv}}$.} \\
    
        \parbox[t]{115pt}{\vspace{2pt}Sections \ref{sec:conformal_measure}, \ref{subsec:privacy_calibration}, \eqref{eq:conformal_pdf}, \eqref{eq:formal_dp_bound}}
        & \parbox[t]{175pt}{\vspace{2pt}Release $Z \sim \mathbb{P}^*(\cdot \mid \eta_{g^*}(D), \lambda^*)$ proportional to:
        $
        \exp\left\{-\lambda^* \rho_{g^*}(\eta_{g^*}(D), z)\right\}.
        $}
        & \parbox[t]{180pt}{\vspace{2pt}Sampling from $\mathbb{P}^*(\cdot\mid \eta_{g^*}(D),\lambda^*)$ with $\lambda^*=\varepsilon_{\mathrm{conf}}/(2\mathcal{S}^*)$ gives pure $\varepsilon_{\mathrm{conf}}$-DP; the output $Z\in\mathcal{M}$ requires no pull-back. For the MCMC sampling implementation refer to \cite{reimherr2021differential}, section MCMC.} \\
    
    \bottomrule
\end{tabular}
\end{table*}

\subsection{Complexity Analysis}
\paragraph{Complexity of anchor-based kernel density estimation.}
With a spatial index constructed on the $J$ public anchors, assigning $N$ data points to their nearest anchors requires $\mathcal{O}(N\log J)$ time. Adding Laplace noise to the anchor counts requires $\mathcal{O}(J)$ time. Evaluating the reconstructed density
$
\tilde{f}_{\mathrm{data}}(x)=\sum_{j=1}^J w_j \bar{K}_{h,j}(x)
$
on a PDE mesh with $M$ grid points requires $\mathcal{O}(MJ)$ operations. Therefore, after anchor indexing and kernel normalization have been prepared offline, the online complexity of Stage~1 is
$
\mathcal{O}(N\log J + MJ).
$
This separates the privacy-sensitive computation from the continuous functional representation and makes the first stage scalable in the dataset size $N$.

\paragraph{Complexity of solving the Helmholtz--Poisson equation.}
To obtain the conformal factor $\phi=e^{2\sigma}$, we solve
$
(-\Delta_g+\upsilon)\sigma=\tilde{f}_{\mathrm{data}}-\tilde{f}_{\mathrm{max}}.
$
Discretizing the Laplace--Beltrami operator on an $M$-vertex mesh yields a sparse linear system. The cost of solving this system depends on the solver: sparse direct methods on typical 2D/3D meshes often require about $\mathcal{O}(M^{1.5})$ time in practice, while iterative solvers such as conjugate gradient can be near-linear up to conditioning and preconditioning quality. In either case, this stage is independent of the dataset size $N$, so its online cost is $\mathcal{O}(1)$ with respect to $N$.

\paragraph{Complexity of conformal MCMC sampling.}
After computing the conformal Fr\'echet mean
$\eta_{g^*}(D)$ by minimizing Eq. (29), we use the Markov Chain
Monte Carlo only to draw the final private output from the
conformal-Laplace distribution in Eq. (25). Thus, the MCMC
procedure is part of the private randomization step rather than
the computation of the conformal Fr\'echet mean. Starting from the current state $x_{\mathrm{curr}}\in\mathcal{M}$, we generate a proposal $x_{\mathrm{prop}}$ locally by drawing a Gaussian perturbation in the tangent space and mapping it back to the manifold. The symmetric proposal reduces the Metropolis--Hastings acceptance probability to
$\alpha(x_{\mathrm{curr}},x_{\mathrm{prop}}) = \min\!\left( 1,\, \exp\!\left\{ \lambda^* \bigl( \rho_{g^*}(x_{\mathrm{curr}},\eta_{g^*}) - \rho_{g^*}(x_{\mathrm{prop}},\eta_{g^*}) \bigr) \right\} \right).$
Assuming the conformal distances needed by the chain are precomputed or can be queried in constant time on the $M$-vertex graph, the online complexity of sampling for $T$ iterations is $\mathcal{O}(T)$. We discuss the selection of $T$ in Appendix M.

\section{Experimental Results}
\label{sec:experiment}
We evaluate our mechanism on two representative manifold settings: the sphere with positive curvature and the symmetric positive-definite (SPD) manifold with negative curvature. Importantly, although the latter is a \emph{non-compact} manifold, we can truncate the manifold by calculating the maximum distance between \emph{any} pair of possible inputs as the diameter for the public geodesic ball when the data allow (see Sec. \ref{ssec:real_world}). For synthetic experiments, we generate heterogeneously distributed samples from a von Mises-Fisher (vMF) distribution \cite{hillen2017moments} on the hypersphere to isolate the effects of manifold dimension, local density heterogeneity, and stringent privacy budgets on privatized Fr\'echet mean estimation. For real-world evaluation, we transform \textbf{Fashion-MNIST} and \textbf{CIFAR-10} images into covariance descriptors on the SPD manifold and benchmark our mechanism against state-of-the-art baselines. Across a range of privacy budgets, sample sizes, sample dimensions, and heterogeneity, our method consistently achieves better utility, scalability, and stability.
We run all experiments on an AMD EPYC 9575F CPU, with 32GB memory, single-core processing out of 12 cores, and test on the Jupyter Lab environment with Python version 3.11.0. 

\subsection{Results on Synthetic Datasets}
\label{sec:synthetic}

\paragraph{Synthetic dataset on the sphere (vMF sampling).}
Let $\mathcal{M}_{\mathrm{syn}}=\mathcal{S}^{d-1}(R)\subset\mathbb{R}^d$ be the sphere of radius $R$. 
We generate heterogeneously distributed points $x_i\in\mathcal{M}_{\mathrm{syn}}$ from a von Mises--Fisher (vMF) distribution \cite{hillen2017moments} with mean direction $\mu\in\mathbb{R}^d$ ($\|\mu\|_2=1$) and concentration $\kappa\ge 0$, following the acceptance--rejection sampler of Wood~\cite{Wood1994Simulation}, which is based on Ulrich’s hyperspherical rejection-sampling construction \cite{ulrich1984computer}. We use the standard deviation (std) parameter to control the data distribution, for example, when std = 0.1, the data distribution is extremely heterogeneous, and when std = 20, the data distribution is nearly uniform.

\paragraph{Mechanism.} \label{sec:details} We compare our mechanism to \textbf{three} existing state-of-the-art works: (a) the Riemannian-Laplace DP \cite{reimherr2021differential}, which injects Laplacian noise to the Fr{\'e}chet mean, (b) the K-Norm gradient-based DP mechanism \cite{soto2022shape}, which treats the Fr{\'e}chet mean as a minimizer of an objective function and injects noises to gradient updates in solving for the mean and (c) the Tangent-Gaussian DP \cite{utpala2022differentially}, which injects Gaussian noise in the tangent space to achieve $(\varepsilon,\delta)$-DP. \emph{The implementation details and hyperparameters are provided in Appendix \ref{app:alg-detail}}.

\paragraph{Utility Metrics.}
We use the geodesic distance between the privatized and non-private Fr{\'e}chet mean $\rho \left( \eta_p(D), \eta(D) \right)$ to compare the utility among the mechanisms. A smaller distance over the same privacy budget $\varepsilon$ indicates a smaller amount of noise needed to achieve the same privacy guarantee, which also means stronger utility for the mechanism. We run each experiment for \textbf{200} iterations (to avoid sample fluctuations caused by single-experiment bias) and include the standard deviation over those runs as the error band; a small error band means a more stable mechanism. Note that the baseline mechanism is Riemannian-Laplace DP \cite{reimherr2021differential}, which is pure $\varepsilon$-DP; Tangent Gaussian DP \cite{utpala2022differentially} is ($\varepsilon,\delta$)-DP, which has a small relaxation budget; our mechanism is ($\varepsilon_{\phi} + \varepsilon_{conf}$)-DP, which has tighter budgets than the other two mechanisms.

\begin{figure}[t!]
    \centering
    \begin{subfigure}[b]{0.49\linewidth}
        \includegraphics[width=\linewidth]{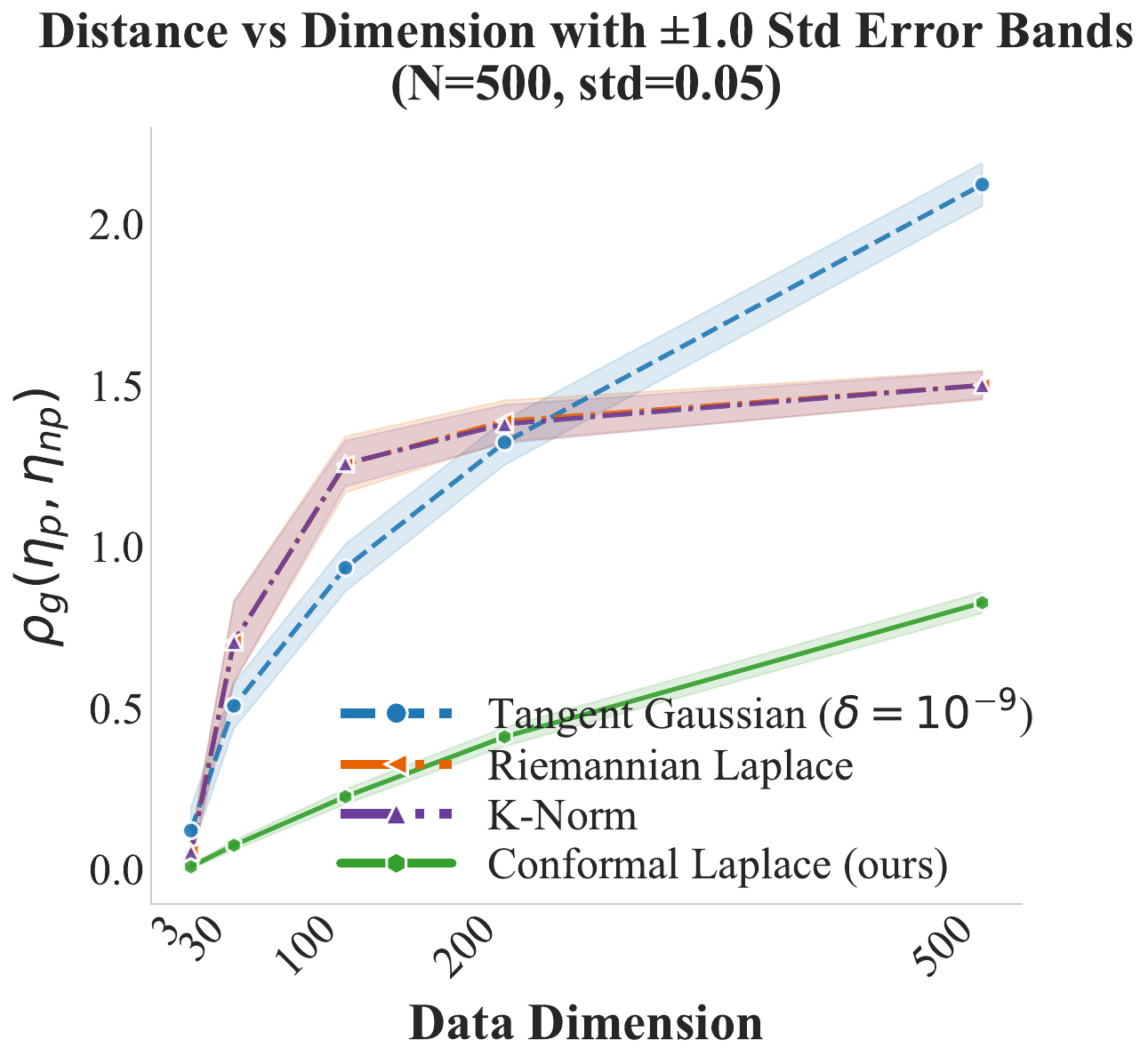}
        \caption{$\mathrm{std} = 0.05$}
    \end{subfigure}
    \begin{subfigure}[b]{0.49\linewidth}
        \includegraphics[width=\linewidth]{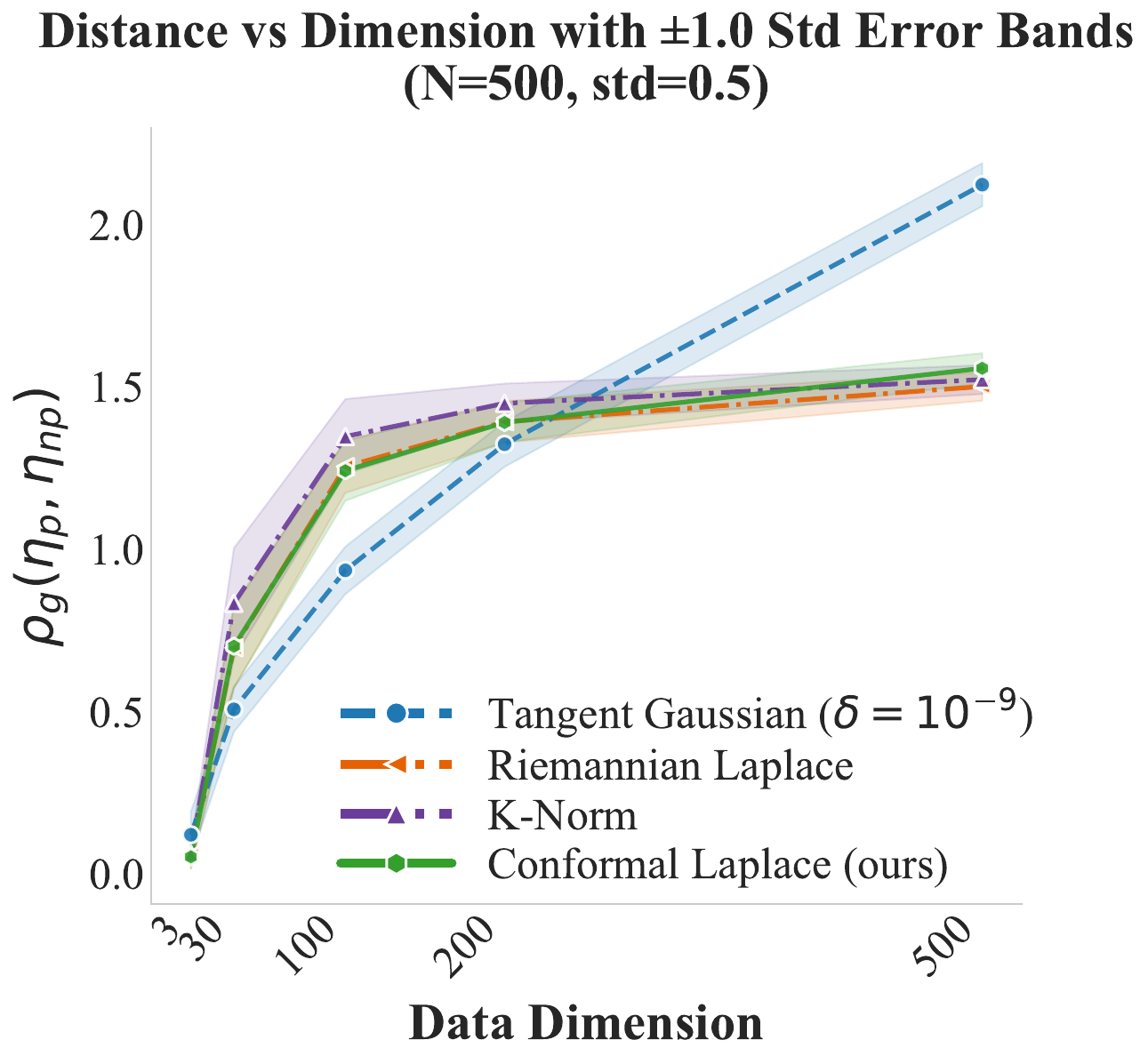}
        \caption{$\mathrm{std} = 0.5$}
    \end{subfigure}
    \caption{\textit{
    Comparison of the utility of all mechanisms under varying data dimensions. We report the geodesic distance between private and non-private Fréchet means (mean $\pm$ one standard deviation) with $N=500$ and $\varepsilon_{\mathrm{total}} = 0.3$. 
    (a) Non-uniform distribution regime: ($\mathrm{std}=0.05$). 
    (b) Near uniform distribution regime ($\mathrm{std}=0.5$). 
    }}
    \vspace{-10pt}
    \label{fig:syn-dim-comp}
\end{figure}

\begin{figure}[t!]
    \centering
    \begin{subfigure}[b]{0.49\linewidth}
        \includegraphics[width=\linewidth]{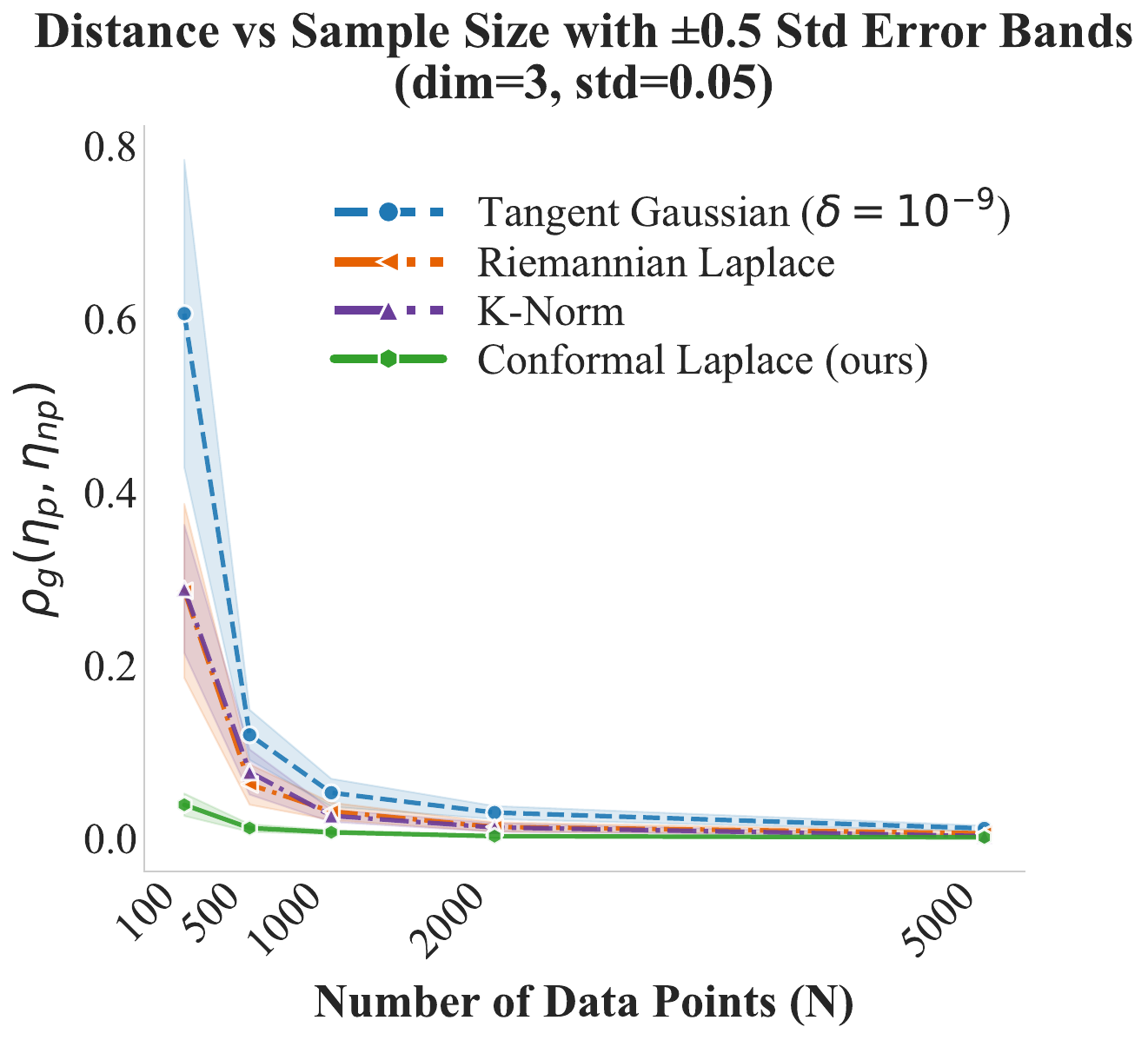}
        \caption{$d=3$, $\mathrm{std}=0.05$}
    \end{subfigure}
    \begin{subfigure}[b]{0.49\linewidth}
        \includegraphics[width=\linewidth]{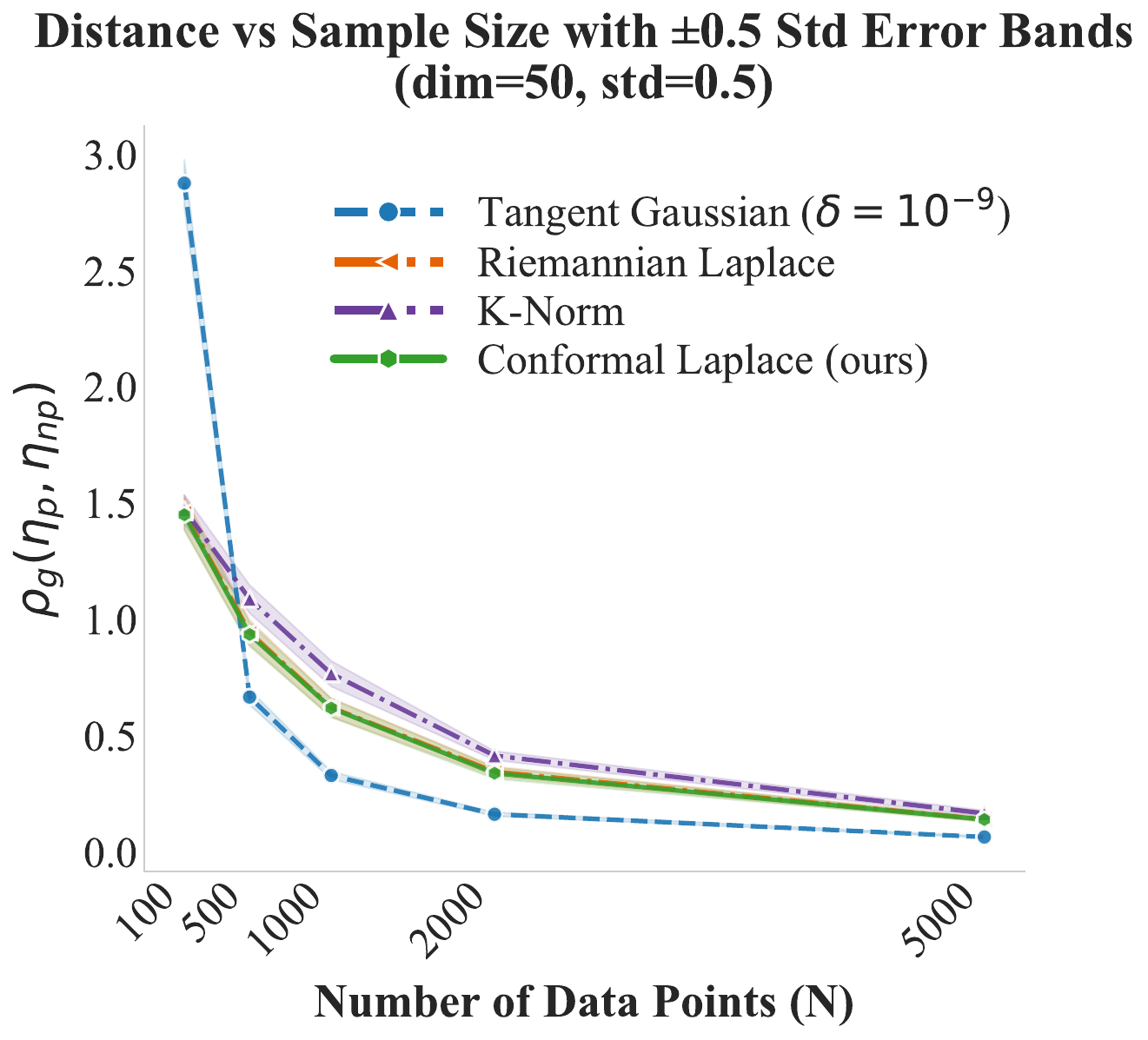}
        \caption{$d=50$, $\mathrm{std}=0.5$}
    \end{subfigure}
    \caption{\textit{
    Comparison of utility under varying sample sizes for all mechanisms.
    The y-axis reports the geodesic distance between private and non-private Fr\'echet means (mean $\pm$ 0.5 standard deviation), with $\varepsilon_{\mathrm{total}} = 0.3$. 
    (a) $d=3$, $\mathrm{std}=0.05$. 
    (b) $d=50$, $\mathrm{std}=0.5$. 
    }}
    \vspace{-15pt}
    \label{fig:syn-N-comp}
\end{figure}

\begin{figure}[t!]
    \centering
    \begin{subfigure}[b]{0.49\linewidth}
        \includegraphics[width=\linewidth]{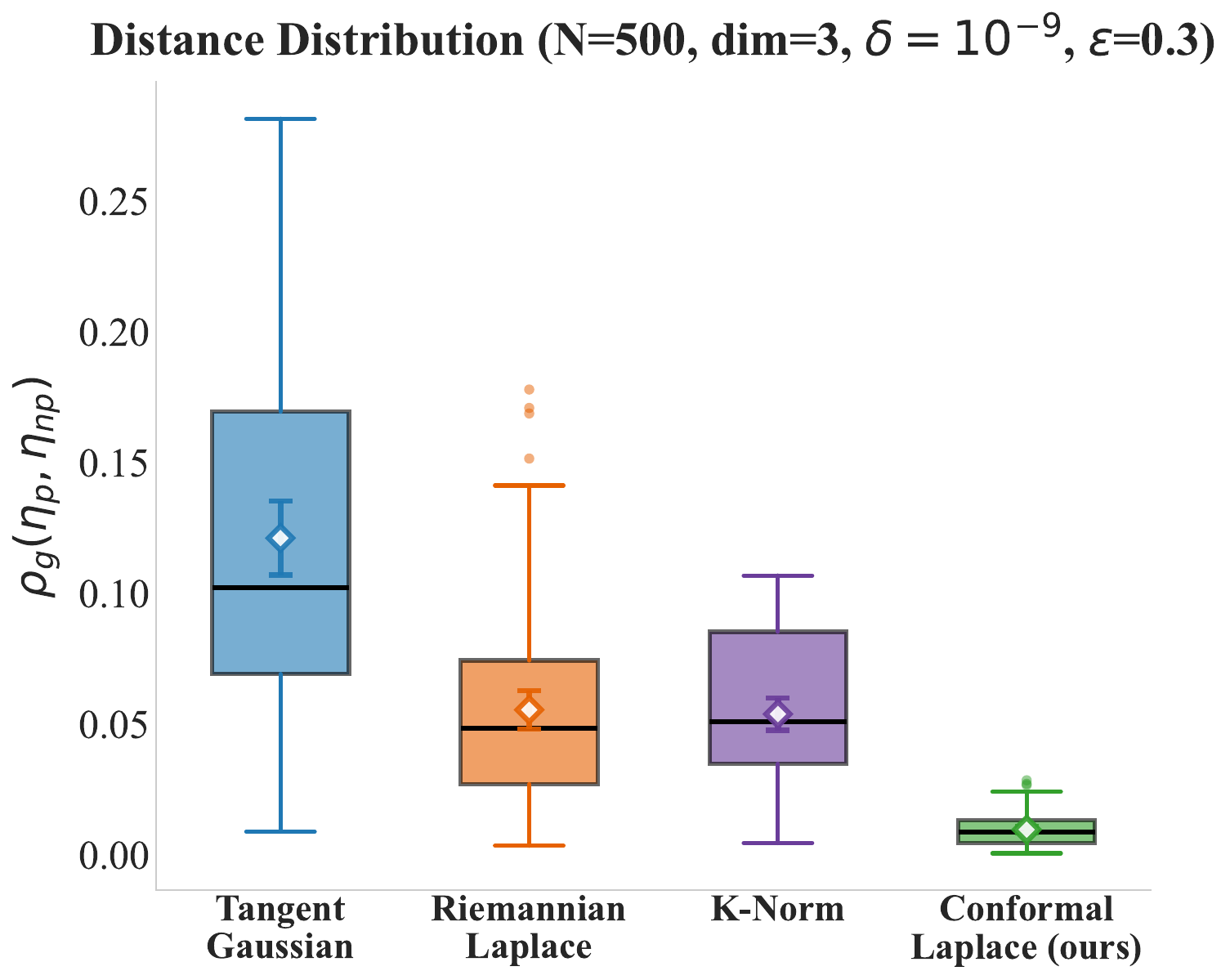}
        \caption{$\mathrm{std} = 0.05$}
    \end{subfigure}
    \begin{subfigure}[b]{0.49\linewidth}
        \includegraphics[width=\linewidth]{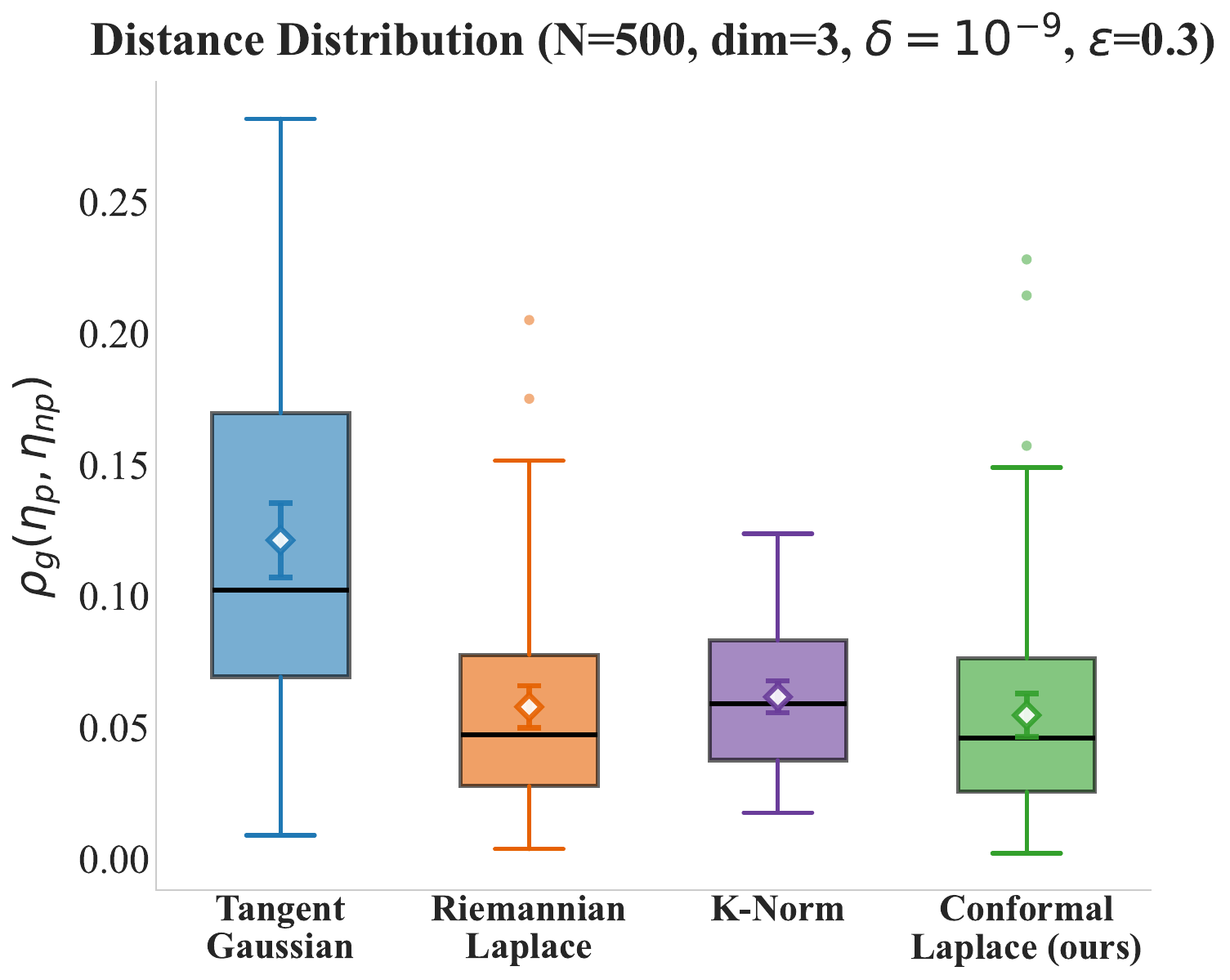}
        \caption{$\mathrm{std} = 0.5$}
    \end{subfigure}
    \caption{\textit{
    Comparison of error distributions for all mechanisms under varying levels of data distribution. We report box plots of the geodesic distance between private and non-private Fr\'echet means over repeated trials, with fixed privacy budget $\varepsilon_{\mathrm{total}}=0.3$. 
    Conformal Laplace achieves the lowest median error and the most concentrated distribution in the non-uniform distribution setting, and remains competitive in the uniform distribution setting, while Tangent Gaussian exhibits substantially larger error and variability.}}
    \label{fig:syn-std-comp}
\end{figure}
\begin{figure}[t!]
    \centering
    \begin{subfigure}[b]{0.49\linewidth}
        \includegraphics[width=\linewidth]{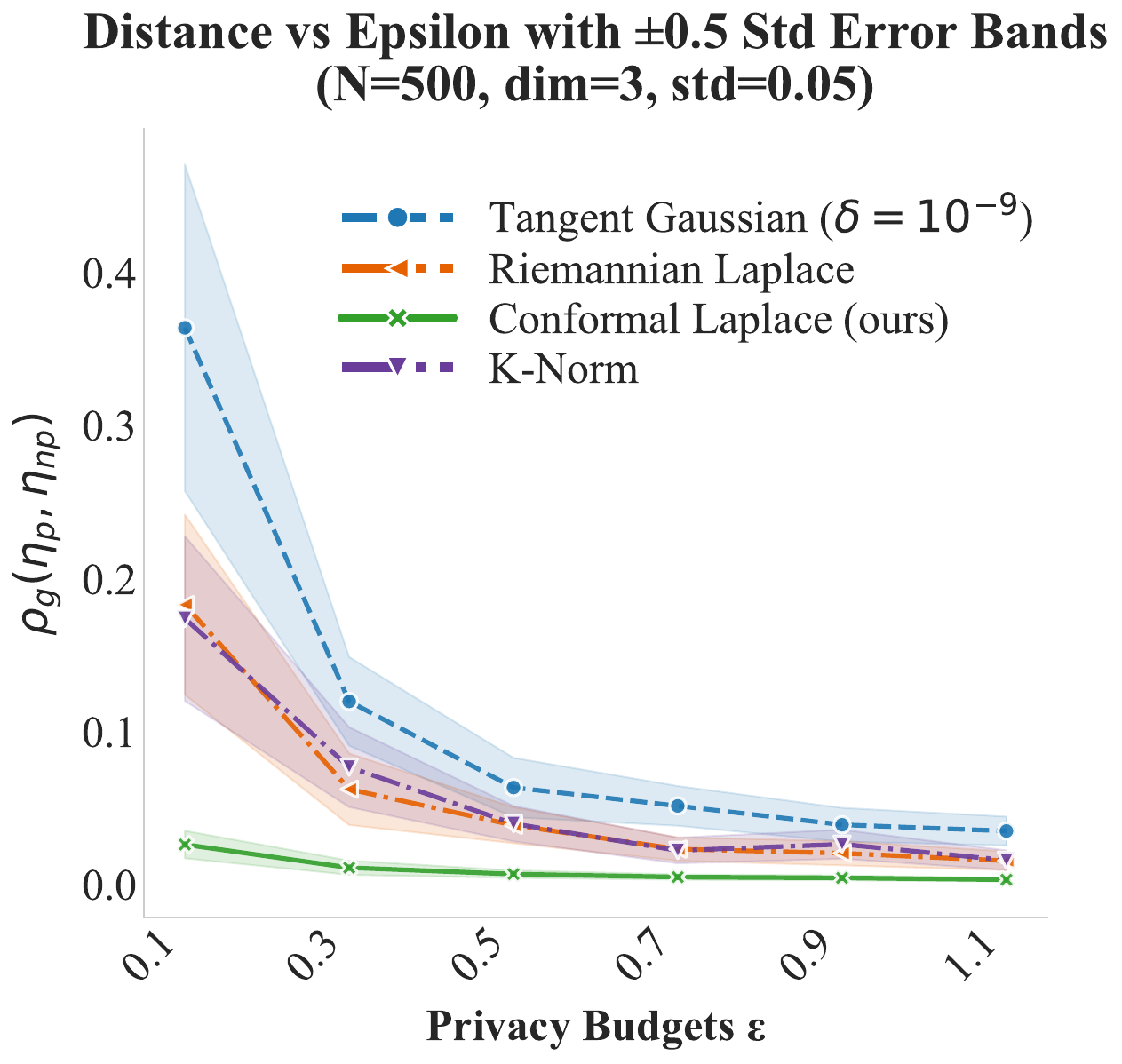}
        \caption{$\mathrm{std}=0.05$}
    \end{subfigure}
    \begin{subfigure}[b]{0.49\linewidth}
        \includegraphics[width=\linewidth]{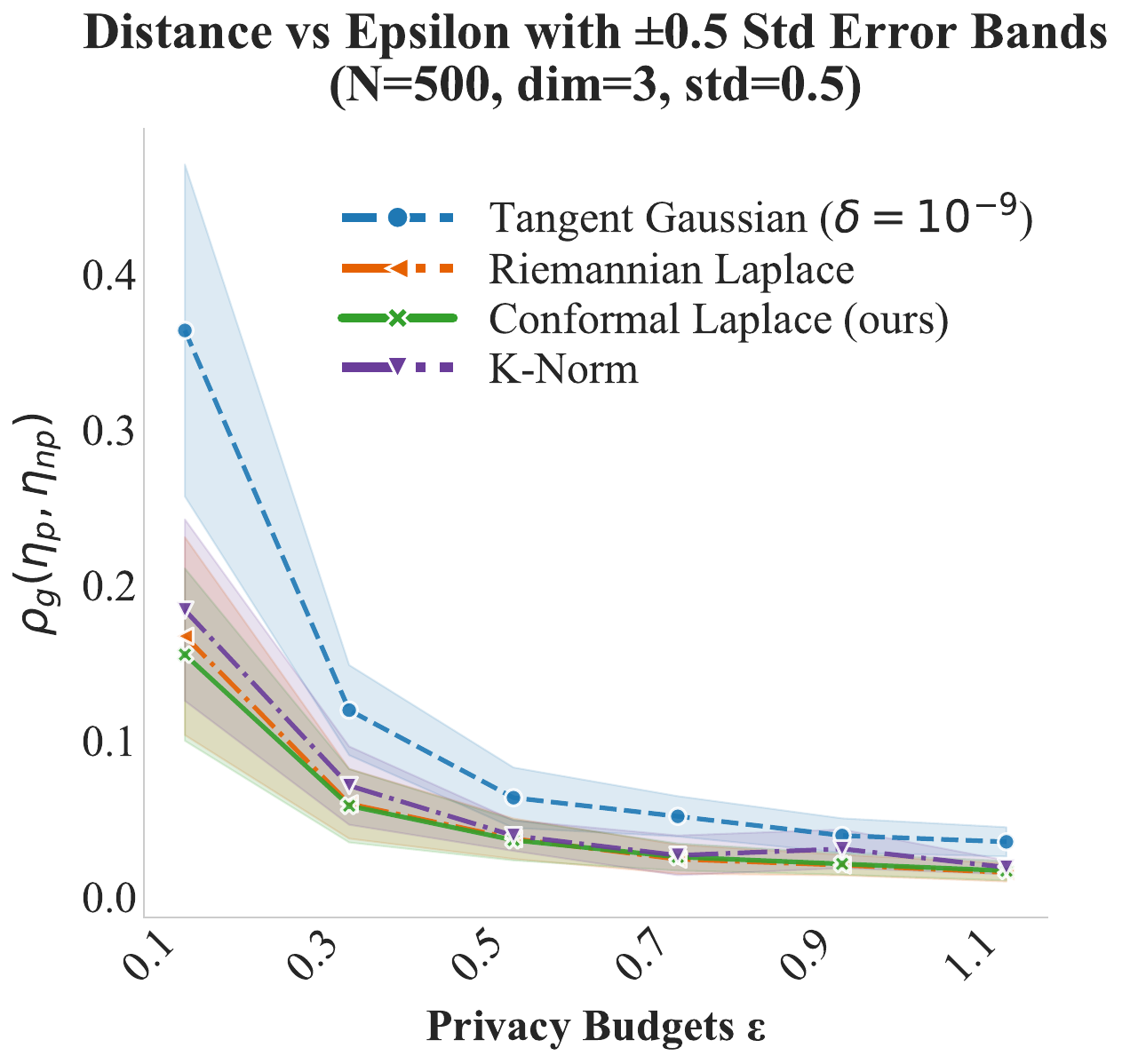}
        \caption{$\mathrm{std}=0.5$}
    \end{subfigure}
    \caption{\textit{
    Utility comparison under varying privacy budgets for all mechansims.
    We report the geodesic distance between private and non-private Fr\'echet means (mean $\pm$ 0.5 standard deviation) with $N=500$. 
    (a) Highly-heterogeneity setting ($d=3$, $\mathrm{std}=0.05$). 
    (b) Uniform setting ($d=3$, $\mathrm{std}=0.5$). 
    As the privacy budget increases, the error decreases for all methods. Conformal Laplace shows improved performance in highly-heterogeneity setting but performs similarly to the Riemannian Laplace mechanism in the uniform setting.
    }}
    \label{fig:syn-eps-comp}
\end{figure}

\paragraph{Results.} We evaluate the proposed mechanism rigorously, and present its advantages against state-of-the-art approaches from four complementary perspectives: (1) utility under varying data dimensionality and heterogeneity with sample size and privacy budgets fixed; (2) utility under varying sample size and heterogeneity with data dimension and privacy budgets fixed; (3) utility across different heterogeneity levels with data dimensionality, sample size, and privacy budget fixed; and (4) utility across different privacy budgets with data dimensionality, sample size, and heterogeneity fixed.

\cref{fig:syn-dim-comp}(a) shows that Conformal Laplace mechanism consistently attains the lowest error, while Tangent Gaussian degrades quickly with dimension due to accumulated tangent-space approximation error. Riemannian Laplace and K-Norm remain more stable but yield consistently higher error, particularly in lower-dimensional regimes. This demonstrates that the conformal geometry better captures data heterogeneity and improves robustness in higher dimensions. \cref{fig:syn-dim-comp}(b) validates our theoretical analysis: under nearly uniform data distributions, the benefit of conformal transformation vanishes, and Conformal Laplace behaves similarly to Riemannian Laplace and K-Norm.

\cref{fig:syn-N-comp} illustrates the effect of sample size $N$ on utility under a fixed privacy budget $\varepsilon_{\mathrm{total}} = 0.3$. As $N$ increases, sensitivity of the empirical Fr\'echet mean decreases with $1/N$, therefore, the error of all mechanisms decreases. In \cref{fig:syn-N-comp}(a), Conformal Laplace consistently achieves the lowest error and better utility across all sample sizes. In \cref{fig:syn-N-comp}(b), we also validate our theoretical analysis as discussed above.

\cref{fig:syn-std-comp} compares the performances between uniform distribution and non-uniform distribution of the four mechanisms under a fixed privacy budget $\varepsilon_{\mathrm{total}}=0.3$. In both settings, Conformal Laplace yields a markedly smaller error than Tangent Gaussian and a more concentrated distribution overall. In \cref{fig:syn-std-comp}(a), it achieves the lowest median error and the tightest spread. In \cref{fig:syn-std-comp}(b), its performance remains close to that of Riemannian Laplace, and both substantially outperform Tangent Gaussian. These results indicate that Conformal Laplace maintains stronger robustness under changes in data distribution.

\begin{figure}[t!]
    \centering
    \includegraphics[width=\linewidth]{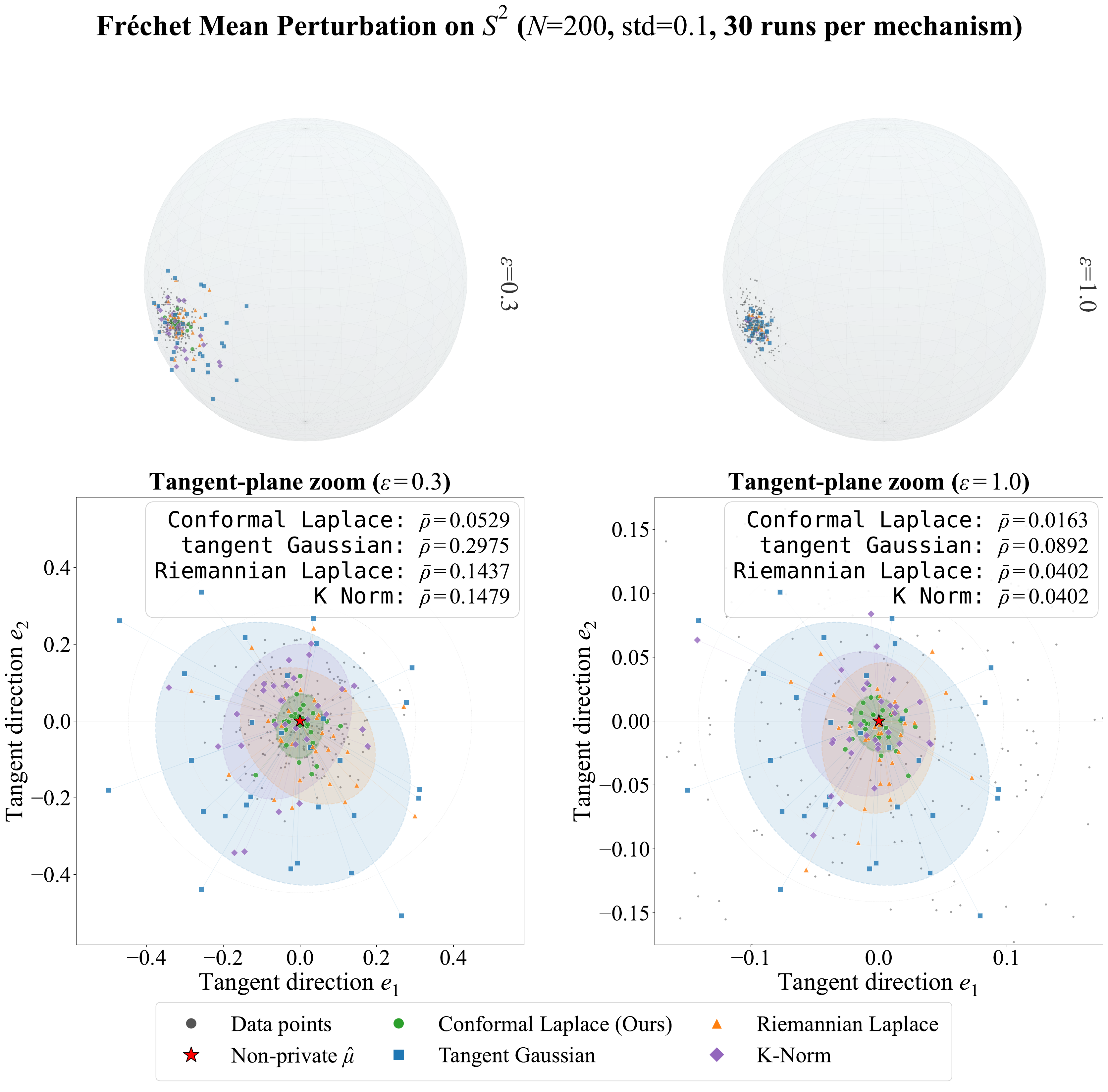}
    \caption{\textit{
    3D visualization of Fr\'echet mean perturbations on $\mathbb{S}^2$ for sphere-valued data with $N=200$ and $\mathrm{std}=0.1$ under $\varepsilon=0.3$ and $\varepsilon=1.0$. Light-gray dots denote the original data points, and the red star denotes the non-private Fr\'echet mean. The private Fr\'echet means obtained from 30 independent runs are shown for each mechanism: Conformal Laplace (green circles), Tangent Gaussian (blue squares), Riemannian Laplace (orange triangles), and K-Norm (purple diamonds). The bottom row presents tangent-plane projections centered at the non-private Fr\'echet mean, with ellipses indicating empirical dispersion regions. Conformal Laplace (green) exhibits a tighter concentration around the non-private mean than Tangent Gaussian (blue), Riemannian Laplace (orange), and K-Norm (purple).
    }}
    \label{fig:privacy_heatmap}
    \vspace{-10pt}
\end{figure}

\cref{fig:syn-eps-comp} illustrates the dependence of utility on the privacy budget $\varepsilon_{\mathrm{total}}$. As $\varepsilon_{\mathrm{total}}$ increases, the error of all mechanisms decreases, reflecting the standard privacy–utility trade-off. Conformal Laplace consistently outperforms Tangent Gaussian, K-Norm, and Riemannian Laplace in non-uniform settings while remaining competitive with Riemannian Laplace and K-Norm in uniform settings across all privacy levels. 

\cref{fig:privacy_heatmap} illustrates the qualitative and quantitative behavior of different private Fr\'echet mean mechanisms under a heterogeneous distribution on $\mathbb{S}^2$. The key observation is that the proposed Conformal Laplace mechanism produces private means that remain more tightly concentrated around the non-private Fr\'echet mean than the baseline mechanisms, especially in the stronger privacy regime. This supports the intended effect of the conformal transformation: by incorporating the sanitized density structure into the metric, the perturbation becomes more localized around high-density regions while still satisfying the same $(\varepsilon,0)$-privacy guarantee. In the tangent-plane projections, Conformal Laplace shows the smallest dispersion region and the lowest mean geodesic error under both $\varepsilon=0.3$ and $\varepsilon=1.0$. At $\varepsilon=0.3$, its mean perturbation is reduced by approximately $80\%$, $59\%$, and $60\%$ compared with Tangent Gaussian, Riemannian Laplace, and K-Norm, respectively.

In addition, we calculate the time cost (though it varies depending on the hardware) of our proposed Conformal-Laplace DP mechanism compared to the Riemannian-Laplacian DP \cite{reimherr2021differential} and tangent Gaussian DP \cite{utpala2022differentially}. On average, our Conformal DP mechanism requires around 6ms per data sample to privatize the result. The runtime increases with the dataset size N, primarily due to the sample-to-anchor assignment and the repeated evaluation of the conformal Fr\'echet objective. For a fixed manifold discretization, the cost of solving the Helmholtz–Poisson equation is determined by the mesh size M and is independent of N (Eq. \ref{eq:poisson}). Under $N = 500$, we observe approximately 3.8 times the runtime of the Riemannian-Laplace mechanism, and around 2610 times less than the K-Norm mechanism; the Tangent-Gaussian mechanism is the fastest but has the least utility on the Sphere manifold. In addition, our proposed mechanism achieves significantly stronger utility than all four mechanisms (see Figs. \ref{fig:syn-dim-comp} through \ref{fig:syn-eps-comp}).

\subsection{Results on Real-world Datasets}
More details on each setup are provided in Appendix \ref{app:alg-detail}.
\label{ssec:real_world}
\begin{figure}[t!]
    \centering
    \begin{subfigure}[b]{\linewidth}
        \includegraphics[width=\linewidth]{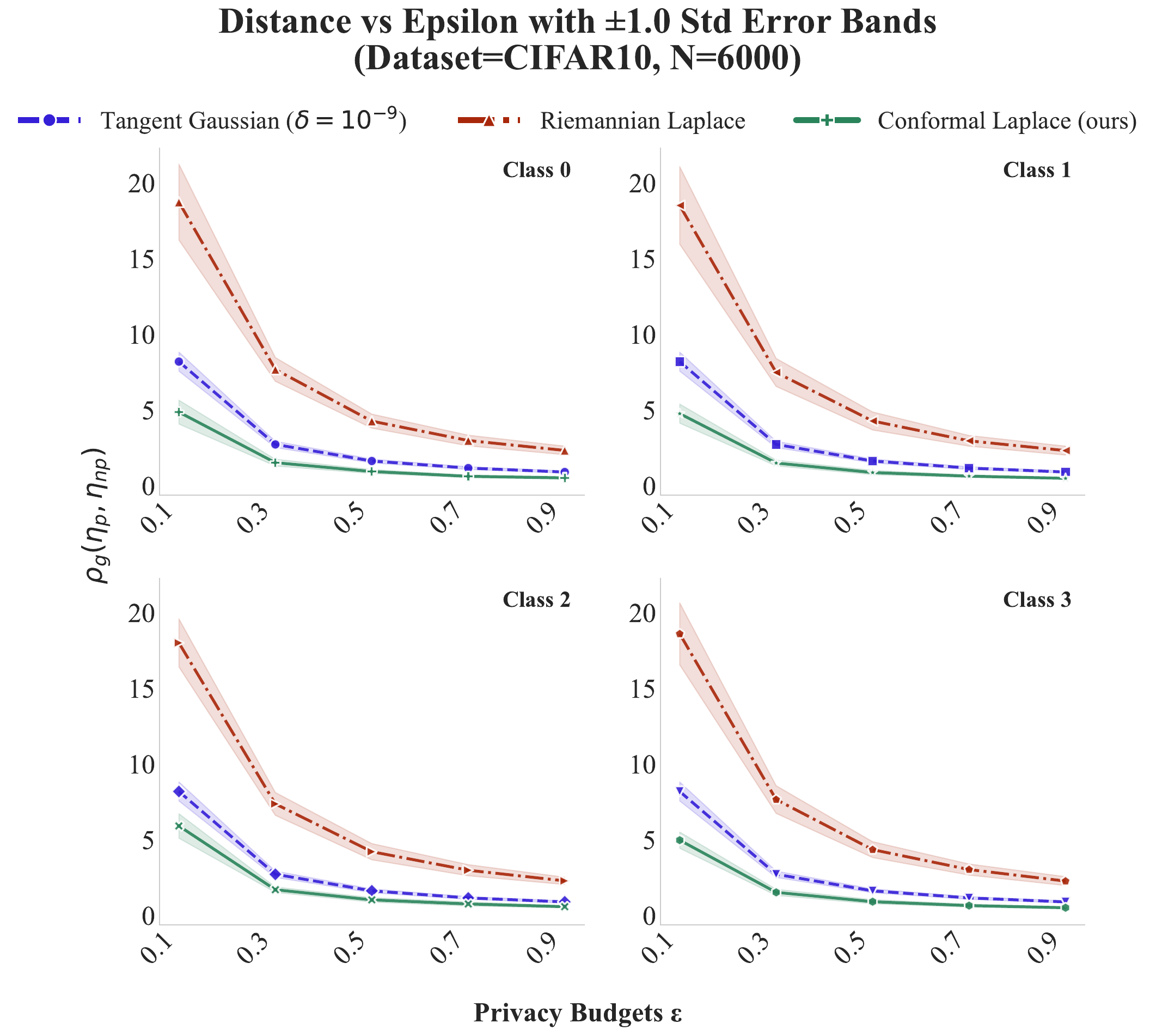}
        \label{fig:cifar10-6000}
    \end{subfigure}
    \caption{\textit{Utilities of private Fr\'echet means under varying privacy budgets ($\varepsilon$) and different classes for CIFAR-10 dataset.}}
    \label{fig:cifar10-comp}
\end{figure}

\begin{figure}[t]
    \centering
    \begin{subfigure}[b]{\linewidth}
        \includegraphics[width=\linewidth]{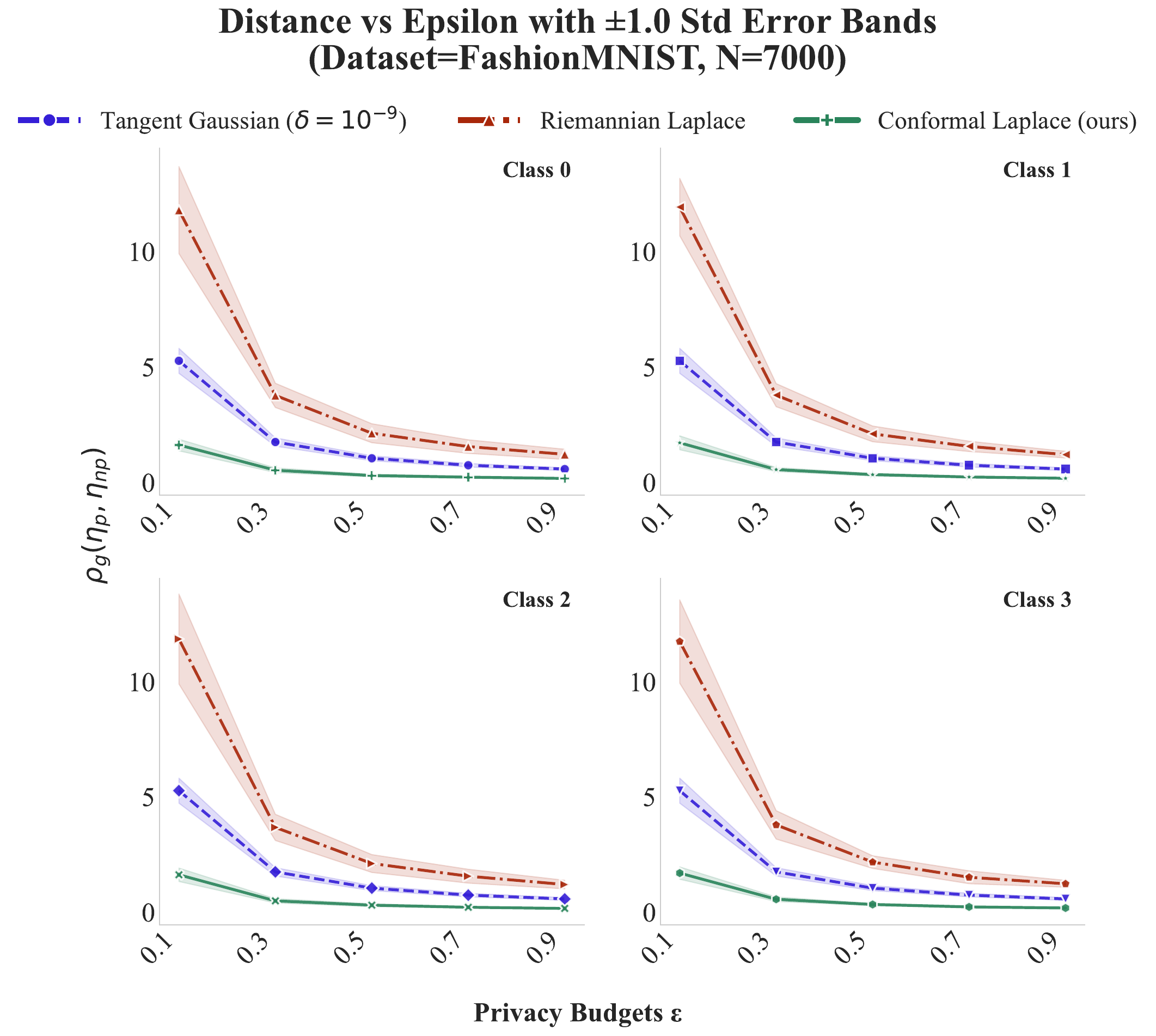}
        \label{fig:fashionmnist-6000}
    \end{subfigure}

    \caption{\textit{Utilities of private Fr\'echet means under varying privacy budgets ($\varepsilon$) and different classes for Fashion-MNIST dataset.}}
    \label{fig:Fashion-MNIST-comp}
\end{figure}

\paragraph{Datasets.} The SPD descriptors are not uniformly distributed over the manifold. Instead of being sampled directly from any uniform distribution on $\mathrm{SPD}(N)$, they are induced by pushing a structured real-image distribution through a deterministic image-to-SPD mapping. Consequently, the resulting manifold-valued samples concentrate around a data-dependent Fr\'echet mean, rather than covering $\mathrm{SPD}(N)$ in a homogeneous manner. This non-uniform geometric structure makes such data a natural testbed for our conformal mechanism, whose design is intended to adapt to heterogeneous manifold distributions. To further demonstrate the practical relevance and scalability of Conformal Laplace, we evaluate its utility on two real-world image datasets, \textbf{CIFAR-10} \cite{Krizhevsky09CIFAR} and \textbf{Fashion-MNIST} \cite{Xiao17FashionMNIST}. In this setting, our goal is to privately release the Fr\'echet mean of the dataset as a manifold-valued descriptor, namely $\eta:\mathcal{M}^N \rightarrow \mathcal{M}$.

\paragraph{Processing.} For image data, we first map them to a Riemannian manifold following Utpala et al. \cite{utpala2022differentially}. Specifically, for an image $\mathcal{I} \in \mathbb{R}^{h \times w \times c}$ represented as $h \times w$ pixels with $c$ channels, we convert it to their covariance descriptor in the form of a $k \times k$ SPD matrix given by:
\begin{equation}
    \mathrm{R}_\iota (\mathcal{I}) = \frac{1}{|S|} \sum_{\mathbf{x} \in S} {\left( \upsilon(\mathcal{I})(\mathbf{x}) - \bar{\upsilon} \right)}{\left( \upsilon(\mathcal{I})(\mathbf{x}) - \bar{\upsilon} \right)}^T + \iota I,  \nonumber
\end{equation}
\noindent in which $\upsilon(\mathcal{I})(\mathbf{x})$ is a pixel-level feature extractor for each pixel $\mathbf{x}=(x,y)$, $\bar{\upsilon}$ is the mean of $\upsilon(\mathcal{I})(\mathbf{x})$ for $\forall \mathbf{x} \in S$ and $\iota$ is a small constant. Following Utpala et al. \cite{utpala2022differentially}, we define $\upsilon(\mathcal{I})(\mathbf{x})$ as:
\begin{align*}
    \upsilon(\mathcal{I})(\mathbf{x}) =
    &\left[\vphantom{\frac{{|\mathcal{I}|}_x}{{|\mathcal{I}|}_y}} x, y, \mathcal{I}, |\mathcal{I}_x|, |\mathcal{I}_y|, |\mathcal{I}_{xx}|, |\mathcal{I}_{yy}|,\right. \\
    & \left. \sqrt{{|\mathcal{I}_x|}^2+{|\mathcal{I}_y|}^2}, \arctan \frac{{|\mathcal{I}|}_x}{{|\mathcal{I}|}_y} \right] .
\end{align*}
The radius of the public geodesic ball under this construction is $\sqrt{9} |\ln \iota|$ for \textbf{Fashion-MNIST} and $\sqrt{11} |\ln \iota|$ for \textbf{CIFAR-10}, which can be utilized to truncate the SPD manifold; for more detail, we refer to Theorem 5 of Utpala et al. \cite{utpala2022differentially}.

\paragraph{Cross-Dataset Evaluation.} We present similar performance comparison for the three mechanisms on CIFAR-10 and Fashion-MNIST. Note that we do not compare ours with the K-Norm gradient (KNG) method \mbox{\cite{soto2022shape}} on SPD matrices. As the authors of \mbox{\cite{soto2022shape}} note, the KNG mechanism does not have a closed-form expression for the acceptance probability in the MCMC sampling process and requires computing the inverse exponential map of the dataset at each step. This incurs a time complexity of $O(M \times N \times k^3)$ for $k \times k$ SPD matrices, making comparisons on CIFAR-10 and Fashion-MNIST computationally infeasible. We present results comparing the utilities of different DP mechanisms across both datasets. We perform tests on 4 classes; \cref{fig:cifar10-comp} and \cref{fig:Fashion-MNIST-comp} show the geodesic distances of private Fr\'echet means under varying $\varepsilon$ across four representative classes for each dataset.

On CIFAR-10, all classes exhibit a consistent trend: the distance decreases as $\varepsilon$ increases, and Conformal Laplace achieves the lowest error across all privacy levels, followed by Tangent Gaussian, while Riemannian Laplace incurs substantially larger distortion. The performance gap is most significant in the low-$\varepsilon$ regime and narrows as $\varepsilon$ increases.

Similar behavior is observed on Fashion-MNIST, albeit at a smaller error scale. Conformal Laplace remains the most accurate across all classes, with a stable performance advantage over Tangent Gaussian and a clear gap from Riemannian Laplace. The consistent ordering across datasets and classes indicates that the improvement of Conformal Laplace is robust and not tied to specific data characteristics.

\section{Conclusion}
\label{conclusion}
In this paper, we introduce the \emph{Conformal}-DP mechanism, a novel differential privacy approach designed specifically for data on Riemannian manifolds. Our approach adjusts privacy perturbations based on the local density of data points while not leaking information. We demonstrate through both theoretical analyses and experimental results that our DP mechanism provides a better trade-off between privacy and utility compared to the state-of-the-art methods, especially when applied to datasets with heterogeneous distributions.

An important future direction is to extend Conformal-DP to Gaussian Differential Privacy or R\'enyi Differential Privacy. A possible pathway is to first condition on the sanitized conformal metric produced in Stage 1 and then define a Gaussian-type release distribution under that fixed metric, for example, using a density based on the squared conformal geodesic distance with respect to the conformal volume measure. The resulting GDP or RDP guarantee would need to hold uniformly over all possible sanitized metrics produced by Stage 1 before it could be composed with the Stage-1 privacy guarantee.

This extension also presents several technical challenges. First, the current pure-DP proof uses the triangle inequality for the linear geodesic-distance term in the conformal Laplace density. A Gaussian-type mechanism would instead involve squared geodesic distance, so the same privacy-loss bound would not apply directly. Second, its normalization constant may depend on the center, the conformal volume measure, and the volume-growth properties of the transformed manifold. Third, a non-constant conformal transformation changes distances, volumes, and curvature. Therefore, curvature and volume conditions used by existing Gaussian or R\'enyi mechanisms on Riemannian manifolds would need to be established again under the learned metric $g^*=\phi g$. Finally, the Stage-2 GDP or RDP guarantee would need to be composed with the pure-DP geometry release in Stage 1. A complete treatment of these issues requires a new mechanism and separate theoretical
analysis, and is left for future work.

\bibliographystyle{IEEEtran}
\bibliography{ref}
\appendix
\subsection{Proof for Theorem \ref{theorem:h}}
\label{proof_theorem_h}
\begin{proof}[Sketch of proof]
We evaluate the bias and variance of the intrinsic KDE $\hat{f}(x)$ under the uniform reference distribution $f_{ref}(x) = 1/V$. $0 \le h \le \frac{1}{2}\mathrm{inj}(\mathcal{M})$, the exponential map $\exp_x : T_x \mathcal{M} \to \mathcal{M}$ is a global diffeomorphism on the kernel support. Pulling back the integration to the tangent space $T_x \mathcal{M}$ using normal coordinates $y = \exp_x(\mathrm{v})$, the Riemannian volume form expands as:
\begin{equation*}
d\mu_g(\exp_x(\mathrm{\mathrm{v}})) = \left( 1 - \frac{1}{6} \mathrm{Ric}_{ij}(x) \mathrm{v}^i \mathrm{v}^j + \mathcal{O}(\|\mathrm{v}\|^3) \right) d\mathrm{v}
\end{equation*}
The expectation of $\hat{f}(x)$ under $f_{ref}$ is:
\begin{equation*}
\begin{aligned}
    \mathbb{E}[\hat{f}(x)] &= \int_{\mathcal{M}} \frac{1}{h^d} K\left(\frac{\rho_g(x, y)^2}{h^2}\right) \frac{1}{V} d\mu_g(y) \\
    &= \frac{1}{V} \int_{T_x \mathcal{M}} \frac{1}{h^d} K\left(\frac{\|\mathrm{v}\|^2}{h^2}\right) \left( 1 - \frac{1}{6} \mathrm{Ric}_{ij}(x) \mathrm{v}^i \mathrm{v}^j \right) d\mathrm{v}
\end{aligned}
\end{equation*}
Applying the substitution $\mathrm{v} = h u$ with $dv = h^d du$:
\begin{equation*}
\mathbb{E}[\hat{f}(x)] \approx \frac{1}{V} \int_{\mathbb{R}^d} K(\|u\|^2) \left( 1 - \frac{h^2}{6} \mathrm{Ric}_{ij}(x) u^i u^j \right) du
\end{equation*}
By the spherical symmetry of $K$, cross-terms vanish, giving $\int_{\mathbb{R}^d} u^i u^j K(\|u\|^2) du = \mu_2(K) \delta^{ij}$. Contracting the Ricci tensor yields the scalar curvature $\mathrm{Ric}_{ij}(x) \delta^{ij} = S(x)$. Since $\int_{\mathbb{R}^d} K(\|u\|^2) du = 1$, we obtain:
\begin{equation*}
\mathbb{E}[\hat{f}(x)] \approx \frac{1}{V} \left( 1 - \frac{h^2 \mu_2(K)}{6} S(x) \right)
\end{equation*}
The point-wise bias is therefore:
\begin{equation*}
\mathrm{Bias}[\hat{f}(x)] = \mathbb{E}[\hat{f}(x)] - \frac{1}{V} \approx - \frac{h^2 \mu_2(K)}{6 V} S(x)
\end{equation*}
Next, the leading-order term for the variance is calculated as:
\begin{equation*}
\mathrm{Var}[\hat{f}(x)] \approx \frac{1}{N} \int_{\mathcal{M}} \frac{1}{h^{2d}} K\left(\frac{\rho_g(x, y)^2}{h^2}\right)^2 \frac{1}{V} d\mu_g(y) \approx \frac{R(K)}{N V h^d}
\end{equation*}
We construct the regularized Asymptotic Mean Integrated Squared Error ($\mathrm{AMISE}$) over $\mathcal{M}$ by introducing a regularizer $\lambda_0 > 0$ to the geometric penalty term:
{\footnotesize
\begin{equation*}
\begin{aligned}
    \mathrm{AMISE}(h) = \int_{\mathcal{M}} \left( \mathrm{Bias}[\hat{f}(x)]^2 + \mathrm{Var}[\hat{f}(x)] \right) d\mu_g(x) + \frac{h^4 \mu_2(K)^2 \lambda_0}{36 V^2}
\end{aligned}
\end{equation*}}
Substituting the squared bias and variance, and noting that $\int_{\mathcal{M}} S(x)^2 d\mu_g(x) = \|S\|_{L^2(\mathcal{M})}^2$ and $\int_{\mathcal{M}} d\mu_g(x) = V$:
\begin{equation*}
\mathrm{AMISE}(h) \approx \frac{h^4 \mu_2(K)^2 \left(\|S\|_{L^2(\mathcal{M})}^2 + \lambda_0\right)}{36 V^2} + \frac{R(K)}{N h^d}
\end{equation*}
Setting the derivative with respect to $h$ to zero:
\begin{equation*}
\frac{\partial \mathrm{AMISE}}{\partial h} = \frac{4 h^3 \mu_2(K)^2 \left(\|S\|_{L^2(\mathcal{M})}^2 + \lambda_0\right)}{36 V^2} - \frac{d \cdot R(K)}{N h^{d+1}} = 0
\end{equation*}
Solving for $h$ yields the unconstrained root:
\begin{equation*}
h = \left( \frac{9 d \cdot V^2 \cdot R(K)}{N \cdot \mu_2(K)^2 \cdot \left(\|S\|_{L^2(\mathcal{M})}^2 + \lambda_0\right)} \right)^{\frac{1}{d+4}}
\end{equation*}
Finally, taking into account the topological constraint required for the exponential map to remain a global diffeomorphism on the kernel support, the bandwidth is truncated by half the injectivity radius:
{\small
\begin{equation*}
    h = \min \left\{ \left( \frac{9 d \cdot V^2 \cdot R(K)}{N \cdot \mu_2(K)^2 \cdot \left(\|S\|_{L^2(\mathcal{M})}^2 + \lambda_0\right)} \right)^{\frac{1}{d+4}},\frac{1}{2} \operatorname{inj}(\mathcal{M}) \right\}
\end{equation*}}
This concludes the proof.
\end{proof}

\subsection{Proof for Lemma \ref{lemma1}}
\begin{proof}
\label{proof_lemma1}
    (i) Assume $w=\sigma_1-\sigma_2$. Then $(-\Delta_g+\upsilon)w=0$. Testing against $w$ gives 
    $\int_{\mathcal M}(|\nabla w|^2+\upsilon w^2)\,d\mu_g=0$, when $w\equiv0$.  
    (ii) $w:=\sigma_1-\sigma_2$ satisfies $(-\Delta_g+\upsilon)w=(c-f_1)-(c-f_2)=f_2-f_1\ge0$. By the strong maximum principle, $w\ge0$.
\end{proof}

\subsection{Proof for Proposition. \ref{prop:conformal_factor_properties}}
\label{proof_prop:conformal_factor_properties}
\begin{proof}
By definition, the right-hand side of \cref{eq:poisson} satisfies
\[
\tilde{f}_{\mathrm{data}}(x) - \tilde{f}_{\mathrm{max}} \le 0,
\qquad \forall x \in \mathcal{M}.
\]
Since the constant function $\sigma_+ \equiv 0$ satisfies
\[
(-\Delta_g + \upsilon)\sigma_+ = 0,
\]
the weak maximum principle for the strongly elliptic operator $-\Delta_g + \upsilon$ implies
\[
\sigma(x) \le 0,
\qquad \forall x \in \mathcal{M}.
\]

For the lower bound, observe that
\[
\tilde{f}_{\mathrm{data}}(x) - \tilde{f}_{\mathrm{max}}
\ge
\tilde{f}_{\mathrm{min}} - \tilde{f}_{\mathrm{max}},
\qquad \forall x \in \mathcal{M}.
\]
Consider the constant function
\[
\sigma_- \equiv \frac{\tilde{f}_{\mathrm{min}} - \tilde{f}_{\mathrm{max}}}{\upsilon},
\]
which satisfies
\[
(-\Delta_g + \upsilon)\sigma_-
=
\tilde{f}_{\mathrm{min}} - \tilde{f}_{\mathrm{max}}.
\]
Applying the comparison argument again yields
\[
\sigma_-
\le
\sigma(x),
\qquad \forall x \in \mathcal{M}.
\]
Combining the two inequalities proves \cref{eq:sigma_bounds_combined}. Since $s \mapsto e^{2s}$ is strictly increasing, exponentiating both sides gives \cref{eq:phi_bounds_combined}.
\end{proof}

\subsection{Proof for Theorem \ref{theorem1}}
\label{proof1}
\begin{proof}
    Note, we consider a typical second-order strongly elliptic operator $\Omega_g$ in local coordinates \(\{x_1,...,x_n\}\) on the Riemannian manifold $(\mathcal{M},g)$ which is in divergence form:
    \begin{equation*}
    \begin{aligned}
        \Omega_g[\sigma](x)=&\frac{1}{\sqrt{\operatorname{det}(g)}} \frac{\partial}{\partial x^i}\left(\sqrt{\operatorname{det}(g)} A^{i j}(x, \sigma(x)) \frac{\partial \sigma}{\partial x^j}(x)\right)+\\
        &B^i(x, \sigma(x)) \frac{\partial \sigma}{\partial x^i}(x)+C(x, \sigma(x)) \sigma(x) \\
        =& \frac{1}{\sqrt{\operatorname{det}(g)}} \frac{\partial}{\partial x^i}\left(\sqrt{\operatorname{det}(g)} g^{i j} \frac{\partial \sigma}{\partial x^j}\right) \\
        =& \Delta_g \sigma(x)
    \end{aligned}
    \end{equation*}
    where $A^{i j}(x, \mu)$ is the leading-order coefficient with the property of positive-definite; in our case, $A^{i j}(x, \mu) = g^{ij}$; \(B^i(x, \mu)\) and \(C(x, \mu)\) are lower-order terms, which can also depend on \(\mu\) and its derivatives if the operator is nonlinear, in our case, \(B^i = C = 0\). Since $\Omega_g$ is strongly elliptic and its coefficients (as well as $H$) are smooth, using Schauder estimates \cite{wang2006schauder} guarantees that any weak solution is in \(C^\infty(\mathcal{M})\), thus $\sigma$ is in fact smooth. In our application, $\Omega_g = \Delta_g$ (a special case of such elliptic operators), so a solution $\sigma$ to $\Delta_g \sigma = H$ exists and is $C^\infty$ by elliptic regularity \cite{gilbarg1977elliptic}. Note that in the theorem, we assume $\phi(x) \;=\; e^{2\sigma(x)}$, which ensures $\phi(x) > 0$ for every $x\in \mathcal{M}$. Since $\sigma \in C^\infty(\mathcal{M})$, it follows that $\phi\in C^\infty(\mathcal{M})$. Because $\phi$ is strictly positive, each tensor $g(x)$ remains positive-definite on the tangent space $T_x\mathcal{M}$, and the smoothness of $\phi$ and $g$ implies $g(x)$ is itself in $C^\infty(\mathcal{M})$. This confirms that $g^*$ is a smooth, positive Riemannian metric, thereby proving the theorem’s assertion that the conformal metric remains smooth and positive-definite.
\end{proof}

\subsection{Proofs for Lemma \ref{lemma_Bounds for the Conformal Geodesic Distance} and Corollary \ref{corollary1}}
\label{proofcorollary1}
\begin{proof}
    Since $\phi$ is smooth on compact $\mathcal{M}$, it attains a finite maximum and positive minimum $\phi$, $\phi_{min}$ and $\phi_{max} = 1$ from Lemma \ref{lemma_Bounds for the Conformal Geodesic Distance}, then it is possible to transfer the inequalities to geodesic distances. For lower bound  transfer, $L_{g^{*}}(\gamma) \geq \sqrt{\phi_{\min }}$ $L_g(\gamma)$ $\Rightarrow$ $\underset{\gamma}{\inf} L_{g^*}(\gamma)$ $\geq$ $\sqrt{\phi_{\min }} \underset{\gamma}{\inf}L_g(\gamma)$; thus we have:
    \[
        \rho_{g^*}(x, y) \geq \sqrt{\phi_{\min }} \rho_g(x, y).
    \]
    For upper bound transfer, $L_{g^{*}}(\gamma) \leq L_g(\gamma)$ $\Rightarrow$ $\underset{\gamma}{\inf} L_{g^*}(\gamma)$ $\leq$ $\underset{\gamma}{\inf}L_g(\gamma)$; thus we have:
    \[
        \rho_{g^*}(x, y) \leq \rho_g(x, y).
    \]
    In particular, distances cannot shrink or stretch by more than factors $\sqrt{\phi_{\min}}$ and $1$, respectively. Therefore, $\sqrt{\phi_{\min }} \rho_g(x, y) \leq \rho_{g^*}(x, y) \leq \rho_g(x, y)$; for any case, these estimates hold for the entire manifold and for all paths.
    This implies that the new distance $\rho_{g^*}$ is bi-Lipschitz equivalent to the original distance $\rho_g$. This implies that no pair of points gets closer by more than a factor $\sqrt{\phi_{\min}}$ nor farther by more than $1$ under the conformal transformation.
\end{proof}

\subsection{Derivation of Eq. \ref{eq:conformal_pdf}}
\label{proof_eq:conformal_pdf}
    Let $\eta \in  \mathcal{M}$ be the chosen center point (Fréchet Mean) of dataset $D$. We aim to construct a Laplace-type distribution that places exponential-decay noise around $\eta$ under the conformal metric $g^*$. Define the unnormalized Laplace-type kernel: $\widetilde{K}^*(z \mid \eta):=\exp \left[-\lambda^* \rho_{g^*}(\eta, z)\right], z \in \mathcal{M}$, where $\rho_{g^*}$ denotes the geodesic distance induced by $g^*$ and $\lambda^*>0$ is a rate parameter. Its normalizing constant over the manifold is:
    \begin{equation}
    \label{eq19}
        C\left(\eta, \lambda^*\right)=\int_{\mathcal{M}} \exp \left[-\lambda^* \rho_{g^*}(\eta, z)\right] d \mu_{g^*}(z) .
    \end{equation}
    Since $(\mathcal{M},g^*)$ is compact, $\rho_{g^*}(\eta,\cdot)$ is bounded and continuous, the integrand is continuous and strictly positive, and $\mu_{g^*}(\mathcal M)<\infty$. Therefore,
    \[
        0<C\left(\eta, \lambda^*\right)<+\infty .
    \]
    We normalize to obtain a probability density:
    {\small
    \begin{equation*}
        \mathbb{P}^*(z \mid \eta):=\frac{\widetilde{K}^*(z \mid \eta)}{C\left(\eta, \lambda^*\right)}
        =\frac{\exp \left[-\lambda^* \rho_{g^*}(\eta, z)\right]}{\int_{\mathcal{M}} \exp \left[-\lambda^*  \rho_{g^*}(\eta, u)\right] , d \mu_{g^*}(u)}, z \in \mathcal{M}
    \end{equation*}}
    We check that its integral over all of $\mathcal{M}$ is $1$:
    \begin{equation*}
    \begin{aligned}
        & \int_{\mathcal{M}} \mathbb{P}^*(z \mid \eta) d \mu_{g^*}(z) \\
        = & \int_{\mathcal{M}} \frac{\exp \left[-\lambda^* \rho_{g^*}(\eta, z)\right]}{\int_{\mathcal{M}} \exp \left[-\lambda^* \rho_{g^*}(\eta, u)\right] d \mu_{g^*}(u)} d \mu_{g^*}(z) \\
        = & \frac{\int_{\mathcal{M}} \exp \left[-\lambda^* \rho_{g^*}(\eta, z)\right] d \mu_{g^*}(z)}{\int_{\mathcal{M}} \exp \left[-\lambda^* \rho_{g^*}(\eta, u)\right] d \mu_{g^*}(u)} \\
        = & \frac{C^*\left(\eta, \lambda^*\right)}{C^*\left(\eta, \lambda^*\right)}=1 .
    \end{aligned}
    \end{equation*}
    Hence $\mathbb{P}^*(\cdot \mid \eta)$. is a well-defined probability density on $(\mathcal{M},\mu_{g^*})$. Equivalently, for any measurable set $\mathcal{A} \subset \mathcal{M}$,
    \[
        \mathbb{P}(Z \in \mathcal{A})=\int_{\mathcal{A}} \mathbb{P}^*(z \mid \eta) d \mu_{g^*}(z).
    \]
    This completes the derivation. \qed

\subsection{Proof for Theorem. \ref{thm:conformal_sensitivity}}
\label{proof_thm:conformal_sensitivity}

\begin{proof}
The proof has three steps: bounding the conformal radius $r^*$, bounding the conformal sectional curvature $\kappa^*$, and substituting both bounds into the Fr\'echet mean sensitivity expansion.

\textbf{Step 1: Conformal radius bound.}
Let $\gamma_x:[0,L]\to\mathcal{M}$ be a minimizing $g$-geodesic from $m_0$ to $x\in D$, with $L\le r$. Under the conformal metric $g^*=\phi g$, the distance is bounded by the conformal path length. Splitting the path into the portions lying in the dense core $\mathcal{C}$ and the sparse region $\mathcal{W}$ gives
\begin{align*}
\rho_{g^*}(m_0,x)
&\le
\int_{\gamma_x\cap\mathcal{C}} 1\,ds
+
\int_{\gamma_x\cap\mathcal{W}} \sqrt{\tau}\,ds \\
&\le
\theta L+\sqrt{\tau}(L-\theta L) \\
&=
L\bigl(\theta+\sqrt{\tau}(1-\theta)\bigr).
\end{align*}
Taking the supremum over $x\in D$ yields
\[
r^*\le C_{\tau,\theta}r.
\]

\textbf{Step 2: Conformal sectional curvature bound.}
For a conformal metric $g^*=e^{2\sigma}g$, the transformed sectional curvature depends on the original curvature, the Hessian of $\sigma$, and gradient terms.

For $d=2$, the curvature formula simplifies because the gradient terms cancel in an orthonormal basis, giving
\[
K^*
=
e^{-2\sigma}
\bigl(
K_g-\Delta_g\sigma
\bigr).
\]
Using the governing PDE,
\[
-\Delta_g\sigma
=
\tilde{f}_{\mathrm{data}}-\tilde{f}_{\max}-\upsilon\sigma,
\]
we obtain
\begin{equation*}
\begin{aligned}
    K^*(x) &= e^{-2\sigma(x)} \bigl(K_g(x)+\tilde{f}_{\mathrm{data}}(x)-\tilde{f}_{\max}-\upsilon\sigma(x) \bigr) \\
    &\le e^{-2\sigma(x)}(\kappa-\upsilon\sigma(x)),
\end{aligned}
\end{equation*}

since $K_g(x)\le \kappa$ and $\tilde{f}_{\mathrm{data}}(x)\le \tilde{f}_{\max}$.

Define
\[
H(\sigma)=e^{-2\sigma}(\kappa-\upsilon\sigma).
\]
By Proposition \ref{prop:conformal_factor_properties},
\[
\sigma\in[\sigma_{\min},0],
\qquad
\sigma_{\min}=\frac{\tilde{f}_{\min}-\tilde{f}_{\max}}{\upsilon}.
\]
Moreover,
\[
H'(\sigma)
=
-e^{-2\sigma}(2\kappa-2\upsilon\sigma+\upsilon)<0
\]
on this interval, so $H$ is decreasing and attains its maximum at $\sigma_{\min}$. Therefore,
\[
\kappa^*
\le
e^{-2\sigma_{\min}}(\kappa-\upsilon\sigma_{\min})
=
\frac{1}{\phi_{\min}}
\bigl(
\kappa+\tilde{f}_{\max}-\tilde{f}_{\min}
\bigr).
\]

For $d\ge 3$, the exact gradient cancellation is no longer available. However, since $\sigma$ is a smooth solution to a strongly elliptic PDE on a compact manifold, standard elliptic regularity yields deterministic bounds
\[
|\nabla \sigma|_g^2 \le C_1,
\qquad
\|\nabla^2\sigma\| \le C_2.
\]
Substituting these bounds into the general conformal curvature formula yields
\[
\kappa^*
\le
\frac{1}{\phi_{\min}}
\bigl(
\kappa+2C_2+C_1
\bigr).
\]

\textbf{Step 3: Sensitivity bound.}
Assume the conformal radius remains in the strongly convex regime, namely
\[
r^*<\frac{\pi}{2}(\kappa^*)^{-1/2}.
\]
Under this condition, the Fr\'echet mean sensitivity under the conformal metric $g^*$ follows the same form as in Reimherr et al. \cite{reimherr2021differential}:
\[
\Delta^*
=
\frac{2r^*(2-h(r^*,\kappa^*))}{N\,h(r^*,\kappa^*)},
\]
where, writing
\[
z=2r^*\sqrt{\kappa^*},
\qquad
h(r^*,\kappa^*)=z\cot z,
\]
The sensitivity can be expressed as
\[
\Delta^*
=
\frac{2r^*}{N}\,
\frac{2-z\cot z}{z\cot z}.
\]
We analyze this expression in the small-radius regime. Noting that $g(z) = \frac{2-z\cot z}{z\cot z} = 2\frac{\tan z}{z}-1$, its higher-order coefficients are positive. To prevent underestimating the sensitivity via asymptotic truncation, we apply Taylor's theorem \cite{folland2005higher} with the Lagrange remainder. Bounding the truncation error explicitly yields:
\[
\frac{2-z\cot z}{z\cot z} 
\le 
1+\frac{2}{3}z^2+\frac{\mathcal{O}}{4!}z^4,
\]
where $\mathcal{O} = \sup_{\xi \in (0, z_{\max})} g^{(4)}(\xi)$ explicitly bounds the fourth derivative over the valid domain $z_{\max} = 2C_{\tau,\theta}r\sqrt{\kappa^*}$. 
Substituting this explicit upper bound into the sensitivity formula yields
\begin{equation*}
\begin{aligned}
    \Delta^* &\le \frac{2r^*}{N} \left(1+\frac{2}{3}z^2+\frac{\mathcal{O}}{24}z^4\right) \\
    &= \frac{2r^*}{N} + \frac{16(r^*)^3\kappa^*}{3N} + \frac{4\mathcal{O}(r^*)^5(\kappa^*)^2}{3N}.
\end{aligned}
\end{equation*}
The leading term is linear in the contracted radius $r^*$, while the curvature enters through a higher-order cubic correction. Therefore, the first-order gain comes directly from the conformal distance contraction, whereas curvature affects only the smaller higher-order terms, with the explicit remainder ensuring absolute privacy protection. Finally, substituting $r^*\le C_{\tau,\theta}r$ together with the corresponding upper bound on $\kappa^*$ gives
\[
\Delta^*
\le
C_{\tau,\theta}\left(\frac{2r}{N}\right)
+
C_{\tau,\theta}^3
\left(\frac{16r^3}{3N}\right)\kappa^*
+
C_{\tau,\theta}^5
\left(\frac{4\mathcal{O}r^5}{3N}\right)(\kappa^*)^2,
\]
which proves the stated bound.
\end{proof}

\subsection{Proof for Theorem. \ref{thm:main_conformal_dp}}
\label{proof_thm:main_conformal_dp}
\begin{proof}
Let $D$ and $D'$ be adjacent datasets, and write
\[
x=\eta_{g^*}(D),
\qquad
x'=\eta_{g^*}(D').
\]
By construction, both conditional output distributions admit positive densities with respect to the same reference measure $\mu_{g^*}$. Under the integrability assumption,
\[
0<C(x,\lambda^*)
=
\int_{\mathcal{M}}
\exp\bigl\{-\lambda^* \rho_{g^*}(x,u)\bigr\}\,d\mu_{g^*}(u)
<\infty,
\]
and similarly for $C(x',\lambda^*)$. Hence, the Radon--Nikodym derivative is well defined, and the privacy loss at any output $z\in\mathcal{M}$ is
\begin{equation}
\label{eq:privacy_loss}
\begin{aligned}
    L_{D,D'}(z)
    &=
    \ln \frac{d\mathbb{P}^*(\cdot \mid x)}{d\mathbb{P}^*(\cdot \mid x')}(z)
    =
    \ln \frac{p^*(z \mid x)}{p^*(z \mid x')} \\
    &=
    \lambda^* \bigl(\rho_{g^*}(x',z)-\rho_{g^*}(x,z)\bigr)
    + \ln \frac{C(x',\lambda^*)}{C(x,\lambda^*)}.
\end{aligned}
\end{equation}

We bound the two terms separately. First, by the reverse triangle inequality for the metric $\rho_{g^*}$,
\begin{equation}
\label{eq:bound_distance}
    \rho_{g^*}(x',z)-\rho_{g^*}(x,z)
    \le
    \rho_{g^*}(x,x')
    \le
    \Delta^*.
\end{equation}

Second, using the triangle inequality
\[
\rho_{g^*}(x',u)\ge \rho_{g^*}(x,u)-\rho_{g^*}(x,x'),
\]
we obtain
\begin{align}
    C(x',\lambda^*)
    &=
    \int_{\mathcal{M}}
    \exp\left\{-\lambda^* \rho_{g^*}(x',u)\right\}\,d\mu_{g^*}(u) \nonumber\\
    &\le
    \int_{\mathcal{M}}
    \exp\left\{-\lambda^* \bigl(\rho_{g^*}(x,u)-\rho_{g^*}(x,x')\bigr)\right\}\,d\mu_{g^*}(u) \nonumber\\
    &=
    \exp\left\{\lambda^* \rho_{g^*}(x,x')\right\}
    \int_{\mathcal{M}}
    \exp\left\{-\lambda^* \rho_{g^*}(x,u)\right\}\,d\mu_{g^*}(u) \nonumber\\
    &=
    \exp\left\{\lambda^* \rho_{g^*}(x,x')\right\} C(x,\lambda^*).
\end{align}
Taking logarithms gives
\begin{equation}
\label{eq:bound_norm}
    \ln \frac{C(x',\lambda^*)}{C(x,\lambda^*)}
    \le
    \lambda^* \rho_{g^*}(x,x')
    \le
    \lambda^* \Delta^*.
\end{equation}

Substituting \cref{eq:bound_distance,eq:bound_norm} into \cref{eq:privacy_loss} yields
\[
L_{D,D'}(z)\le 2\lambda^*\Delta^*.
\]
By symmetry, exchanging $D$ and $D'$ gives
\[
L_{D,D'}(z)\ge -2\lambda^*\Delta^*.
\]
Therefore,
\[
|L_{D,D'}(z)|\le 2\lambda^*\Delta^*.
\]

Now substitute the calibration
\[
\lambda^*=\frac{\varepsilon_{\mathrm{conf}}}{2\mathcal{S}^*}
\]
from \cref{eq:lambda_calibration}. Since $\Delta^*\le \mathcal{S}^*$, we obtain
\[
|L_{D,D'}(z)|
\le
2\left(\frac{\varepsilon_{\mathrm{conf}}}{2\mathcal{S}^*}\right)\Delta^*
=
\varepsilon_{\mathrm{conf}}\frac{\Delta^*}{\mathcal{S}^*}
\le
\varepsilon_{\mathrm{conf}},
\qquad \forall z\in\mathcal{M}.
\]

Equivalently,
\[
p^*(z\mid \eta_{g^*}(D))
\le
e^{\varepsilon_{\mathrm{conf}}}
\,p^*(z\mid \eta_{g^*}(D'))
\qquad \forall z\in\mathcal{M}.
\]
Integrating both sides over any measurable set $S\subseteq\mathcal{M}$ with respect to $\mu_{g^*}$ gives
\begin{align}
    \mathbb{P}[\mathcal{A}_{\mathrm{conf}}(D)\in S]
    &=
    \int_S p^*\bigl(z\mid \eta_{g^*}(D)\bigr)\,d\mu_{g^*}(z) \nonumber\\
    &\le
    \int_S e^{\varepsilon_{\mathrm{conf}}}
    p^*\bigl(z\mid \eta_{g^*}(D')\bigr)\,d\mu_{g^*}(z) \nonumber\\
    &=
    e^{\varepsilon_{\mathrm{conf}}}
    \mathbb{P}[\mathcal{A}_{\mathrm{conf}}(D')\in S].
\end{align}
Hence, $\mathcal{A}_{\mathrm{conf}}$ satisfies $\varepsilon_{\mathrm{conf}}$-differential privacy.
\end{proof}

\subsection{Proof of Theorem \ref{thm:composition}}
\label{proof:thm:composition}

\begin{proof}
Let $D\sim D'$ be adjacent datasets. Let $P_D$ denote the law of the noisy count vector $\tilde{c}\sim \mathcal{A}_{\mathrm{anchor}}(D)$ on $\mathbb{R}^J$, and let $P_{D'}$ denote the corresponding law under $D'$. Since $\mathcal{A}_{\mathrm{anchor}}$ is $\epsilon_\phi$-differentially private, its output distributions satisfy the Radon--Nikodym bound
\begin{equation}
\label{eq:RN-int}
    \frac{dP_D}{dP_{D'}}(\tilde{c}) \le e^{\epsilon_\phi}
\end{equation}
almost everywhere with respect to $P_{D'}$.

Now fix a realized $\tilde{c}$. Conditional on $\tilde{c}$, the second-stage mechanism $\mathcal{A}_{\mathrm{conf}}(\cdot\mid \tilde{c})$ is defined under the conformal metric $g^*$ and reference measure $\mu_{g^*}$. By the calibration
\[
\lambda^* \le \frac{\varepsilon_{\mathrm{conf}}}{2\Delta^*},
\]
Theorem \ref{thm:main_conformal_dp} implies that $\mathcal{A}_{\mathrm{conf}}(\cdot\mid \tilde{c})$ is $\varepsilon_{\mathrm{conf}}$-differentially private. Therefore, for any measurable set $B\subseteq \mathcal{M}$,
\begin{equation}
\label{eq:conf-DP}
    \mathbb{P}^*\bigl(z\in B \mid \eta_{g^*}(D),\tilde{c}\bigr)
    \le
    e^{\varepsilon_{\mathrm{conf}}}
    \mathbb{P}^*\bigl(z\in B \mid \eta_{g^*}(D'),\tilde{c}\bigr).
\end{equation}

Let $S\subseteq \mathbb{R}^J\times \mathcal{M}$ be any measurable set, and define its section at $\tilde{c}$ by
\[
S_{\tilde{c}} = \{z\in \mathcal{M} : (\tilde{c},z)\in S\}.
\]
Then
\[
\mathbb{P}\bigl(\mathcal{A}_{\mathrm{total}}(D)\in S\bigr)
=
\int_{\mathbb{R}^J}
\mathbb{P}^*\bigl(z\in S_{\tilde{c}} \mid \eta_{g^*}(D),\tilde{c}\bigr)\,dP_D(\tilde{c}).
\]
Applying \cref{eq:conf-DP} pointwise inside the integral gives
\[
\mathbb{P}\bigl(\mathcal{A}_{\mathrm{total}}(D)\in S\bigr)
\le
e^{\varepsilon_{\mathrm{conf}}}
\int_{\mathbb{R}^J}
\mathbb{P}^*\bigl(z\in S_{\tilde{c}} \mid \eta_{g^*}(D'),\tilde{c}\bigr)\,dP_D(\tilde{c}).
\]
Next, using \cref{eq:RN-int} with the nonnegative measurable function
\[
g(\tilde{c}) =
\mathbb{P}^*\bigl(z\in S_{\tilde{c}} \mid \eta_{g^*}(D'),\tilde{c}\bigr),
\]
we obtain
\[
\frac{\int_{\mathbb{R}^J}
\mathbb{P}^*\bigl(z\in S_{\tilde{c}} \mid \eta_{g^*}(D'),\tilde{c}\bigr)\,dP_D(\tilde{c})}{\int_{\mathbb{R}^J}
\mathbb{P}^*\bigl(z\in S_{\tilde{c}} \mid \eta_{g^*}(D'),\tilde{c}\bigr)\,dP_{D'}(\tilde{c})} \le e^{\epsilon_\phi}
\]
Combining the two bounds yields
\[
\mathbb{P}\bigl(\mathcal{A}_{\mathrm{total}}(D)\in S\bigr)
\le
e^{\epsilon_\phi+\varepsilon_{\mathrm{conf}}}
\int_{\mathbb{R}^J}
\mathbb{P}^*\bigl(z\in S_{\tilde{c}} \mid \eta_{g^*}(D'),\tilde{c}\bigr)\,dP_{D'}(\tilde{c}).
\]
The remaining integral is exactly
\[
\mathbb{P}\bigl(\mathcal{A}_{\mathrm{total}}(D')\in S\bigr),
\]
so
\[
\mathbb{P}\bigl(\mathcal{A}_{\mathrm{total}}(D)\in S\bigr)
\le
e^{\epsilon_\phi+\varepsilon_{\mathrm{conf}}}
\mathbb{P}\bigl(\mathcal{A}_{\mathrm{total}}(D')\in S\bigr).
\]
Since $S$ was arbitrary, $\mathcal{A}_{\mathrm{total}}$ satisfies $(\epsilon_\phi+\varepsilon_{\mathrm{conf}})$-differential privacy.
\end{proof}

\subsection{Proof for Theorem \ref{thm:phi_error_bound}}
\label{proof_stage1_utility}
\begin{proof}
    The proof consists of four steps, adapted to bound the expected squared error.
    
    \textbf{Step 1: Maximum anchor noise squared.}
    The anchor discretization maps the dataset to a count vector $c\in\mathbb{N}^J$ with global substitution sensitivity $\Delta_1=2$. To achieve $\varepsilon_\phi$-differential privacy, the mechanism adds independent Laplace noise
    \[
    Y_j \sim \mathrm{Lap}(2/\varepsilon_\phi),
    \qquad j=1,\dots,J.
    \]
    For $J$ i.i.d. Laplace variables, the expected squared maximum absolute deviation can be bounded using the extreme value properties of sub-exponential distributions. It satisfies:
    \begin{equation}
    \label{eq:max_lap_bound_squared}
    \mathbb{E}\!\left[\max_{1\le j\le J}|Y_j|^2\right]
    =
    \mathcal{O}\!\left( \frac{1}{\varepsilon_\phi^2}(\ln J+1)^2 \right).
    \end{equation}
    
    \textbf{Step 2: Density reconstruction error.}
    The privatized density is reconstructed as
    \[
    \tilde{f}_{\mathrm{data}}(x)=\sum_{j=1}^J w_j \bar{K}_{h,j}(x),
    \]
    while the non-private density uses weights
    \[
    w_j^{np}=\frac{c_j}{N}.
    \]
    Assuming $N$ is large enough that the normalization step is dominated by the additive noise, the deviation in the normalized weights satisfies
    \[
    \max_j |w_j-w_j^{np}|
    \le
    \frac{1}{N}\max_{1\le j\le J}|Y_j|.
    \]
    Therefore, for every $x\in\mathcal{M}$,
    \begin{equation*}
    \begin{aligned}
        \left|\tilde{f}_{\mathrm{data}}(x)-f_{\mathrm{data}}(x) \right| &= \left| \sum_{j=1}^J (w_j - w_j^{np})\bar{K}_{h,j}(x) \right| \\
        & \le \frac{1}{N}\max_{1\le j\le J}|Y_j| \sum_{j=1}^J \bar{K}_{h,j}(x).
    \end{aligned}
    \end{equation*}
    
    Taking the supremum over $x$ yields
    \begin{equation}
    \label{eq:density_error_bound}
    \|\tilde{f}_{\mathrm{data}}-f_{\mathrm{data}}\|_{L^\infty(\mathcal{M})}
    \le
    \frac{C_K}{N}\max_{1\le j\le J}|Y_j|.
    \end{equation}
    
    \textbf{Step 3: Source-term perturbation.}
    The Helmholtz--Poisson equation uses the shifted source term
    \[
    \tilde{f}_{\mathrm{data}}-\tilde{f}_{\max}.
    \]
    By the triangle inequality,
    \begin{equation*}
    \begin{aligned}
        &\left\| (\tilde{f}_{\mathrm{data}}-\tilde{f}_{\max}) - (f_{\mathrm{data}}-f_{\max}) \right\|_{L^\infty}  \\
        &\le \|\tilde{f}_{\mathrm{data}}-f_{\mathrm{data}}\|_{L^\infty}
    +
    |\tilde{f}_{\max}-f_{\max}|.
    \end{aligned}
    \end{equation*}
    
    Since the essential supremum is 1-Lipschitz with respect to the $L^\infty$ norm,
    \[
    |\tilde{f}_{\max}-f_{\max}|
    \le
    \|\tilde{f}_{\mathrm{data}}-f_{\mathrm{data}}\|_{L^\infty}.
    \]
    Combining this with \cref{eq:density_error_bound} gives
    \begin{equation}
    \label{eq:source_error_bound}
    \left\|
    (\tilde{f}_{\mathrm{data}}-\tilde{f}_{\max})
    -
    (f_{\mathrm{data}}-f_{\max})
    \right\|_{L^\infty}
    \le
    \frac{2C_K}{N}\max_{1\le j\le J}|Y_j|.
    \end{equation}
    
    \textbf{Step 4: PDE stability and squared conformal-factor error.}
    Let $e_\sigma=\sigma_p-\sigma_{np}$. By the $L^\infty$ stability of the elliptic PDE,
    \[
    \|e_\sigma\|_{L^\infty(\mathcal{M})}
    \le
    \frac{1}{\upsilon}
    \left\|
    (\tilde{f}_{\mathrm{data}}-\tilde{f}_{\max})
    -
    (f_{\mathrm{data}}-f_{\max})
    \right\|_{L^\infty}.
    \]
    Using \cref{eq:source_error_bound}, we obtain
    \begin{equation}
    \label{eq:sigma_error_bound}
    \|e_\sigma\|_{L^\infty(\mathcal{M})}
    \le
    \frac{2C_K}{\upsilon N}\max_{1\le j\le J}|Y_j|.
    \end{equation}
    
    Now recall that both scaling functions are non-positive, so $\phi(x)=e^{2\sigma(x)}$ lies in the non-expansive regime of the exponential map, where it is Lipschitz with constant $2$. Hence
    \[
    \|\phi_p-\phi_{np}\|_{L^\infty(\mathcal{M})}
    \le
    2\|e_\sigma\|_{L^\infty(\mathcal{M})}
    \le
    \frac{4C_K}{\upsilon N}\max_{1\le j\le J}|Y_j|.
    \]
    Squaring both sides and taking expectations, we apply \cref{eq:max_lap_bound_squared} to obtain
    \begin{equation*}
    \begin{aligned}
        \mathbb{E}\!\left[\|\phi_p-\phi_{np}\|^2_{L^\infty(\mathcal{M})}\right]
        & \le
        \frac{16C_K^2}{\upsilon^2 N^2}
        \mathbb{E}\!\left[\max_{1\le j\le J}|Y_j|^2\right] \\
        &=
        \mathcal{O}\!\left( \frac{C_K^2}{\upsilon^2 N^2\varepsilon_\phi^2}(\ln J+1)^2 \right),
    \end{aligned}   
    \end{equation*}  
    which completes the proof.
\end{proof}

\subsection{Proof for Theorem \ref{thm:stage2_utility}}
\label{proof_stage2_utility}
\begin{proof}
    The proof separates the total expected squared error into a squared geometric bias term and a stochastic sampling variance term using the fundamental inequality $(a+b)^2 \le 2a^2 + 2b^2$.
    
    \textbf{Step 1: Error decomposition and squared conformal bias.}
    By the triangle inequality under the original metric $g$,
    \[
    \rho_g\bigl(z,\eta_{np}(D)\bigr)
    \le
    \rho_g\bigl(\eta_{g^*}(D),\eta_{np}(D)\bigr)
    +
    \rho_g\bigl(z,\eta_{g^*}(D)\bigr).
    \]
    Squaring both sides and applying $(a+b)^2 \le 2a^2 + 2b^2$, we decouple the bounds:
    \begin{equation}
    \label{eq:squared_triangle_ineq}
    \rho_g\bigl(z,\eta_{np}(D)\bigr)^2 
    \le 
    2 \underbrace{\rho_g\bigl(\eta_{g^*}(D),\eta_{np}(D)\bigr)^2}_{\mathcal{B}_{\mathrm{conf}}^2} 
    + 
    2 \rho_g\bigl(z,\eta_{g^*}(D)\bigr)^2.
    \end{equation}
    Here, $\mathcal{B}_{\mathrm{conf}}$ is the deterministic bias induced by replacing the original metric $g$ with the conformal metric $g^*$. By the stability of Fr\'echet minimizers on manifolds with bounded curvature, the displacement of the mean is Lipschitz continuous with respect to metric deformations. This stability is strictly governed by the strong convexity parameter $h(r,\kappa)$, yielding
    \[
    \mathcal{B}_{\mathrm{conf}}
    \le
    \frac{1}{h(r,\kappa)}\,
    \|\phi_p-\phi_{np}\|_{L^\infty(\mathcal{M})}.
    \]
    Squaring this inequality, taking expectations, and applying the Stage-1 expected squared $L^\infty$ bound established in Theorem \ref{thm:phi_error_bound} yields
    \begin{equation}
    \begin{aligned}
    \label{eq:bias_bound_stage2_squared}
        \mathbb{E}[\mathcal{B}_{\mathrm{conf}}^2]
        &\le
        \frac{1}{h(r,\kappa)^2}\,
        \mathbb{E}\!\left[\|\phi_p-\phi_{np}\|^2_{L^\infty(\mathcal{M})}\right] \\
        &=
        \mathcal{O}\!\left( \frac{C_{\mathrm{stab}}(r,\kappa)^2 C_K^2}{\upsilon^2 N^2\varepsilon_\phi^2}(\ln J+1)^2 \right),
    \end{aligned}
    \end{equation}
    where $C_{\mathrm{stab}}(r,\kappa) = 1/h(r,\kappa)$.
    
    \textbf{Step 2: Expected squared sampling error under the original metric.}
    It remains to bound the stochastic sampling variance term
    $\mathbb{E}_{\mathrm{var}} = \mathbb{E}\!\left[\rho_g\bigl(z,\eta_{g^*}(D)\bigr)^2\right]$.
    By definition of the conformal mechanism, the exact integral over the conformal probability measure is
    \[
    \mathbb{E}_{\mathrm{var}}
    =
    \int_{\mathcal{M}}
    \rho_g\bigl(z,\eta_{g^*}(D)\bigr)^2\,
    \frac{\exp\{-\lambda^* \rho_{g^*}(\eta_{g^*}(D),z)\}}
         {C(\eta_{g^*}(D),\lambda^*)}
    \,d\mu_{g^*}(z).
    \]
    Since $\phi(x) \ge \phi_{\min} > 0$, we have the squared distance comparison $\rho_g\bigl(z,\eta_{g^*}(D)\bigr)^2 \le \frac{1}{\phi_{\min}}\, \rho_{g^*}\bigl(z,\eta_{g^*}(D)\bigr)^2$. This yields the reduction
    \begin{equation}
    \label{eq:sampling_reduction_squared}
    \begin{aligned}
        & \mathbb{E}_{\mathrm{var}} \le \\
        & \frac{1}{\phi_{\min}}
        \int_{\mathcal{M}}
        \rho_{g^*}\bigl(z,\eta_{g^*}(D)\bigr)^2\,
        \frac{\exp\{-\lambda^* \rho_{g^*}(\eta_{g^*}(D),z)\}}
         {C(\eta_{g^*}(D),\lambda^*)}
        \,d\mu_{g^*}(z).
    \end{aligned}
    \end{equation}
    
    \textbf{Step 3: Second-order radial moment bound via volume growth.}
    Let $r' = \rho_{g^*}\bigl(z,\eta_{g^*}(D)\bigr)$. To evaluate this integral on a $d$-dimensional manifold, we approximate the local volume growth of geodesic balls for small radius. The volume element grows at the leading order as $\mathrm{area}_{g^*}(r') \approx \omega_{d-1} (r')^{d-1}$, where $\omega_{d-1}$ is the area of the $(d-1)$-sphere. 
    The numerator of the radial expectation expands to the second moment:
    {\small
    \begin{equation}
    \begin{aligned}
        \int_0^\infty (r')^2\,e^{-\lambda^* r'} \mathrm{area}_{g^*}(r')\,dr' & \approx
    \int_0^\infty (r')^2\,e^{-\lambda^* r'} \omega_{d-1} (r')^{d-1}\,dr' \\
    & =
    \omega_{d-1}\frac{\Gamma(d+2)}{(\lambda^*)^{d+2}}.
    \end{aligned}
    \end{equation}}
    Similarly, the normalization constant evaluates as $\omega_{d-1}\frac{\Gamma(d)}{(\lambda^*)^d}$.
    Taking the ratio of the evaluated numerator and denominator, the geometric constant $\omega_{d-1}$ gracefully cancels out. Using the properties of the Gamma function $\Gamma(d+2) = d(d+1)\Gamma(d)$, we naturally obtain:
    \begin{equation}
    \label{eq:radial_moment_bound_squared}
    \frac{\omega_{d-1} \Gamma(d+2) / (\lambda^*)^{d+2}}{\omega_{d-1} \Gamma(d) / (\lambda^*)^d}
    =
    \frac{d(d+1)}{(\lambda^*)^2} = \mathcal{O}\!\left(\frac{d^2}{(\lambda^*)^2}\right).
    \end{equation}
    Combining Eq. \ref{eq:sampling_reduction_squared} and Eq. \ref{eq:radial_moment_bound_squared} gives the variance bound in the conformal space:
    \begin{equation}
    \label{eq:sampling_bound_stage2_squared}
    \mathbb{E}_{\mathrm{var}}
    \le
    \mathcal{O}\!\left( \frac{d^2}{\phi_{\min} (\lambda^*)^2} \right).
    \end{equation}
    
    \textbf{Step 4: Privacy calibration and instance-dependent sensitivity substitution.}
    To rigorously satisfy pure $\varepsilon_{\mathrm{conf}}$-differential privacy, the conformal mechanism calibrates the rate parameter according to the true conformal sensitivity:
    \[
    \lambda^*=\frac{\varepsilon_{\mathrm{conf}}}{2\Delta^*}.
    \]
    Substituting this calibration into Eq. \ref{eq:sampling_bound_stage2_squared} yields the sampling variance conditioned on the conformal metric:
    \begin{equation*}
    \mathbb{E}_{z}\!\left[\rho_g\bigl(z,\eta_{g^*}(D)\bigr)^2 \mid g^* \right]
    =
    \mathcal{O}\!\left( \frac{d^2 (\Delta^*)^2}{\phi_{\min} \varepsilon_{\mathrm{conf}}^2} \right).
    \end{equation*}
    Because the scaling function is defined as a non-positive shift ($\sigma(x) \le 0$), the conformal factor satisfies $\phi(x) = e^{2\sigma(x)} \le 1$ globally on $\mathcal{M}$. Therefore, the conformal metric $g^* = \phi \cdot g$ guarantees that all pairwise distances are non-increasing: $\rho_{g^*}(x,y) \le \rho_g(x,y)$. Let $r^*$ denote the dataset's effective bounding radius under $g^*$. The true conformal sensitivity is bounded by this local geometry, meaning $\Delta^* \le \mathcal{O}(r^*/N)$. Substituting this bound into the conditional variance term and taking the expectation over the Stage-1 mechanism yields the tighter instance-dependent bound. Furthermore, the non-expansive property ensures $r^* \le r$ almost surely, yielding the worst-case global baseline:
    \begin{align}
    \label{eq:variance_rate}
    \mathbb{E}\!\left[\rho_g\bigl(z,\eta_{g^*}(D)\bigr)^2\right] &\le \mathcal{O}\!\left( \frac{1}{\phi_{\min}}\frac{d^2 \mathbb{E}[(r^*)^2]}{N^2 \varepsilon_{\mathrm{conf}}^2} \right) \\  & \le \mathcal{O}\!\left( \frac{1}{\phi_{\min}}\frac{d^2 r^2}{N^2 \varepsilon_{\mathrm{conf}}^2} \right). \nonumber
    \end{align}
    
    Finally, combining the squared geometric bias Eq. \ref{eq:bias_bound_stage2_squared} and the sampling variance Eq. \ref{eq:variance_rate} completes the proof.
    {\footnotesize
    \begin{equation*}
    \begin{aligned}
        \mathbb{E}\!\left[\rho_g\bigl(z,\eta_{np}(D)\bigr)^2\right] &\le 2\,\mathbb{E}[\mathcal{B}^2_{\mathrm{conf}}] + 2\,\mathbb{E}_{\mathrm{var}} \\
        &\le \mathcal{O}\!\left( \frac{C_{\mathrm{stab}}(r,\kappa)^2 C_K^2 (\ln J+1)^2}{\upsilon^2 N^2 \varepsilon_\phi^2} + \frac{1}{\phi_{\min}}\frac{d^2 r^2}{N^2 \varepsilon_{\mathrm{conf}}^2} \right).
    \end{aligned}
    \end{equation*}}
\end{proof}

\subsection{Proof of Proposition \ref{prop:optimal_allocation}}
\label{proof_optimal_budgets}

\begin{proof}
We consider the constrained optimization problem
\[
\min_{\varepsilon_\phi,\varepsilon_{\mathrm{conf}}>0}
\left(
\frac{C_\phi}{\varepsilon_\phi}
+
\frac{C_{\mathrm{conf}}}{\varepsilon_{\mathrm{conf}}}
\right)
\qquad
\text{subject to }
\varepsilon_\phi+\varepsilon_{\mathrm{conf}}=\varepsilon_{\mathrm{total}}.
\]
Applying the Cauchy--Schwarz inequality yields
\[
(\varepsilon_\phi+\varepsilon_{\mathrm{conf}})
\left(
\frac{C_\phi}{\varepsilon_\phi}
+
\frac{C_{\mathrm{conf}}}{\varepsilon_{\mathrm{conf}}}
\right)
\ge
\left(\sqrt{C_\phi}+\sqrt{C_{\mathrm{conf}}}\right)^2.
\]
Equality holds if and only if
\[
\frac{\varepsilon_\phi}{\varepsilon_{\mathrm{conf}}}
=
\sqrt{\frac{C_\phi}{C_{\mathrm{conf}}}}.
\]
Therefore, the optimal allocation ratio is
\[
\frac{\varepsilon_\phi^{opti}}{\varepsilon_{\mathrm{conf}}^{opti}}
=
\sqrt{\frac{C_\phi}{C_{\mathrm{conf}}}}.
\]
Substituting the definitions
\[
C_\phi
=
\frac{8 C_{\mathrm{stab}}(r,\kappa) C_K (\ln J+1)}{\upsilon N},
C_{\mathrm{conf}}
=
\frac{2d\Delta^*}{\sqrt{\phi_{\min}}},
\]
gives
\begin{equation*}
\begin{aligned}
\frac{\varepsilon_\phi^{opti}}{\varepsilon_{\mathrm{conf}}^{opti}}
&=
\sqrt{
\frac{
\frac{8 C_{\mathrm{stab}}(r,\kappa) C_K (\ln J+1)}{\upsilon N}
}{
\frac{2d\Delta^*}{\sqrt{\phi_{\min}}}
}
} \\
&=
\sqrt{
\frac{
4 C_{\mathrm{stab}}(r,\kappa) C_K (\ln J+1)\sqrt{\phi_{\min}}
}{
\upsilon N d \Delta^*
}
}
\end{aligned}
\end{equation*}
which is exactly \cref{eq:optimal_ratio}.

Combining this ratio with the budget constraint
\[
\varepsilon_\phi+\varepsilon_{\mathrm{conf}}=\varepsilon_{\mathrm{total}}
\]
yields the closed-form allocations in \cref{eq:optimal_budgets}.
\end{proof}

\subsection{Implementation Details for Section. \ref{ssec:real_world}}
\label{app:alg-detail}
We provide further implementation details and hyperparameters of our proposed Conformal Differential Privacy and Riemannian-Laplace DP by Reimherr er al. \cite{reimherr2021differential} and tangent Gaussian DP by Utpala et al. \cite{utpala2022differentially} as follows:

\paragraph{Hyperparameters.} We test for algorithm utility under privacy budgets $\varepsilon \in \{0.1, 0.2, \cdots, 2.0\}$ and $\delta = 10^{-9}$. CIFAR-10 \cite{Krizhevsky09CIFAR} contains a total of $N = 60,000$ images in 10 classes, with 6,000 images for each class; Fashion-MNIST \cite{Xiao17FashionMNIST} has $N = 70,000$ images in total, with 7,000 images for each class. Our selected $\delta = 10^{-9}$ satisfies $\delta \ll \frac{1}{N}$ for both datasets. The \emph{global} sensitivity $\Delta$ is calculated as $\Delta = \frac{2r}{N}$, in which $r$ is the radius of the geodesic ball containing all images, $N$ is the number of data samples in the dataset. For more calculations in determining $r$, we refer to the work of Utpala et al. \cite{utpala2022differentially}.

\paragraph{Implementation.} We implement both Riemannian-Laplace DP \cite{reimherr2021differential} and tangent Gaussian DP \cite{utpala2022differentially} with the same parameters described in Utpala et al. \cite{utpala2022differentially}, specifically, for the global sensitivity of $\Delta_{glb} = \frac{2r}{N}$, in which $r$ is the radius of the geodistic ball that contains all data samples. We adopt the \emph{classical} Gaussian noise described by Utpala et al. \cite{utpala2022differentially}, in which we have $\sigma = \frac{\Delta_{glb}}{\varepsilon} \sqrt{2\ln (\frac{1.25}{\delta})}$.
For Riemannian-Laplace DP and our proposed Conformal-DP, we perform 5,000 steps for burn-in and 5,0000 steps of proposal sampling for the MCMC process following Utpala et al. \cite{utpala2022differentially} and decided empirically, which demonstrate reasonable performance.
For the K-Norm DP method discussed in Sec. \ref{sec:synthetic}, due to the high computational cost, we perform 5,000 steps for burn-in and 2,500 steps of sampling, the other implementations are consistent with the original paper of Soto et al. \cite{soto2022shape}
The algorithms and calculations over the manifold are implemented with the \texttt{geomstats} Python library by Miolane et al. \cite{Miolane2020Geomstats}. The calculation of the Frech\'et mean on the manifold through gradient descent can be easily implemented with \texttt{geomstats}. We thank the authors for the invaluable help in making the code publicly available.

\end{document}